\documentclass[11pt,reqno]{amsart}
\usepackage{dsfont, amssymb,amsmath,amscd,latexsym, amsthm, amsxtra,amsfonts,enumerate}
\allowdisplaybreaks[3]
\usepackage{caption,subcaption}
\usepackage[all]{xy}
\usepackage[active]{srcltx}
\usepackage{mathrsfs}
\usepackage[all]{xy}
\usepackage[active]{srcltx}
\usepackage{graphicx}
\usepackage{bm}
\usepackage{bbm}
\usepackage{graphicx}
\usepackage{caption}
\usepackage{color}
\usepackage{algorithmicx,algorithm}
\usepackage{longtable}
\usepackage[round]{natbib}
\usepackage{ulem}
\usepackage{verbatim}
\usepackage[pdfstartview=FitH, bookmarksnumbered=true,bookmarksopen=true, colorlinks=true, pdfborder=001, citecolor=blue, linkcolor=blue,urlcolor=blue]{hyperref}
\usepackage{CJK}
\newtheorem{theorem}{Theorem}[section]
\newtheorem{lemma}[theorem]{Lemma}

\newtheorem{proposition}[theorem]{Proposition}

\renewcommand{\theequation}{\arabic{section}.\arabic{equation}}
\newtheorem{Def}{Definition}[section]

\newtheorem{Assumption}{Assumption}[section]

\newcommand{\dif}{\mathrm{d}}

\renewcommand{\theequation}{\arabic{section}.\arabic {equation}}

\begin{document}
\makeatletter
    \def\@setauthors{%
        \begingroup
        \def\thanks{\protect\thanks@warning}%
        \trivlist \centering\footnotesize \@topsep30\p@\relax
        \advance\@topsep by -\baselineskip
        \item\relax
        \author@andify\authors
        \def\\{\protect\linebreak}%
        {\authors}%
        \ifx\@empty\contribs \else ,\penalty-3 \space \@setcontribs
        \@closetoccontribs \fi
        \endtrivlist
        \endgroup } \makeatother
\baselineskip 17pt
 \title[{{\tiny A Stackelberg reinsurance-investment game}}]
    {{\tiny A Stackelberg reinsurance-investment game under $\alpha$-maxmin mean-variance criterion and stochastic volatility}}
    \vskip 10pt\noindent
    \author[{\tiny Guohui Guan, Zongxia Liang and Yilun Song}]
    {\tiny {\tiny  Guohui Guan $^{a,\dag}$, Zongxia Liang$^{b,\ddag}$, Yilun Song$^{b,*}$ }
        \vskip 10pt\noindent
        {\tiny ${}^a$ School of Statistics, Renmin University of China, Beijing
            100872, China \vskip 10pt\noindent\tiny ${}^b$Department of
            Mathematical Sciences, Tsinghua University, Beijing 100084, China }
        \footnote{\\
            $ \dag$ email: guangh@ruc.edu.cn\\
            $ \ddag$ email: liangzongxia@mail.tsinghua.edu.cn \\
            $*$ Corresponding author, email: songyl18@mails.tsinghua.edu.cn    }}
    \maketitle
\begin{abstract}
This paper investigates a Stackelberg game between an insurer and a reinsurer under the $\alpha$-maxmin mean-variance criterion. The insurer can purchase per-loss reinsurance from the reinsurer. With the insurer's feedback reinsurance strategy, the reinsurer optimizes the reinsurance premium in the Stackelberg game.  The financial market consists of cash and stock with Heston's stochastic volatility. Both the insurer and reinsurer maximize their respective $\alpha$-maxmin mean-variance preferences in the market. The criterion is time-inconsistent and we derive the equilibrium strategies by the extended Hamilton-Jacobi-Bellman equations. Similar to the non-robust case in \cite{li2022stackelberg}, excess-of-loss reinsurance is the optimal form of reinsurance strategy for the insurer. The equilibrium investment strategy is determined by a system of Riccati differential equations. Besides, the equations determining the equilibrium reinsurance strategy and reinsurance premium rate are given semi-explicitly, which is simplified to an algebraic equation in a specific example. Numerical examples illustrate that the game between the insurer and reinsurer makes the insurance more radical when the agents become more ambiguity aversion or risk aversion. Furthermore, the level of ambiguity, ambiguity attitude, and risk attitude of the insurer (reinsurer) have similar effects on the equilibrium reinsurance strategy, reinsurance premium, and investment strategy.
%
%
\end{abstract}
\maketitle
\section{Introduction}
The optimal reinsurance-investment problem has long been a popular topic in actuarial sciences. Insurance risk is one of the main risks in the insurance market. For the insurer, reinsurance is an efficient way to control the insurance risk. In recent decades, researchers study different types of reinsurance for the insurer, such as proportional reinsurance (see \cite{schmidli2001optimal}, \cite{liang2007optimal}, \cite{yuen2015optimal}, etc.), the excess-of-loss reinsurance (see \cite{asmussen2000optimal}, \cite{li2018robust}, \cite{zhang2021optimal}, etc.), quota-share reinsurance (see \cite{centeno1991combining}, \cite{zhang2018optimal}, etc.). The abundant reinsurance contracts in the insurance market help the insurer manage the insurance risk well.  Besides, the insurer also suffers from financial risk in the financial market. Investment is an efficient tool for the insurer in the financial market, which has been studied in \cite{browne1995optimal}, \cite{irgens2004optimal}, \cite{guan2014optimal}, etc.

However, most of the previous studies derive the optimal reinsurance strategy from the perspective of the insurer. \cite{borch1969optimal} shows that the two parties have conflicting interests over the reinsurance premium. Reinsurance is a mutual agreement between the insurer and the reinsurer. The reinsurance strategy only concerning one party's interest may not be acceptable to another party. Generally, there are two ways to study the reinsurance arrangement from both the insurer's and the reinsurer's point of view. One way is to study a weighted goal of the insurer and reinsurer, which has been comprehensively investigated in \cite{cai2016optimal}, \cite{zhang2018optimal}, \cite{yang2022optimal}, etc. Another way is to formulate a Stackelberg game between the insurer and reinsurer. In the reinsurance contract, the insurer decides the number of losses ceded to the reinsurer, while the reinsurer sets the list of the reinsurance premium corresponding to the amount of reinsurance. In \cite{chen2018new}, a stochastic Stackelberg differential game between the insurer and reinsurer in a  continuous-time framework is studied. In the Stackelberg game, the reinsurer moves first and the insurer does subsequently to achieve a Stackelberg equilibrium toward optimal reinsurance arrangement.

As the Stackelberg game balances the interests of the insurer and reinsurer, recently, many scholars study the reinsurance game from different aspects. \cite{CS2019} investigate the  Stackelberg differential reinsurance games under the time-inconsistent mean–variance framework. \cite{bai2021stackelberg} consider the insurer and reinsurer's wealth processes with delay to characterize bounded memory. In \cite{bai2022hybrid},  a hybrid stochastic differential reinsurance and investment game between one reinsurer and two insurers is studied. \cite{zheng2021stackelberg} suppose that the follower cannot observe the state process directly, but could observe a noisy observation process, while the leader can completely observe the state process. \cite{li2022stackelberg} consider a reinsurance problem for a mean-variance Stackelberg game with a random time horizon. However, except for risk in the market,  ambiguity is another main concern for the agents, see \cite{ellsberg1961risk}. The agents cannot estimate the economic parameters accurately and are in fact ambiguity aversion.  For the ambiguity attitude, most of the existing works apply maxmin expected utility preference proposed by \cite{GS1989}. \cite{maenhout2004robust} studies the robust portfolio choices under the max-min expected utility. The ambiguity over the market is characterized by a set of equivalent probability measures. \cite{maenhout2004robust} indicates that the agent will consume and invest less with a higher ambiguity aversion attitude. From the perspective of the insurer, \cite{YLVZ2013}, \cite{yi2015robust}, \cite{bauerle2021robust} show that an ambiguity aversion attitude will decrease the insurance risk undertaken by the insurer. Besides, when ignoring ambiguity, the insurer may be suffered from large utility losses. In the insurance market, the insurer and reinsurer are also ambiguity averse. Although many works study the effect of insurance risk in the Stackelberg game, it remains to show whether model uncertainty has a big influence on the game.  As such, it is also necessary to introduce the ambiguity attitude in the Stackelberg game between the insurer and reinsurer.

Recently, there are some researches on the robust Stackelberg game. \cite{gu2020optimal} study the robust excess-of-loss reinsurance contract for the insurer and reinsurer with exponential utility. The numerical results show that the insurer will pay little attention to the reinsurer's ambiguity aversion level while the reinsurer will pay rather close attention to the insurer's ambiguity aversion. \cite{wang2020robust} simultaneously suppose that the insurer is subject to a dynamic Value-at-Risk constraint in the robust game. In \cite{yuan2022robust},  a loss-dependent premium principle is applied to the insurer. They also derive the robust reinsurance contract explicitly by solving the coupled extended Hamilton–Jacobi–Bellman systems under the time-consistent mean-variance criterion. They find that the reinsurer would like to raise the reinsurance price to guard against the model uncertainty. Other related works can refer to \cite{hu2022robust}, \cite{yang2022robust}, etc.

Although there are various works studying the robust Stackelberg game, most of the works
apply the maxmin exponential/mean-variance preferences. In the framework,  the agents only concentrate on the worst case of the market and are extremely ambiguity-averse.  However, \cite{HT1991} show that the agents' ambiguity attitude is not systematically negative and sometimes even ambiguity-seeking in their empirical studies. \cite{M2002} and \cite{GMM2004} introduce and axiomatically characterize the $\alpha$-maxmin expected utility to mix the ambiguity-averse case and the ambiguity-seeking case. Later, the $\alpha$-maxmin mean-variance criterion is studied in \cite{LLX2016} for an insurer. The insurer may have different levels of ambiguity and different ambiguity attitudes. Inspired by the above-mentioned works, we have three concerns in this paper. First, most literature considers the robust Stackelberg game in the maxmin robust model. We aim to study the Stackelberg game between the insurer and reinsurer under the $\alpha$-maxmin mean-variance criterion. We compare and discuss the effects of the level of ambiguity,  ambiguity attitude, and risk attitude on the reinsurance strategy and reinsurance price. The related robust problem is time-inconsistent and cannot be solved by the traditional dynamic programming method. Second, most of the previous works consider a specific reinsurance contract. The insurer can neither buy proportional reinsurance nor excess-of-loss reinsurance, such as  \cite{gu2020optimal} and \cite{yuan2022robust}. However, except for the amount of the reinsurance contract, the type of the reinsurance contract is also a mutual agreement between the insurer and reinsurer. In this paper, the insurer is allowed to purchase the per-loss reinsurance strategy, while the reinsurer optimizes the reinsurance premium. Third, although \cite{LLX2016} study the equilibrium investment strategy under the $\alpha$-maxmin mean-variance criterion, the stock price is supposed to follow the geometric Brownian motion in their work.  However,  empirical results support the existence of stochastic volatility such as \cite{FSS1987} and \cite{PS1990}, etc. The volatility of the stock fluctuates and volatility risk plays an important role in risky allocation. Most works considering the robust investment strategy under stochastic volatility are formulated in the maxmin robust model. Only the effect of the level of ambiguity is discussed for an investor with stochastic volatility. However, when the investor has different ambiguity attitudes, the impact on the investment strategy has not been discussed.
\vskip 5pt
We have the following contributions to this paper. (1) We formulate the reinsurance-investment Stackelberg differential game between the insurer (AAI) and the ambiguity averse reinsurer (AAR) under the $\alpha$-maxmin mean-variance criterion. Different from \cite{bai2021stackelberg}, \cite{yang2022robust}, etc., the insurer is allowed
to purchase the per-loss reinsurance business as in \cite{liang2018minimizing} and \cite{li2022stackelberg}. However, model uncertainty is ignored in \cite{liang2018minimizing} and \cite{li2022stackelberg}. (2) As the empirical results do not support the geometric Brownian motion, we consider Heston's stochastic volatility in the financial market. The optimal investment strategy under the $\alpha$-maxmin mean-variance criterion is studied. (3) The optimization problem is time-inconsistent and we present the verification theorem of the equilibrium strategies and value functions rigorously. (4) The closed form of the AAI's feedback retention level to AAR's reinsurance premium is given, and we find that similar to the non-robust case of \cite{li2022stackelberg}, the excess-of-loss reinsurance is still the optimal form of reinsurance strategy. Besides, we solve the equilibrium investment strategies semi-explicitly, which is determined by a system of Riccati differential equations. We also present a global existence result for the equations. Furthermore, the equations determining the equilibrium reinsurance strategy and reinsurance premium rate are given semi-explicitly, which is simplified to an algebraic equation in a specific example.  (5) Numerical analyses reveal that the level of ambiguity, ambiguity attitude, and risk attitude of the AAI (AAR) have similar effects on the equilibrium strategies. However, the attitudes of the AAI and AAR have different impacts on the equilibrium retention level. When the AAI is more averse to risk or ambiguity, the equilibrium retention level decreases, and the reinsurance premium increases. However, when the AAR is {suffering from a higher level of ambiguity or} more averse to risk or ambiguity, the equilibrium retention level and reinsurance premium increase, leading to less insurance risk divided by the AAR. Moreover, we compare the equilibrium investment strategy with different levels of ambiguity, ambiguity attitudes, and risk attitudes. The results show that these parameters have a similar effect on the equilibrium investment, i.e., the AAI or AAR will become less interested in the financial market with a higher level of ambiguity, ambiguity attitude, or risk attitude.

\vskip 5pt
The rest of this paper is organized as follows. The insurance and financial markets, the AAI, and AAR’s wealth processes are presented in Section 2. In Section 3, we show the Stackelberg game between the AAI and AAR under the $\alpha$-maxmin mean-variance criterion. Section 4 presents the definition of the equilibrium strategy and verification theorem of the time-inconsistent problem. Section 5 shows the numerical results and Section 6 is a conclusion. All the proofs are presented in the Appendices.
\vskip 20pt
\setcounter{equation}{0}
\section{\bf{Financial Model}}
In this paper, we consider a stochastic Stackelberg differential game between an insurer (follower) and a reinsurer (leader). The insurer and the reinsurer have conflicting interests over the reinsurance premium. For the reinsurance premium determined by the reinsurer, the insurer chooses the reinsurance policy to divide insurance risk. The reinsurer sets the reinsurance premium according to the insurer's retention level. 

Assume that $(\Omega.\mathcal{F},\{\mathcal{F}_t\}_{t\in[0,T]},\mathbb{P})$ is the filtered complete probability space satisfying the usual conditions. $\mathbb{P}$ represents the reference probability measure. The insurer and reinsurer are both ambiguous about the insurance and financial market. We investigate the behaviors of the ambiguity averse insurer (AAI) and ambiguity averse reinsurer (AAR) in this paper. All the processes below are assumed to be well-defined and adapted to $\{\mathcal{F}_t\}_{t\in[0,T]}$. We suppose that the AAI and AAR are both concerned with wealth at some terminal time $T$ and adjust the strategies within the time horizon $[0,T]$.
\subsection{{\bf{Insurance and financial markets}}}
The surplus process of the insurer without taking investment and reinsurance strategies is given by the following spectrally negative L\'evy process
\begin{equation}\nonumber
U(t)=ct-\int_0^t\int_0^\infty zN(\dif t,\dif z),~~U(0)>0,
\end{equation}
where $c>0$ is the premium rate and $N(\dif t,\dif z)$ is a Poisson random measure characterizing the claims of the insurer with size range of $(z,z+\dif z)$ within time horizon $(t,t+\dif t)$. Denote $\tilde{N}(\dif t,\dif z) = N(\dif t,\dif z)-\nu(\dif z)\dif t$ as the compensated measure of
$N(\dif t, \dif z)$, where $\nu(\dif z)$ is a L\'evy measure satisfying $\int_0^1 z\nu(\dif z)<\infty$, and $\int_1^\infty z^2\nu(\dif z)<\infty$, which are required to ensure the existence of the first two moments of $U(t)$. We suppose that the insurance premium is calculated according to the expected value principle:
\begin{equation}\nonumber
c=\frac{1+\theta}{t}\mathbf{E}\left[\int_0^t\int_0^\infty z N(\dif s, \dif z)\right]=(1+\theta)\int_0^\infty z\nu(\dif z),
\end{equation}
where $\theta>0$ represents insurance premium factor of the insurer. The above spectrally negative L\'evy process includes the compound Poisson process as a special case.
\vskip 5pt
The insurer can purchase per-loss reinsurance from a reinsurer to divide the insurance risk.
Denote the retention level of the insurer as $a=\{a(t,z), t\in[0,T]\}$ with the requirement $$a(t,z)\in[0,z],\ \  \forall t\in[0,T].$$ Then the surplus process of the insurer with reinsurance strategy $a$ is as follows{{:}}
\begin{equation}\label{reini}
\dif U_I(t)=\int_0^\infty [(\theta-\eta(t))z+\eta(t)a(t,z)]\nu(\dif z)\dif t -\int_0^\infty a(t,z)\tilde{N}(\dif t,\dif z),
\end{equation}
where $\eta(t)>\theta$ is the reinsurance premium, which is determined by the reinsurer.
\vskip 5pt
With the reinsurance strategy $a$ adopted by the insurer, the reinsurer is faced with insurance risk and has a surplus process as follows
\begin{equation}\label{reinr}
\dif U_R(t)=\int_0^\infty \eta(t)[z-a(t,z)]\nu(\dif z)\dif t -\int_0^\infty [z-a(t,z)]\tilde{N}(\dif t,\dif z).
\end{equation}
From Eqs.~(\ref{reini}) and (\ref{reinr}), we see that the reinsurance premium $\eta$ has opposite effects on the surplus processes of the insurer and reinsurer. The insurer's wealth process decreases with $\eta$. The reinsurer expects to set a higher reinsurance premium, which however will reduce the willingness of the insurer to the reinsurance purchase. Hence, the trade-off of the reinsurance premium and reinsurance strategy formulates a game between the reinsurer and the insurer.
\vskip 5pt
Besides participating in the insurance market, the insurer and reinsurer also invest in the financial market. In this paper, we suppose that there are two assets: a risk-free asset (cash) and a risky asset (stock). The process of the cash is given by
\begin{equation}\nonumber
\frac{\dif S_0(t)}{S_0(t)}=r\dif t,~~S_0(0)=s_0,
\end{equation}
where $r\geq 0$ is the risk-free interest rate. For the risky asset, we consider the process of the stock with  Heston's stochastic volatility model:
\begin{equation}\nonumber
\begin{cases}
&\frac{\dif S(t)}{S(t)}=(r+\xi Y(t))\dif t+\sqrt{Y(t)}\dif W(t),\\
&\dif Y(t)=\kappa(\delta-Y(t))\dif t+\sigma\sqrt{Y(t)}\left(\rho_0\dif W(t)+\rho \dif W_1(t)\right),
\end{cases}
\end{equation}
where $W_1=\{W_1(t), t\in[0,T]\}$ and $W=\{W(t), t\in[0,T]\}$ are two independent standard Brownian motions on the space $(\Omega.\mathcal{F},\{\mathcal{F}_t\}_{t\in[0,T]},\mathbb{P})$. $\xi$, $\kappa$, $\delta$ and $\sigma$ are all positive constants. $\rho=\sqrt{1-\rho_0^2}$ and, $\rho_0$ represents the correlation between the stock index and the stochastic volatility. $\kappa$ and $\delta$ represent the reversion rate and long-term level of the volatility, respectively. We require $2\kappa\delta\geq\sigma^2$ to ensure the positivity of $Y(t)$.

The insurer and reinsurer both can invest in the financial market. Denote the money invested in $S$ at time $t$ by $\pi_I(t)$ and $\pi_R(t)$ for the insurer and reinsurer, respectively. The wealth process of the insurer under strategy $(a,\pi_I)$ is
\begin{equation}\label{XI}
\begin{split}
\dif X_I(t)=&\int_0^\infty\left [(\theta-\eta(t))z+\eta(t)a(t,z)\right]\nu(\dif z)\dif t-\int_0^\infty a(t,z)\tilde{N}(\dif t,\dif z)+rX_I(t)\dif t\\
&+\pi_I(t)\left\{[\xi Y(t)+\int_{-1}^\infty z\nu_1(\dif z)]\dif t+\sqrt{Y(t)}\dif W(t)\right\}.
\end{split}
\end{equation}
Meanwhile, the wealth process of the reinsurer under strategy $(\eta,\pi_R)$ is
\begin{equation}\label{XR}
\begin{split}
\dif X_R(t)=&\int_0^\infty\eta(t)[z-a(t,z)]\nu(\dif z)\dif t-\int_0^\infty [z-a(t,z)]\tilde{N}(\dif t,\dif z)+rX_R(t)\dif t\\
&+\pi_R(t)\left\{[\xi Y(t)+\int_{-1}^\infty z\nu_1(\dif z)]\dif t+\sqrt{Y(t)}\dif W(t)\right\}.
\end{split}
\end{equation}
\vskip 5pt
\subsection{Ambiguity attitudes}
In practice, the insurer and reinsurer cannot estimate the financial model accurately. As such, it is necessary to consider ambiguity when making a decision. Many previous works show that the agent is faced with large utility loss when ignoring ambiguity, see \cite{YLVZ2013}. To show the agents' ambiguity attitudes, we define a set of equivalent probability measures to $\mathbb{P}$
$$\mathcal{Q}:=\{\mathbb{Q}|\mathbb{Q}\sim\mathbb{P}\},$$
where the equivalent probability measure $\mathbb{Q}$ is defined by  a process $\Phi\triangleq\{(\phi_0(t,Y(t)), \phi_Y(t,$ $Y(t)), \phi_{{Z}}(t,z) ), t\in[0,T]\}$ satisfying the following conditions:
\begin{enumerate}
	\item $\phi_0(t,y)$, $\phi_Y(t,y)$  are deterministic functions of $t$ and $y$
	, and $\phi_{{Z}}(t,z)$ is a deterministic function of $t$ and $z$.

\item  $\mathbf{E}\left[\exp\left\{\frac{1}{2}\int_0^T[\phi_0(t,Y(t))^2+\phi_Y(t,Y(t))^2]\dif t\right\}\right]<\infty$, and \\
$\mathbf{E}\left[\exp\left\{\int_0^T\int_0^\infty\left [(1-\phi_{{Z}}(t,z))\ln(1-\phi_{{Z}}(t,z))+\phi_{{Z}}(t,z)\right ]\nu(\dif z)\dif t\right\}\right]<\infty$.\\
\end{enumerate}
Denote $\phi\triangleq\{(\phi_0(t,y), \phi_Y(t,y), \phi_{{Z}}(t,z) ), t\in[0,T]\}$. And the set of all functions $\phi$ satisfying the above conditions is denoted by $\Theta$.
For each $\phi\in\Theta$, we define the process $\Xi^\phi$ as
\begin{equation}
\begin{split}
\Xi^\phi(t)=\exp&\left\{-\int_0^t\phi_0(s,Y(s))\dif W(s)-\frac{1}{2}\int_0^t\phi_0(s,Y(s))^2\dif s\right.\\
&-\int_0^t\phi_Y(s,Y(s))\dif W_1(s)-\frac{1}{2}\int_0^t\phi_Y(s,Y(s))^2\dif s\\
&+\int_0^t\int_0^\infty\ln(1-\phi_{{Z}}(s,z))\tilde{N}(\dif s,\dif z)\\
&\left.+\int_0^t\int_0^\infty\left[\ln(1-\phi_{{Z}}(s,z))+\phi_{{Z}}(s,z)\right]\nu(\dif z)\dif s\right\}.
\end{split}
\end{equation}

Under Conditions (1) and (2), $\Xi^\phi$ is a $\mathbb{P}$-martingale and there is an equivalent probability measure $\mathbb{Q}$ defined by
\begin{equation}\nonumber
\frac{\dif \mathbb{Q}}{\dif \mathbb{P}}|_{\mathcal{F}_T}=\Xi^\phi(T).
\end{equation}

It follows from the Girsanov's Theorem that the following processes are standard Brownian motions under $\mathbb{Q}$
\begin{equation}\nonumber
\begin{split}
&\dif W^{\mathbb{Q}}(t)=\dif W(t)+\phi_0(t)\dif t,\\
&\dif W_1^{\mathbb{Q}}(t)=\dif W_1(t)+\phi_Y(t)\dif t.
\end{split}
\end{equation}
Besides
\begin{equation}\nonumber
\begin{split}
&\tilde{N}^{\mathbb{Q}}(\dif t,\dif z)=\tilde{N}(\dif t,\dif z)+\phi_{{Z}}(t,z)\nu(\dif z)\dif t,\\
\end{split}
\end{equation}
is the compensated Poisson random measure under $\mathbb{Q}$ with compensator  $[1-\phi_{{Z}}(t,z)]\nu(\dif z)$.

Under $\mathbb{Q}$, the wealth process of the insurer is
\begin{equation}\label{XIQ}
\begin{split}
\dif X_I(t)=&\int_0^\infty\left[(\theta-\eta(t))z+\eta(t)a(t,z)+\phi_{{Z}}(t,z)a(t,z)\right ]\nu(\dif z)\dif t\\
&-\int_0^\infty a(t,z)\tilde{N}^{\mathbb{Q}}(\dif t,\dif z)+rX_I(t)\dif t\\
&+\pi_I(t)\left\{[\xi Y(t)
-\sqrt{Y(t)}\phi_0(t)]\dif t+\sqrt{Y(t)}\dif W^{\mathbb{Q}}(t)\right\},
\end{split}
\end{equation}
the wealth process of the reinsurer is
\begin{equation}\label{XRQ}
\begin{split}
\dif X_R(t)=&\int_0^\infty(\eta(t)+\phi_{{Z}}(t,z))[z-a(t,z)]\nu(\dif z)\dif t-\int_0^\infty [z-a(t,z)]\tilde{N}^{\mathbb{Q}}(\dif t,\dif z)\\
&+rX_R(t)\dif t
+\pi_R(t)\left\{[\xi Y(t)
-\sqrt{Y(t)}\phi_0(t)]\dif t+\sqrt{Y(t)}\dif W^{\mathbb{Q}}(t)\right\},
\end{split}
\end{equation}
and the stochastic volatility $\{Y(t),t\in[0,T]\}$ follows
\begin{equation*}\nonumber
\begin{split}
\dif Y(t)=&\kappa(\delta-Y(t))\dif t+\sigma\sqrt{Y(t)}\left[-\rho_0\phi_0(t)\dif t -\rho\phi_Y(t)\dif t\right.\\
&\left.+\rho_0\dif W^{\mathbb{Q}}(t)+\rho \dif W_1^{\mathbb{Q}}(t)\right].
\end{split}
\end{equation*}
The insurer and reinsurer search for the equilibrium per-loss reinsurance strategy and reinsurance premium, respectively. Besides, they also participate in the financial market. The insurer adopts the strategy  $u_I=\{(a(t,z),\pi_I(t)), t\in[0,T]\}$ and the reinsurer adopts $u_R=\{(\eta(t),\pi_R(t)), t\in[0,T]\}$. $u_I$ is said to be admissible if it satisfies the following conditions \\
(1) $a=\{a(t,z), t\in[0,T]\}$ and $\pi_I=\{\pi_I(t), t\in[0,T]\}$ are progressively measurable w.r.t. $\{\mathcal{F}_t\}_{t\in[0,T]}$.\\
(2) $a(t,z)\in[0,z]$ and $\mathbf{E}^{\phi}\left[\int_0^T\pi_I(t)^2Y(t)\dif t\right]<+\infty$ for any $\phi\in\Theta$.\\
(3) SDE~(\ref{XI}) exists a unique strong solution for any $(t,x,y)\in{[0,T]}\times\mathbb{R}\times\mathbb{R}^+$.\\
The admissible set of all the admissible strategies of $u_I$ is denoted by $\Pi_I$. Similarly, the admissible set $\Pi_R$ of the reinsurer contains all the admissible strategies satisfying \\
(1) $\eta=\{\eta(t), t\in[0,T]\}$ and $\pi_R=\{\pi_R(t), t\in[0,T]\}$ are progressively measurable w.r.t. $\{\mathcal{F}_t\}_{t\in[0,T]}$.\\
(2) $\eta(t)>\theta$ and $\mathbf{E}^{\phi}\left[\int_0^T\pi_R(t)^2Y(t)\dif t\right]<+\infty$ for any $\phi\in\Theta$.\\
(3) SDE~(\ref{XR}) exists a unique strong solution for any $(t,x,y)\in{[0,T]}\times\mathbb{R}\times\mathbb{R}^+$.
\vskip 20pt
\setcounter{equation}{0}
\section{{\bf $\alpha$-Robust Stackelberg Game}}
In this section, we formulate the Stackelberg differential game between the AAI and AAR under the $\alpha$-maxmin mean-variance criterion. The AAI and AAR are both uncertain about the financial market and search the $\alpha$-robust optimal strategies under the mean-variance criterion, which is formulated in \cite{LLX2016}. The penalty to describe the distance between the probability measures $\mathbb{Q}$ and $\mathbb{P}$ is given by {{
\begin{equation*}
\begin{split}
h_{\mathbf{b}}(\phi(s))=&\frac{1}{\beta}\int_0^\infty\left [(1-\phi_{{Z}}(s,z))\ln(1-\phi_{{Z}}(s,z))+\phi_{{Z}}(s,z)\right]\nu(\dif z)\\
&+\frac{\phi_0(s)^2}{2\beta_0}+\frac{\phi_Y(s)^2}{2\beta_Y},
\end{split}
\end{equation*}
where $\beta$, $\beta_0$ and $\beta_Y$}} are positive constants representing the level of ambiguity for the AAI 
towards insurance risk, equity risk and volatility risk, respectively. Denote  $\mathbf{b}=(\beta,\beta_0,\beta_Y)$.

The optimization problem of the insurer under the $\alpha$-maxmin mean-variance criterion is {{
\begin{equation}\label{JIm}
\sup\limits_{a,\pi_I}J_{\alpha}(t,x,y,a,\pi_I)=\sup\limits_{a,\pi_I}
\left\{\alpha\inf\limits_{\phi}\underline{J}_{\alpha}^\phi(t,x,y,a,\pi_I)+\hat{\alpha}
\sup\limits_{\phi}\overline{J}_{\alpha}^\phi(t,x,y,a,\pi_I)\right\},
\end{equation}
where $\alpha\in[0.5,1]$ and $\hat{\alpha}=1-\alpha$. $\alpha$ represents the AAI's ambiguity attitude, and the AAI is more ambiguity aversion with a larger $\alpha$. 
\begin{equation}\label{JIp}
\begin{split}
\underline{J}_{\alpha}^\phi(t,x,y,a,\pi_I)=\mathbf{E}_{t,x,y}^\phi[X_I(T)]
-\frac{\gamma}{2}\rm{Var}_{t,x,y}^\phi[X_I(T)]+\mathbf{E}^\phi_{t,x,y}\left[\int_t^T h_{\mathbf{b}}(\phi(s))\dif s\right],\\
\overline{J}_{\alpha}^\phi(t,x,y,a,\pi_I)=\mathbf{E}_{t,x,y}^\phi[X_I(T)]
-\frac{\gamma}{2}\rm{Var}_{t,x,y}^\phi[X_I(T)]-\mathbf{E}^\phi_{t,x,y}\left[\int_t^T h_{\mathbf{b}}(\phi(s))\dif s\right],
\end{split}
\end{equation}
where $\gamma>0$ represents the AAI's risk aversion coefficient over financial and insurance risks.}}

After the AAI adopts the optimal strategy $(a^*,\pi_I^*)$, the optimization problem for the reinsurer is as follows
\begin{equation}\label{objecti}
\sup\limits_{\eta,\pi_R}\!JR_{\alpha_R}(t,x,y,\eta,\pi_R)\!=\sup\limits_{\eta,\pi_R}\!\left\{\!\alpha_R\inf\limits_{\phi}\underline{JR}_{\alpha_R}^\phi(t,x,y,\eta,\pi_R)\!+\!\hat{\alpha}_R\sup\limits_{\phi}\overline{JR}_{\alpha_R}^\phi(t,x,y,\eta,\pi_R)\!\right\}\!,
\end{equation}
where $\alpha_R\in[0.5,1]$ and $\hat{\alpha}_R=1-\alpha_R$.  $\alpha_R$ represents the AAR's ambiguity attitude, and the AAR is more ambiguity aversion with larger $\alpha_R$.
\begin{equation}\label{JRp}
\begin{split}
\underline{JR}_{\alpha_R}^\phi(t,x,y,\eta,\pi_R)=\mathbf{E}_{t,x,y}^\phi[X_R(T)]-\frac{\gamma_R}{2}\rm{Var}_{t,x,y}^\phi[X_R(T)]+\mathbf{E}^\phi_{t,x,y}\left[\int_t^T h_{\mathbf{b_r}}(\phi(s))\dif s\right],\\
\overline{JR}_{\alpha_R}^\phi(t,x,y,\eta,\pi_R)=\mathbf{E}_{t,x,y}^\phi[X_R(T)]-\frac{\gamma_R}{2}\rm{Var}_{t,x,y}^\phi[X_R(T)]-\mathbf{E}^\phi_{t,x,y}\left[\int_t^T h_{\mathbf{b_r}}(\phi(s))\dif s\right],
\end{split}
\end{equation}
where $\gamma_R>0$ represents the AAR's risk aversion coefficient of financial risk. The penalty term is given by
\begin{equation*}
\begin{split}
h_{\mathbf{b_r}}(\phi(s))=&\frac{1}{\beta_r}\int_0^\infty\left[(1-\phi_{{Z}}(s,z))\ln(1-\phi_{{Z}}(s,z))+\phi_{{Z}}(s,z)\right ]\nu(\dif z)\\
&+\frac{\phi_{0}(s)^2}{2\beta_{r0}}+\frac{\phi_{Y}(s)^2}{2\beta_{rY}},
\end{split}
\end{equation*}
where  $\beta_r$, $\beta_{r0}$ and $\beta_{rY}$ are positive constants representing the level of ambiguity for the AAR 
towards insurance risk, equity risk and volatility risk, respectively. Denote $\mathbf{b_r}=(\beta_r,\beta_{r0},\beta_{rY})$.

The procedure to solve the Stackelberg game consists of three steps:
\begin{enumerate}
	\item The AAR (leader) moves first by adopting any admissible strategy $u_R\in\Pi_R$.
	\item With response to the AAR's strategy, the AAI solves Problem (\ref{JIm}) and chooses the optimal  $u_I^*(\cdot,u_R)$.
	\item Knowing the AAI's behaviour $u_I^*(\cdot,u_R)$, the AAR solves Problem (\ref{objecti}) and derives the optimal strategy $u_R^*$.
\end{enumerate}

\vskip 20pt
\setcounter{equation}{0}
\section{{\bf{Equilibrium Strategy and Verification Theorem}}}
In this section, we present the verification theorem in the Stackelberg game between the AAI and the AAR. The mean-variance criterion is time-inconsistent, which makes Problems (\ref{JIm}) and (\ref{objecti})  time-inconsistent. Besides, the linear combination of the two mean-variance criteria also leads to time inconsistency. To solve the time-inconsistent problem, we first present the equilibrium strategy and extended HJB equations similar to  \cite{bjork2017time}.  Then, we derive the equilibrium reinsurance strategy, reinsurance premium, and investment strategies. To ensure the integrability condition, throughout this paper, we make the following integrability assumption on the L\'evy measure $\nu$.
\begin{Assumption}\label{ass1}
	(1) If there is ambiguity to the distribution of insurance claims, i.e., the function $\phi_{{Z}}(t,z):[0,T]\times[0,\infty)\rightarrow(-\infty,1)$, 
	we assume
	\begin{equation}\nonumber
		\begin{split}
			\int_{0}^1 z\nu(\mathrm{d}z)<\infty~\text{and}~\int_1^\infty e^{cz^2}\nu(\mathrm{d}z)<\infty,~\text{for~some}~c>0.\\
		\end{split}
	\end{equation}
	
	(2) If there is no ambiguity to the distribution of insurance claims, i.e., the function $\phi_{{Z}}(t):[0,T]\rightarrow(-\infty,1)$, 
	we assume
	\begin{equation}\nonumber
		\begin{split}
			\nu(0,\infty)<\infty~\text{and}~\int_1^\infty z^2\nu(\mathrm{d}z)<\infty.\\
		\end{split}
	\end{equation}
\end{Assumption}

In order to handle the time-inconsistent optimization  problems (\ref{JIm}) and (\ref{objecti}), we follow \cite{bjork2017time} to present the definition of equilibrium strategies for the two agents.
\begin{Def}\label{equilibrium}
For any fixed $(t,x,y)\in[0,T]\times\mathbb{R}\times\mathbb{R}^+$, consider an admissible strategy $u^*_I$. For any $u_I\in\Pi_I$ and $\epsilon>0$, define a new perturbed strategy $u^\epsilon_I$ by
\begin{equation}
u^\epsilon_I(s)=\begin{cases}
u_I,&t\leq s\leq t+\epsilon,\\
u^*_I(s),&t+\epsilon\leq s\leq T.
\end{cases}
\end{equation}
If for all $u_I\in\Pi_I$,
\begin{equation}
\liminf\limits_{h\rightarrow 0^+}\frac{J_\alpha(t,x,y,u^*_I)-J_\alpha(t,x,y,u^\epsilon_I)}{\epsilon}\geq 0,
\end{equation}
then $u^*_I$ is an equilibrium strategy  of the AAI and the equilibrium value function of the AAI is given by
\begin{equation}
V(t,x,y)=J_\alpha(t,x,y,u^*_I).
\end{equation}
\end{Def}
The definition of the equilibrium strategies for the AAR is similar and we omit it here. In the following, we present the verification theorems for the optimization problems of the AAI and AAR.

 First, we define the infinitesimal generator of $(X_I(t),Y(t))$ under $\mathbb{Q}$ for any function $f(t,x,y)\in C^{1,2,2}([0,T]\times\mathbb{R}\times{\mathbb{R}^+})$ by
\begin{equation}\label{geni}
\begin{split}
\mathcal{A}^{u_I,\phi}f(t,x,y)&=f_t+\left[\kappa(\delta-y)-\sigma\rho_0\sqrt{y}\phi_0-\sigma\rho\sqrt{y}\phi_Y\right]f_y\\
&+f_x\left\{rx+\!\!\int_0^\infty\![(\theta-\eta(t))z\!+\!(1+\eta)a(t,z)]\nu(\dif z)\!+\!\pi_I(\xi y-\sqrt{y}\phi_0)\right\}\\
&+\frac{1}{2}\pi_I^2 yf_{xx}+\frac{1}{2}\sigma^2 yf_{yy}+\sigma\pi_I yf_{xy}\\
&+\int_0^\infty\left[f(t,x-a(t,x),y)-f(t,x,y)\right ](1-\phi_{{Z}}(t,z))\nu(\dif z).
\end{split}
\end{equation}

Similarly, we also define the infinitesimal generator of $(X_R(t),Y(t))$ under $\mathbb{Q}$ for any function $f(t,x,y)\in C^{1,2,2}([0,T]\times\mathbb{R}\times{\mathbb{R}^+})$ by
\begin{equation}\label{genr}
\begin{split}
\mathcal{AR}^{u_R,\phi}f(t,x,y)&=f_t+[\kappa(\delta-y)-\sigma\rho_0\sqrt{y}\phi_0-\sigma\rho\sqrt{y}\phi_Y]f_y\\
&+f_x\left\{rx+\int_0^\infty(1+\eta)[z-a(t,z)]\nu(\dif z)+\pi_R(\xi y-\sqrt{y}\phi_0)\right\}\\
&+\frac{1}{2}\pi_R^2 yf_{xx}+\frac{1}{2}\sigma^2 yf_{yy}+\sigma\pi_R yf_{xy}\\
&+\int_0^\infty\left [f(t,x-z+a(t,x),y)-f(t,x,y)\right](1-\phi_{{Z}}(t,z))\nu(\dif z).
\end{split}
\end{equation}

Next, we present the verification theorems of the AAI and AAR separately. The verification theorem of Problem (\ref{JIm})  for the AAI is as follows.
\begin{theorem}[Verification theorem (AAI)]\label{THM4.1}
Suppose that there exist $V_I(t,x,y)$, $\underline{g}_I(t,x,y)$, $\overline{g}_I(t,x,y)\in C^{1,2,2}([0,T]\times\mathbb{R}\times\mathbb{R}^+)$ satisfying the following conditions{{:}}

(1) For any $(t,x,y)\in[0,T]\times\mathbb{R}\times\mathbb{R}^+$,
\begin{equation}\label{veri}
\begin{split}
0=\sup\limits_{u_I\in\Pi_I}&\left\{\alpha\inf\limits_{\phi\in\Theta}\left[\mathcal{A}^{u_I,\phi}V_I(t,x,y)-\frac{\gamma}{2}\mathcal{A}^{u_I,\phi}\underline{g}_I^2(t,x,y)\right.\right.\\
&\left.+\gamma\underline{g}_I(t,x,y)\mathcal{A}^{u_I,\phi}\underline{g}_I(t,x,y)+h_{\mathbf{b}}(\phi(t,y))\right]\\
&+\hat{\alpha}\sup\limits_{\phi\in\Theta}\left[\mathcal{A}^{u_I,\phi}V_I(t,x,y)-\frac{\gamma}{2}\mathcal{A}^{u_I,\phi}\overline{g}_I^2(t,x,y)\right.\\
&\left.\left.+\gamma\overline{g}_I(t,x,y)\mathcal{A}^{u_I,\phi}\overline{g}_I(t,x,y)-h_{\mathbf{b}}(\phi(t,y))\right]\right\}.
\end{split}
\end{equation}

(2) For any $(t,x,y)\in[0,T]\times\mathbb{R}\times\mathbb{R}^+$,
\begin{equation}\label{bondi}
\begin{cases}
V_I(T,x,y)=x,\\
\mathcal{A}^{u^*_I,\underline{\phi}^*}\underline{g}_I(t,x,y)=\mathcal{A}^{u^*_I,\overline{\phi}^*}\overline{g}_I(t,x,y)=0,\\
\underline{g}_I(T,x,y)=\overline{g}_I(T,x,y)=x.
\end{cases}
\end{equation}

(3) For any $(t,x,y)\in[0,T]\times\mathbb{R}\times\mathbb{R}^+$, $u^*(t)$, $\overline{\phi}^*(t,y)$ and $\underline{\phi}^*(t,y)$ are independent of $x$.

(4) $\underline{\phi}^*=\underline{\phi}^{u^*_I}$ and $\overline{\phi}^*=\overline{\phi}^{u^*_I}$.\\
Then $u_I^*$ is the equilibrium strategy of the AAI and the equilibrium value function is given by $V_I(t,x,y)$. Besides, $\overline{g}_I(t,x,y)=\mathbf{E}^{\overline{\phi}^*}_{t,x,y}[X_I^{u^*_I}(T)]$, and  $\underline{g}_I(t,x,y)=\mathbf{E}^{\underline{\phi}^*}_{t,x,y}[X_I^{u^*_I}(T)]$.
\end{theorem}
\begin{proof}
	See Appendix \ref{A.1}.
\end{proof}

The verification theorem of Problem  (\ref{objecti}) for the AAR  is as follows.
\begin{theorem}[Verification theorem (AAR)]\label{THM4.2}
Suppose that there exist $V_R(t,x,y)$, $\underline{g}_R(t,x,y)$, $\overline{g}_R(t,x,y)\in C^{1,2,2}([0,T]\times\mathbb{R}\times\mathbb{R}^+)$ satisfying the following conditions

(1) For any $(t,x,y)\in[0,T]\times\mathbb{R}\times\mathbb{R}^+$,
\begin{equation*}
\begin{split}
0=\sup\limits_{u_R\in\Pi_I}&\left\{\alpha_R\inf\limits_{\phi\in\Theta}\left[\mathcal{AR}^{u_R,\phi}V_R(t,x,y)-\frac{\gamma_R}{2}\mathcal{AR}^{u_R,\phi}\underline{g}_R^2(t,x,y)\right.\right.\\
&\left.+\gamma_R\underline{g}_R(t,x,y)\mathcal{AR}^{u_R,\phi}\underline{g}_R(t,x,y)+h_{\mathbf{b_r}}(\phi(t,y))\right]\\
&+\hat{\alpha}_R\sup\limits_{\phi\in\Theta}\left[\mathcal{AR}^{u_R,\phi}V_R(t,x,y)-\frac{\gamma_R}{2}\mathcal{AR}^{u_R,\phi}\overline{g}_R^2(t,x,y)\right.\\
&\left.\left.+\gamma_R\overline{g}_R(t,x,y)\mathcal{AR}^{u_R,\phi}\overline{g}_R(t,x,y)-h_{\mathbf{b_r}}(\phi(t,y))\right]\right\}.
\end{split}
\end{equation*}

(2) For any $(t,x,y)\in[0,T]\times\mathbb{R}\times\mathbb{R}^+$,
\begin{equation*}
\begin{cases}
V_R(T,x,y)=x,\\
\mathcal{AR}^{u^*_R,\underline{\phi}^*}\underline{g}_R(t,x,y)=\mathcal{AR}^{u^*_R,\overline{\phi}^*}\overline{g}_R(t,x,y)=0,\\
\underline{g}_R(T,x,y)=\overline{g}_R(T,x,y)=x.
\end{cases}
\end{equation*}

(3) For any $(t,x,y)\in[0,T]\times\mathbb{R}\times\mathbb{R}^+$, $u^*(t)$, $\overline{\phi}^*(t,y)$ {{and}} $\underline{\phi}^*(t,y)$ are independent of $x$.

(4) $\underline{\phi}^*=\underline{\phi}^{u^*_I}$ and $\overline{\phi}^*=\overline{\phi}^{u^*_I}$.\\
Then $u_R^*$ is the equilibrium strategy of the AAR and the equilibrium value function is given by $V_R(t,x,y)$. Besides, $\overline{g}_R(t,x,y)=\mathbf{E}^{\overline{\phi}^*}_{t,x,y}[X_R^{u^*_R}(T)]$, and  $\underline{g}_R(t,x,y)=\mathbf{E}^{\underline{\phi}^*}_{t,x,y}[X_R^{u^*_R}(T)]$.
\end{theorem}
\begin{proof}
	The proof is similar to Theorem \ref{THM4.1} and we omit it here.
\end{proof}

\setcounter{equation}{0}
\section{\bf Equilibrium strategies in the Stackelberg Game}
In this section, we derive the equilibrium strategies of the Stackelberg game based on Theorems \ref{THM4.1} and \ref{THM4.2}. First, we derive the equilibrium reinsurance strategy for a given reinsurance premium. Then, knowing the response of the AAI, the AAR finds the equilibrium reinsurance premium.

\subsection{{$\alpha$-}robust equilibrium strategies of the AAI with given $\eta$}
 As stated above, we need to solve Problem  (\ref{JIm}) first. In the subsection, we present our main results on the solution to the equilibrium reinsurance-investment problem for the AAI with a given reinsurance premium $\eta$.

\begin{theorem}\label{insurer}
For the AAI with a given reinsurance premium $\eta$,
\begin{enumerate}
	\item The feedback equilibrium reinsurance strategy is the excess-of-loss reinsurance
$$\hat{a}(t,z;\eta)=a_0(t;\eta)e^{-r(T-t)}\wedge z.$$
	\item  For a given reinsurance premium $\eta$,  the equilibrium retention level $a_0$ is determined as follows
\begin{equation}\label{a}
0=(1+\eta)- \left(1+\gamma a_0\right)\left[\alpha e^{\beta\left(a_0+\frac{\gamma}{2}a_0^2\right)}+\hat{\alpha}e^{-\beta\left(a_0+\frac{\gamma}{2}a_0^2\right)}\right].
\end{equation}
	\item The equilibrium investment strategy $\pi_I^*(t)$ is as follows 
\begin{equation}\label{pii}
\begin{split}
\pi_I^*(t)=\frac{\xi -(2\alpha-1)\beta_0\rho_0\sigma A(t)-\gamma\sigma\rho_0\left(\alpha\underline{H}(t)+\hat{\alpha}\overline{H}(t)\right)}{[\gamma+(2\alpha-1)\beta_0]e^{r(T-t)}}.
\end{split}
\end{equation}

	\item  The corresponding equilibrium value function is given by
\begin{equation*}
V(t,x)=e^{r(T-t)}x+A(t)y+B(t),
\end{equation*}
where {{ $A'(t)$, $\underline{H}'(t)$ and $\overline{H}'(t)$ are solutions to the following Riccati differential equations
\begin{equation}\label{ABHH}
\begin{split}
A'(t)=&\kappa A(t)+\frac{1}{2}(2\alpha-1)\sigma^2[\beta_0\rho_0^2+\beta_Y(1-\rho_0^2)]A(t)^2+\frac{\gamma}{2}\sigma^2\left[\alpha\underline{H}(t)^2+\hat{\alpha}\overline{H}(t)^2\right]\\
&-\frac{1}{2[\gamma+(2\alpha-1)\beta_0]}\left[\xi -(2\alpha-1)\beta_0\rho_0\sigma A(t)-\gamma\sigma\rho_0\left(\alpha\underline{H}(t)+\hat{\alpha}\overline{H}(t)\right)\right]^2,\\
\underline{H}'(t)=&\kappa \underline{H}(t)-\xi\frac{\xi -(2\alpha-1)\beta_0\rho_0\sigma A(t)-\gamma\sigma\rho_0\left(\alpha\underline{H}(t)+\hat{\alpha}\overline{H}(t)\right)}{\gamma+(2\alpha-1)\beta_0}\\
&+\beta_0\left[\frac{\xi -(2\alpha-1)\beta_0\rho_0\sigma A(t)-\gamma\sigma\rho_0\left(\alpha\underline{H}(t)+\hat{\alpha}\overline{H}(t)\right)}{\gamma+(2\alpha-1)\beta_0}\right]^2\\
&+\sigma\beta_0\frac{\xi -(2\alpha-1)\beta_0\rho_0\sigma A(t)-\gamma\sigma\rho_0\left(\alpha\underline{H}(t)+\hat{\alpha}\overline{H}(t)\right)}{\gamma+(2\alpha-1)\beta_0}\left(A(t)+\underline{H}(t)\right)\\
&+\sigma^2(\beta_0\rho_0^2+\beta_Y\rho^2)A(t)\underline{H}(t),\\
\overline{H}'(t)=&\kappa \overline{H}(t)-\xi\frac{\xi -(2\alpha-1)\beta_0\rho_0\sigma A(t)-\gamma\sigma\rho_0\left(\alpha\underline{H}(t)+\hat{\alpha}\overline{H}(t)\right)}{\gamma+(2\alpha-1)\beta_0}\\
&-\beta_0\left[\frac{\xi -(2\alpha-1)\beta_0\rho_0\sigma A(t)-\gamma\sigma\rho_0\left(\alpha\underline{H}(t)+\hat{\alpha}\overline{H}(t)\right)}{\gamma+(2\alpha-1)\beta_0}\right]^2\\
&-\sigma\beta_0\frac{\xi -(2\alpha-1)\beta_0\rho_0\sigma A(t)-\gamma\sigma\rho_0\left(\alpha\underline{H}(t)+\hat{\alpha}\overline{H}(t)\right)}{\gamma+(2\alpha-1)\beta_0}\left(A(t)+\overline{H}(t)\right)\\
&-\sigma^2(\beta_0\rho_0^2+\beta_Y\rho^2)A(t)\overline{H}(t),
\end{split}
\end{equation}}}
and
{{
\begin{equation*}\nonumber
\begin{split}
B(t)=\int_t^T&\left\{\kappa\lambda A(s)+(\theta-\eta(s))e^{r(T-s)}\int_0^\infty z\nu(\mathrm{d}z)+(1+\eta(s))\int_0^\infty a_0\wedge ze^{r(T-s)}\nu(\mathrm{d}z)\right.\\
&+\frac{\alpha}{\beta}\int_{0}^\infty\left\{1-e^{\beta\left[a_0\wedge ze^{r(T-s)}+\frac{\gamma}{2}(a_0\wedge ze^{r(T-s)})^2\right]}\right\}\nu(\mathrm{d}z)\\
&\left.-\frac{\hat{\alpha}}{\beta}\int_{0}^\infty\left\{1-e^{-\beta\left[a_0\wedge ze^{r(T-s)}+\frac{\gamma}{2}(a_0\wedge ze^{r(T-s)})^2\right]}\right\}\nu(\mathrm{d}z)\right\}\dif s.
\end{split}
\end{equation*}}}


\item The associated probability distortion functions of extremely ambiguity-averse measure and the extremely ambiguity-seeking measure are given respectively by
\begin{equation}\nonumber
\begin{cases}
\underline{\phi}_0(t,y)=\beta_0\sqrt{y}\left[\frac{\xi -(2\alpha-1)\beta_0\rho_0\sigma A(t)-\gamma\sigma\rho_0\left(\alpha\underline{H}(t)+\hat{\alpha}\overline{H}(t)\right)}{\gamma+(2\alpha-1)\beta_0}+\sigma\rho_0 A(t)\right],\\
\underline{\phi}_Y(t,y)=\sigma\beta_Y\rho A(t)\sqrt{y},\\
\underline{\phi}_{{Z}}(t,x,z)=1-e^{\beta\left[a_0\wedge ze^{r(T-s)}+\frac{\gamma}{2}(a_0\wedge ze^{r(T-s)})^2\right]},
\end{cases}
\end{equation}
and
\begin{equation}\nonumber
\begin{cases}
\overline{\phi}_0(t,y)=-\beta_0\sqrt{y}\left[\frac{\xi -(2\alpha-1)\beta_0\rho_0\sigma A(t)-\gamma\sigma\rho_0\left(\alpha\underline{H}(t)+\hat{\alpha}\overline{H}(t)\right)}{\gamma+(2\alpha-1)\beta_0}+\sigma\rho_0 A(t)\right],\\
\overline{\phi}_Y(t,y)=-\sigma\beta_Y\rho A(t)\sqrt{y},\\
\overline{\phi}_{{Z}}(t,x,z)=1-e^{-\beta\left[a_0\wedge ze^{r(T-s)}+\frac{\gamma}{2}(a_0\wedge ze^{r(T-s)})^2\right]}.
\end{cases}
\end{equation}
\end{enumerate}
\end{theorem}
\begin{proof}
	See Appendix \ref{A.2}.
\end{proof}
Notice that retention level $\hat{a}$ still depends on the reinsurance premium $\eta$, which is the feedback {equilibrium reinsurance} strategy to a given reinsurance premium $\eta$. From the above theorem, we see that under the $\alpha$-maxmin mean-variance criterion, the optimal per-loss reinsurance is the excess-of-loss reinsurance. The policy limit of the reinsurance business is $a_0(t;\eta)e^{-r(T-t)}$, where $a_0$ is derived from Eq.~(\ref{a}). The equilibrium investment strategy is a deterministic function and is composed of two parts: one part to manage equity risk and another part to manage volatility risk.

\vskip 20pt
\subsection{{$\alpha$-}robust equilibrium strategies of the AAR} 
After the AAI adopts the equilibrium strategies, the AAR solves Problem  (\ref{objecti}) and derives the optimal strategy $u_R^*$.  We give the equilibrium reinsurance premium $\eta^*$ of the {$\alpha$-}robust Stackelberg game in the following theorem.
\begin{theorem}\label{reinsurer}
For the AAR,
\begin{enumerate}
	\item  The equilibrium reinsurance premium $\eta^*$ of the {$\alpha$-}robust Stackelberg game is determined by the following equation
\begin{equation}\label{eta}
\begin{split}
0&=\int_{a_0(t;\eta^*)e^{-r(T-t)}}^\infty\left\{\left[ze^{r(T-t)}-a_0(t;\eta^*)-(1+\eta^*) \frac{\partial a_0}{\partial \eta}\right]\right.\\
&+\alpha_R\frac{\partial a_0}{\partial \eta}\left[1+\gamma_R\left(ze^{r(T-t)}-a_0(t;\eta^*)\right)\right]e^{\beta_r\left[ze^{r(T-t)}-a_0(t;\eta^*)+\frac{\gamma_R}{2}(ze^{r(T-t)}-a_0(t;\eta^*))^2\right]}\\
&\left.+\hat{\alpha}_R\frac{\partial a_0}{\partial \eta}\left[1+\gamma_R\left(ze^{r(T-t)}-a_0(t;\eta^*)\right)\right]e^{-\beta_r\left[ze^{r(T-t)}-a_0(t;\eta^*)+\frac{\gamma_R}{2}(ze^{r(T-t)}-a_0(t;\eta^*))^2\right]}\right\}\nu(\mathrm{d}z),
\end{split}
\end{equation}
where
\begin{equation}\label{a0}
0=(1+\eta^*)- \left(1+\gamma a_0\right)\left[\alpha e^{\beta\left(a_0+\frac{\gamma}{2}a_0^2\right)}+\hat{\alpha}e^{-\beta\left(a_0+\frac{\gamma}{2}a_0^2\right)}\right]
\end{equation}
and
\begin{equation*}
\begin{split}
1=\frac{\partial a_0}{\partial \eta}&\left\{\gamma \left[\alpha e^{\beta\left(a_0+\frac{\gamma}{2}a_0^2\right)}+\hat{\alpha}e^{-\beta\left(a_0+\frac{\gamma}{2}a_0^2\right)}\right]\right.\\
&+\left.\beta\left(1+\gamma a_0\right)^2\left[\alpha e^{\beta\left(a_0+\frac{\gamma}{2}a_0^2\right)}-\hat{\alpha}e^{-\beta\left(a_0+\frac{\gamma}{2}a_0^2\right)}\right]\right\}.
\end{split}
\end{equation*}
\item The equilibrium investment strategy $\pi_R^*(t)$ is
\begin{equation}\label{pir}
\begin{split}
\pi_R^*(t)=\frac{\xi -(2\alpha_R-1)\beta_{rY}\sigma A_R(t)-\gamma_R\sigma\rho_0\left(\alpha_R\underline{H}_R(t)+\hat{\alpha}_R\overline{H}_R(t)\right)}{[\gamma_R+(2\alpha_R-1)\beta_{r0}]e^{r(T-t)}}.
\end{split}
\end{equation}

\item The corresponding equilibrium value function is given by
\begin{equation*}
V_R(t,x)=e^{r(T-t)}x+A_R(t)y+B_R(t),
\end{equation*}
where {{$A'_R(t)$, $\underline{H}'_R(t)$ and $\overline{H}'_R(t)$ are solutions to the following Riccati differential equations
\begin{equation}\label{ARBHH}
\begin{split}
A_R'(t)=&\kappa A_R(t)+\frac{1}{2}(2\alpha_R-1)\sigma^2
(\beta_{r0}\rho_0^2+\beta_{rY}\rho^2)
A_R(t)^2+\frac{\gamma_R}{2}\sigma^2\left[\alpha_R\underline{H}_R(t)^2\right.\left.+\hat{\alpha}_R\overline{H}_R(t)^2\right]\\
&-\frac{\left[\xi -(2\alpha_R-1)\beta_{r0}\rho_0\sigma A_R(t)-\gamma_R\sigma\rho_0\left(\alpha_R\underline{H}_R(t)+\hat{\alpha}_R\overline{H}_R(t)\right)\right]^2}{2[\gamma_R+(2\alpha_R-1)\beta_{r0}]},\\
\underline{H}_R'(t)=&\kappa \underline{H}_R(t)-\xi\frac{\xi -(2\alpha_R-1)\beta_{r0}\rho_0\sigma A_R(t)-\gamma_R\sigma\rho_0\left(\alpha_R\underline{H}_R(t)+\hat{\alpha}_R\overline{H}_R(t)\right)}{\gamma_R+(2\alpha_R-1)\beta_{r0}}\\
&+\beta_{r0}\left[\frac{\xi -(2\alpha_R-1)\beta_{r0}\rho_0\sigma A_R(t)-\gamma_R\sigma\rho_0\left(\alpha_R\underline{H}_R(t)+\hat{\alpha}_R\overline{H}_R(t)\right)}{\gamma_R+(2\alpha_R-1)\beta_{r0}}\right]^2\\
&+\sigma\beta_{r0}\frac{\xi -(2\alpha_R-1)\beta_{r0}\rho_{r0}\sigma A_R(t)-\gamma_R\sigma\rho_0\left(\alpha_R\underline{H}_R(t)+\hat{\alpha}_R\overline{H}_R(t)\right)}{\gamma_R+(2\alpha_R-1)\beta_{r0}}\left(A_R(t)+\underline{H}_R(t)\right)\\
&+\sigma^2(\beta_{r0}\rho_0^2+\beta_{rY}\rho^2)A_R(t)\underline{H}_R(t),\\
\overline{H}_R'(t)=&\kappa \overline{H}_R(t)-\xi\frac{\xi -(2\alpha_R-1)\beta_{r0}\rho_0\sigma A_R(t)-\gamma_R\sigma\rho_0\left(\alpha_R\underline{H}_R(t)+\hat{\alpha}_R\overline{H}_R(t)\right)}{\gamma_R+(2\alpha_R-1)\beta_{r0}}\\
&-\beta_{r0}\left[\frac{\xi -(2\alpha_R-1)\beta_{r0}\rho_0\sigma A_R(t)-\gamma_R\sigma\rho_0\left(\alpha_R\underline{H}_R(t)+\hat{\alpha}_R\overline{H}_R(t)\right)}{\gamma_R+(2\alpha_R-1)\beta_{r0}}\right]^2\\
&-\sigma\beta_{r0}\frac{\xi -(2\alpha_R-1)\beta_{r0}\rho_0\sigma A_R(t)-\gamma_R\sigma\rho_0\left(\alpha_R\underline{H}_R(t)+\hat{\alpha}_R\overline{H}_R(t)\right)}{\gamma_R+(2\alpha_R-1)\beta_{r0}}\left(A_R(t)+\overline{H}_R(t)\right)\\
&-\sigma^2(\beta_{r0}\rho_0^2+\beta_{rY}\rho^2)A_R(t)\overline{H}_R(t),
\end{split}
\end{equation}
and
\begin{equation}\nonumber
\begin{split}
B_R(t)=\int_t^T&\left\{\kappa\lambda A_R(s)+(1+\eta^*(s))e^{r(T-s)}\int_0^\infty (z-a^*(t,z))\nu(\mathrm{d}z)\right.\\
&+\frac{\alpha_R}{\beta_r}\int_{0}^\infty\left\{1-e^{\beta_r\left[a_0^*\wedge ze^{r(T-s)}+\frac{\gamma_R}{2}(a_0^*\wedge ze^{r(T-s)})^2\right]}\right\}\nu(\mathrm{d}z)\\
&\left.-\frac{\hat{\alpha}_R}{\beta_r}\int_{0}^\infty\left\{1-e^{-\beta\left[a_0^*\wedge ze^{r(T-s)}+\frac{\gamma_R}{2}(a_0^*\wedge ze^{r(T-s)})^2\right]}\right\}\nu(\mathrm{d}z)\right\}\dif s.
\end{split}
\end{equation}}}
\item The associated probability distortion functions of extremely ambiguity-averse measure and the extremely ambiguity-seeking
measure are given respectively by
\begin{equation}\nonumber
\begin{cases}
\underline{\phi}_{r0}(t,y)=\beta_{r0}\sqrt{y}\left[\frac{\xi -(2\alpha_R-1)\beta_{r0}\rho_0\sigma A_R(t)-\gamma_R\sigma\rho_0\left(\alpha_R\underline{H}_R(t)+\hat{\alpha}_R\overline{H}_R(t)\right)}{\gamma_R+(2\alpha_R-1)\beta_{r0}}+\sigma\rho_0 A_R(t)\right],\\
\underline{\phi}_{rY}(t,y)=\sigma\beta_{rY}\rho A_R(t)\sqrt{y},\\
\underline{\phi}_{rZ}(t,x,z)=1-e^{\beta_r\left[a_0^*\wedge ze^{r(T-s)}+\frac{\gamma_R}{2}(a_0^*\wedge ze^{r(T-s)})^2\right]}{{,}}
\end{cases}
\end{equation}
and
\begin{equation}\nonumber
\begin{cases}
\overline{\phi}_{r0}(t,y)=-\beta_{r0}\sqrt{y}\left[\frac{\xi -(2\alpha_R-1)\beta_{r0}\rho_0\sigma A_R(t)-\gamma_R\sigma\rho_0\left(\alpha_R\underline{H}_R(t)+\hat{\alpha}_R\overline{H}_R(t)\right)}{\gamma_R+(2\alpha_R-1)\beta_{r0}}+\sigma\rho_0 A_R(t)\right],\\
\overline{\phi}_{rY}(t,y)=-\sigma\beta_{rY}\rho A_R(t)\sqrt{y},\\
\overline{\phi}_{rZ}(t,x,z)=1-e^{-\beta_r\left[a_0^*\wedge ze^{r(T-s)}+\frac{\gamma_R}{2}(a_0^*\wedge ze^{r(T-s)})^2\right]}.
\end{cases}
\end{equation}
\end{enumerate}
\end{theorem}
\begin{proof}
	The proof is similar to Theorem \ref{insurer} and we omit it here.
\end{proof}
Theorem \ref{reinsurer} shows that the equilibrium reinsurance premium is derived by Eq.~(\ref{eta}). In the Stackelberg game, the AAR sets the reinsurance price $\eta^*$ and the AAI purchases excess-of-loss reinsurance with policy limit $a_0(t;\eta^*)e^{-r(T-t)}$. We also observe that the equilibrium investment strategies of the AAI and AAR have the same form. To ensure that the existence of the equilibrium reinsurance premium $\eta^*$, we have the following proposition.
\begin{proposition}\label{exists}
Eq.~(\ref{eta}) exists a solution, i.e., the equilibrium reinsurance premium $\eta^*$ exists.
\end{proposition}
\begin{proof}
See Appendix \ref{existproof}.	
\end{proof}

{Thus}, we obtain the equilibrium reinsurance premium $\eta^*$ of the {$\alpha$-}robust Stackelberg game in Theorem \ref{reinsurer}. Based on part 1 of Theorem \ref{insurer}, the equilibrium retention level of the {$\alpha$-}robust Stackelberg game is $a^*(t,z)=\hat{a}(t,z;\eta^*)=a_0(t;\eta^*)e^{r(T-t)}\wedge z$, where $a_0(t;\eta^*)$ is determined by
\begin{equation}\label{ae}
0=(1+\eta^*)- \left(1+\gamma a_0\right)\left[\alpha e^{\beta\left(a_0+\frac{\gamma}{2}a_0^2\right)}+\hat{\alpha}e^{-\beta\left(a_0+\frac{\gamma}{2}a_0^2\right)}\right].
\end{equation}
{Denote $a^*_0=a_0(\cdot;\eta^*)$.} Then, the equilibrium per-loss reinsurance  of the $\alpha$-maxmin Stackelberg game is obtained.

Last, in order to give an example of the equilibrium reinsurance contract, we choose a specific L\'{e}vy measure. Suppose that the aggregate insurance claims follow a compound Poisson process,
that is, the dynamics of the insurance surplus process without reinsurance and investment are governed by
\begin{equation}\nonumber
U(t)=U(0)+ct-\sum_{i=1}^{N(t)}Z_i,
\end{equation}
where $\{N(t),t\geq0\}$ is a homogeneous Poisson process with intensity $\lambda_0>0$, and $\{Z_i\}_{i\in\mathbb{N}}$ is a sequence of positive, independent and identically distributed random variables. $Z_i$ is assumed to follow the Rayleigh distribution supported on $(0,+\infty)$ with parameter $\lambda$. As such, the associated L\'{e}vy measure $\nu$ is given by
\begin{equation}\nonumber
\nu(\mathrm{d}z)=\lambda_0 \frac{2z}{\lambda^2}e^{-\frac{z^2}{\lambda^2}},
\end{equation}
which clearly satisfies Assumption \ref{ass1}.

Then, Eq.~(\ref{eta}) is equivalent to
\begin{equation}\label{equ111}
0=I_0(a_0)+I_-(a_0)+I_+(a_0),
\end{equation}
where {{
\begin{equation}\nonumber
\begin{split}
I_0(a_0)=&\sqrt{\pi}\lambda e^{r(T-t)}\Phi\left(-\frac{\sqrt{2}a_0}{\lambda e^{r(T-t)}}\right)-(1+\eta)\frac{\partial a_0}{\partial \eta}e^{-\frac{a_0^2e^{-2r(T-t)}}{\lambda^2}},\\
I_\pm(a_0)=&\frac{1\pm(2\alpha_R-1)}{2} \frac{\partial a_0}{\partial \eta} e^{\pm\frac{1}{2}\beta_{r} u a_0+\frac{1}{2}\sigma_\pm^2\beta_r^2u^2}\\
&\left[\frac{2\sigma_\pm^2e^{2r(T-t)}}{\lambda^2}\left(1\mp\sigma_\pm^2\beta_r\gamma_R(1-a_0)\right)e^{-\frac{1}{2\sigma_\pm^2}(a_0\pm\sigma_\pm^2\beta_r u)^2}\right.\\
&+\left.\frac{2\sqrt{2\pi}\sigma_\pm^3e^{2r(T-t)}}{\lambda^2}\left(\gamma_R\mp\beta_r(1\mp\sigma_\pm^2\beta_r\gamma_R)u^2\right)\Phi\left(\mp\sigma_\pm^2\beta_r u-\frac{a_0}{\sigma_\pm^2}\right)\right]
\end{split}
\end{equation}
and
\begin{equation}\nonumber
\begin{cases}
\frac{1}{\sigma_\pm^2}=\frac{2e^{2r(T-t)}}{\lambda^2}\pm\beta_r\gamma_R,\\
u=1-\gamma_R a_0.
\end{cases}
\end{equation}}}
Eq.~(\ref{equ111}) can be solved numerically to obtain the equilibrium retention level and reinsurance price.
\vskip 25pt
\setcounter{equation}{0}
\section{\bf Sensitivity Analysis}
In this section, we present some numerical examples to show the relationship between the AAI's equilibrium per-loss reinsurance and the AAR's equilibrium reinsurance premium, and the effects of different parameters on the equilibrium strategies of the AAI and the AAR. Unless otherwise stated, the baseline parameters are given by $r=0.05$, $\xi=1$, $\kappa=3$, $\delta=0.09$, $\sigma=0.5$, $\rho_0=0.5$, $\alpha=0.8$, $\gamma=0.5$, $\beta_0=4$, $\beta_Y=4$, $\beta=0.1$, $\beta_{r0}=4$, $\beta_{rY}=4$, $\beta_r=0.1$, $T=10$.
Based on the parameter settings above, the equilibrium retention level of the AAI  is $a^*_0(5)=1.108$, and the equilibrium reinsurance premium of the AAR is $\eta^*=0.7017$.

\subsection{The AAI's feedback retention level $a_0$ for a given  $\eta$}
\begin{figure}[htbp]
	\centering
	\begin{minipage}{0.5\textwidth}
		\centering
		\includegraphics[totalheight=6cm]{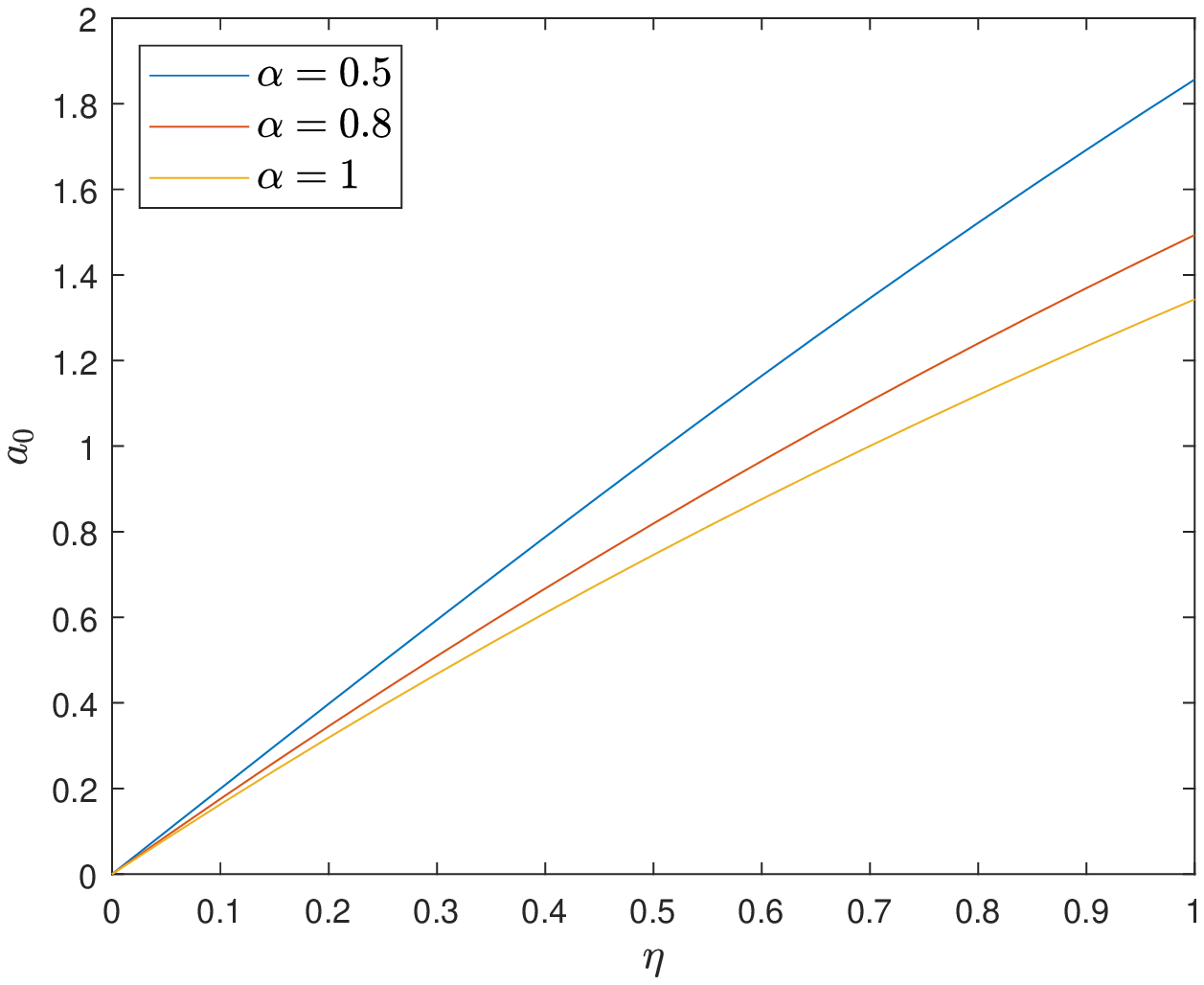}
		\caption{{Feedback retention levels $a_0$.}}
		\label{fig:aetaalpha}
	\end{minipage}\hfill
	\begin{minipage}{0.5\textwidth}
		\centering
		\includegraphics[totalheight=6cm]{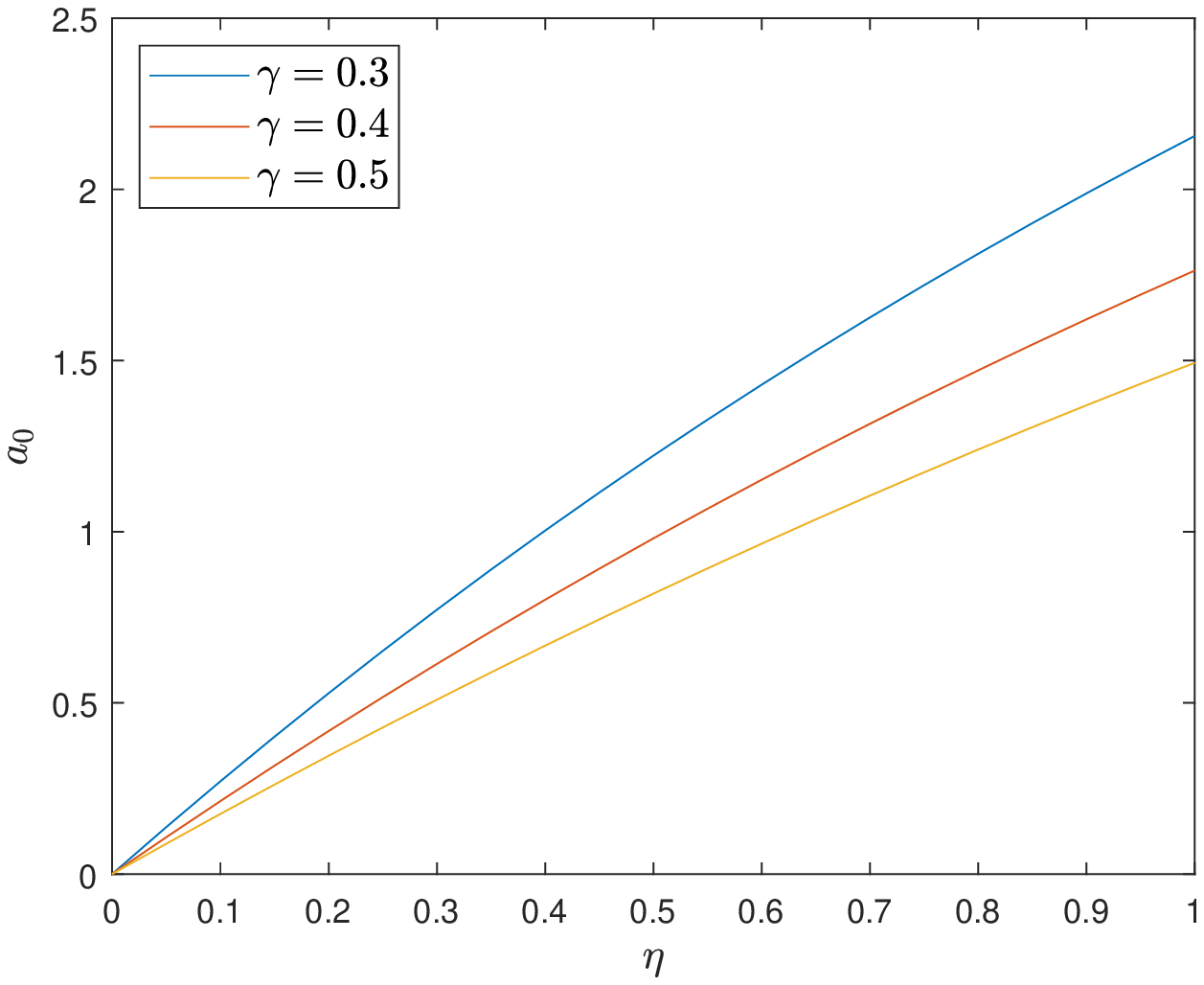}
		\caption{{Feedback retention levels $a_0$.}}
		\label{fig:aetagamma}
	\end{minipage}\hfill
\end{figure}

For a given reinsurance premium $\eta$, the AAI's retention level $a_0$ is given by Eq.~(\ref{a}). Fig.~\ref{fig:aetaalpha} shows the effects of $\eta$ and $\alpha$ on the AAI's feedback retention level. With a larger $\eta$, the reinsurance business is more expensive. Then the AAI will undertake more insurance risk by itself. Figs.~\ref{fig:aetaalpha}  and \ref{fig:aetagamma} show that the AAI's feedback retention level increases with $\eta$, i.e., the AAI will divide less insurance risk to the AAR when the reinsurance premium becomes larger. Besides,  when the AAI is more ambiguity aversion, the AAI will undertake less insurance risk by itself and decreases the retention level, which is illustrated in Fig.~\ref{fig:aetaalpha}. Fig.~\ref{fig:aetagamma} also depicts the relationship between risk aversion coefficient $\gamma$ and the retention level. We see that the risk aversion attitude has a similar effect to the ambiguity aversion attitude. If the AAI is more risk aversion, it will take fewer risks and divide more insurance risk to the AAR, which is well shown in Fig.~\ref{fig:aetagamma}.

%


%



\subsection{The equilibrium retention level and reinsurance premium}
In this subsection, we study the effects of the economic parameters on the equilibrium retention level $a^*_0$ and reinsurance premium $\eta^*$. The effects of the ambiguity aversions $\alpha$ and $\alpha_R$ on the equilibrium retention level $a^*_0$ and the equilibrium reinsurance premium $\eta^*$ are shown in Figs.~\ref{fig:aealpha} and \ref{fig:etaealpha}.

\begin{figure}[h]
  \centering
  \includegraphics[totalheight=5cm]{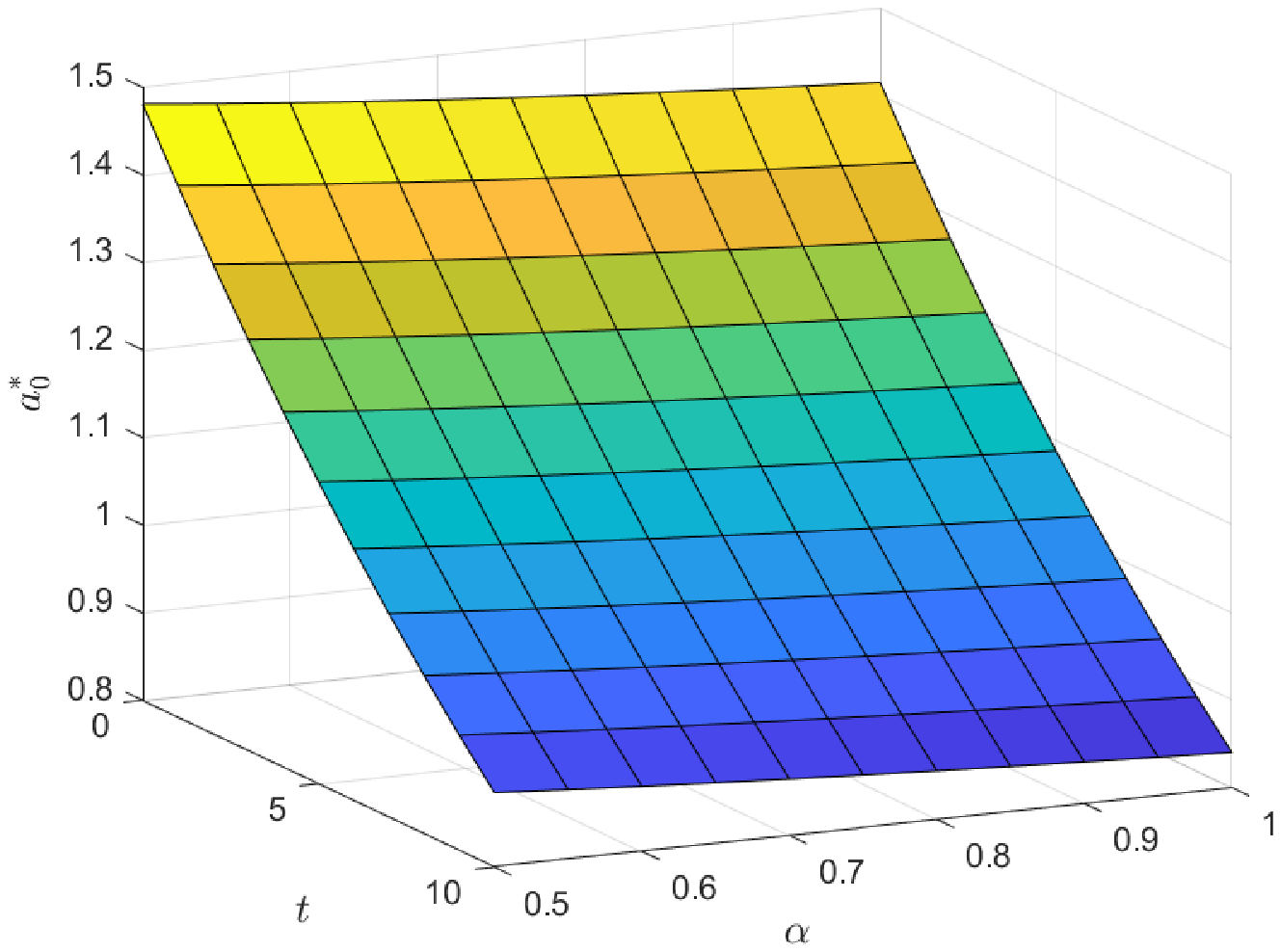}
  \includegraphics[totalheight=5cm]{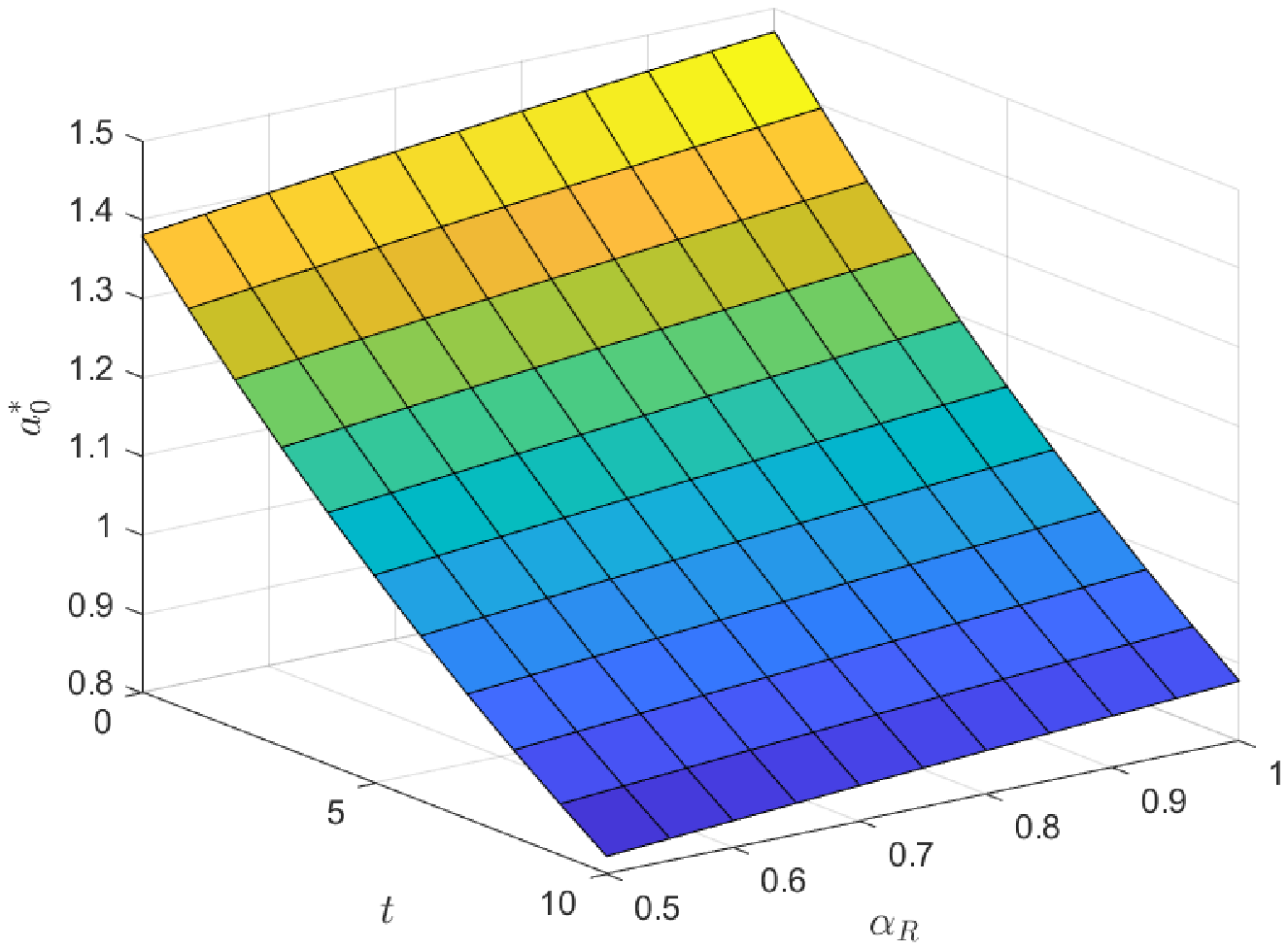}
\caption{{Effects of $\alpha$ and $\alpha_R$ on the equilibrium retention level $a^*_0$.}}
  \label{fig:aealpha}
\end{figure}

Fig.~\ref{fig:aealpha} shows that the AAI's equilibrium retention level decreases with time, which means that the AAI will undertake less insurance risk as time goes by. Meanwhile, we observe that the AAI's equilibrium retention level decreases with $\alpha$, which is also stated in  Fig.~\ref{fig:aetaalpha}, {{while the decrease is smaller than that in Fig.~\ref{fig:aetaalpha}}}. Fig.~\ref{fig:aealpha} also reveals that the AAI and AAR's ambiguity aversion attitudes have the opposite effects on the equilibrium retention level. If the AAR is more ambiguity aversion, the AAR will expect to undertake less insurance risk. Then more insurance risk is taken by the AAI. Thus, the AAI's equilibrium retention level increases with $\alpha_R$.


\begin{figure}[h]
  \centering
  \includegraphics[totalheight=5cm]{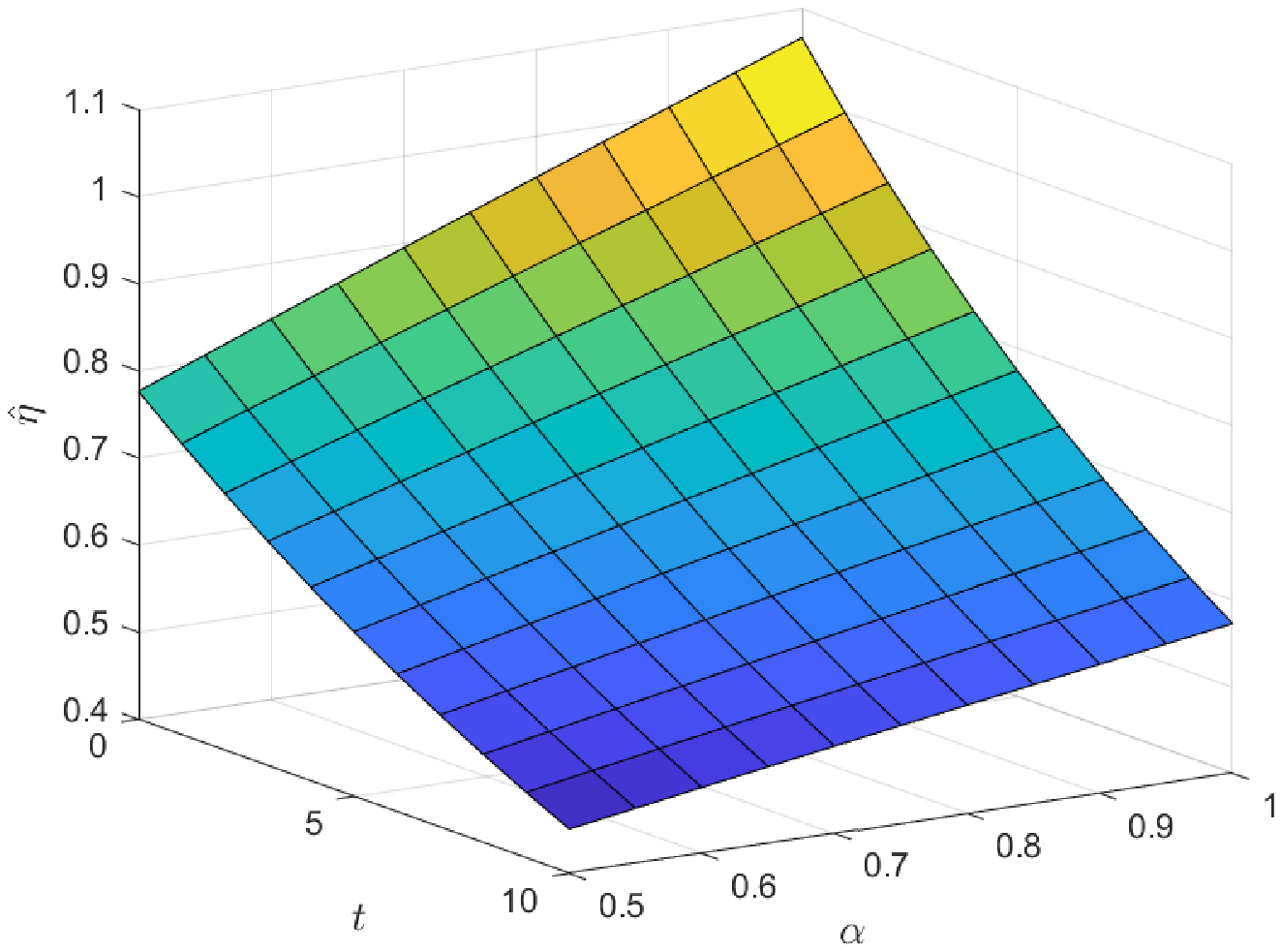}
  \includegraphics[totalheight=5cm]{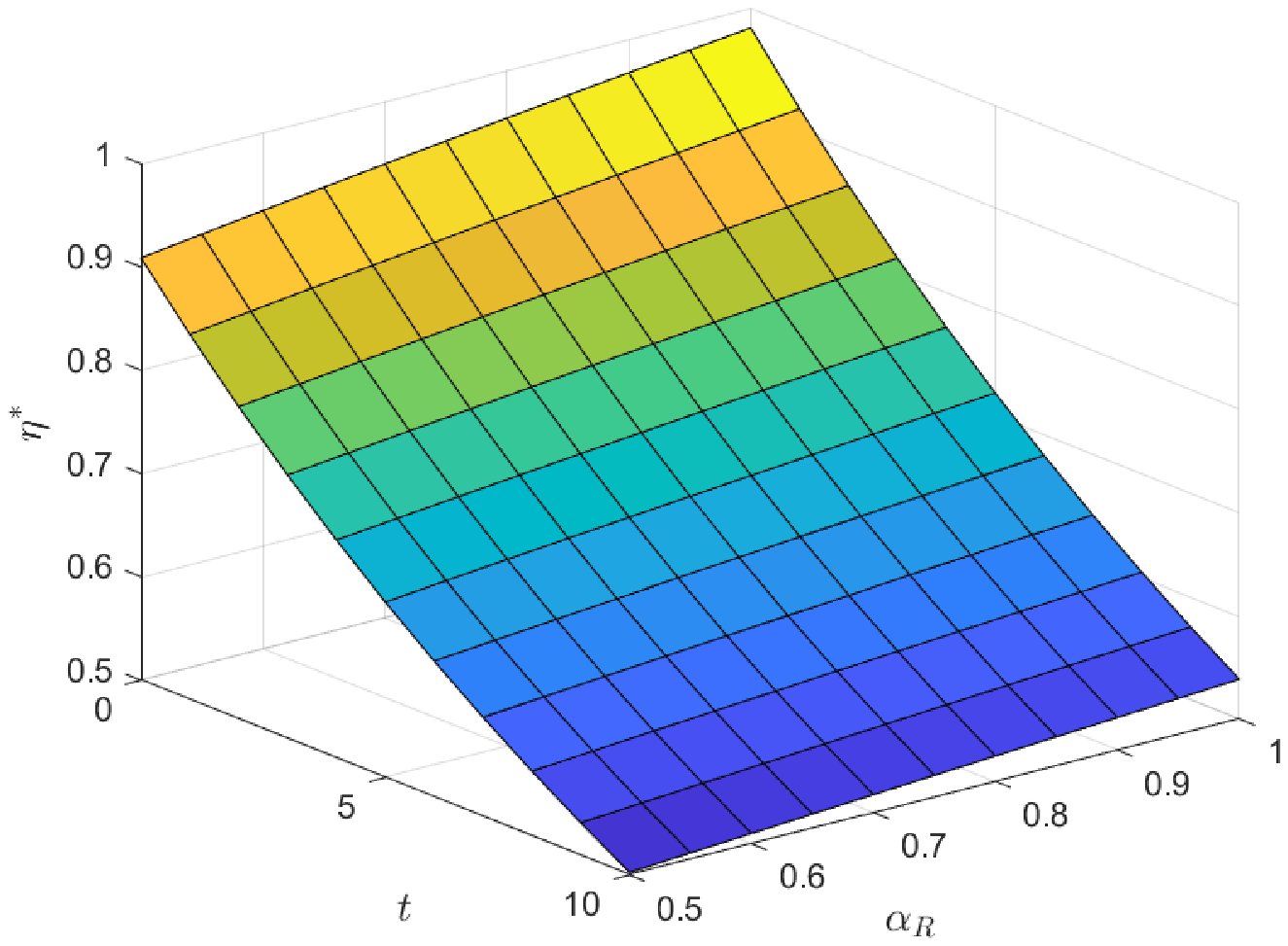}
\caption{{Effects of $\alpha$ and $\alpha_R$ on the equilibrium reinsurance premium $\eta^*$.}}
  \label{fig:etaealpha}
\end{figure}

The effects of $\alpha$ and $\alpha_R$ on the equilibrium reinsurance premium are indicated in Fig.~\ref{fig:etaealpha}. When the AAI is more ambiguity aversion, the AAI tends to undertake less insurance risk {{and expects the AAR to share more insurance risk, which means that the AAI is willing to afford a higher premium for the same retention level. Then, the AAI would like to purchase more reinsurance, which makes the reinsurance market become a seller's market}}, and the AAR will {{increase}} the reinsurance premium. Fig.~\ref{fig:etaealpha} well illustrates the {{positive}} relationship between $\alpha$ and $\eta^*$. On the other hand, when the AAR is more ambiguity aversion, the AAR expects to undertake fewer risks. Then the AAR will increase the reinsurance premium, which in turn decreases the insurance risk divided tothe AAR. In Fig.~\ref{fig:etaealpha}, we see that the equilibrium reinsurance premium {{increases}} with the AAR's ambiguity aversion coefficient $\alpha_R$.



\begin{figure}[h]
  \centering
  \includegraphics[totalheight=5cm]{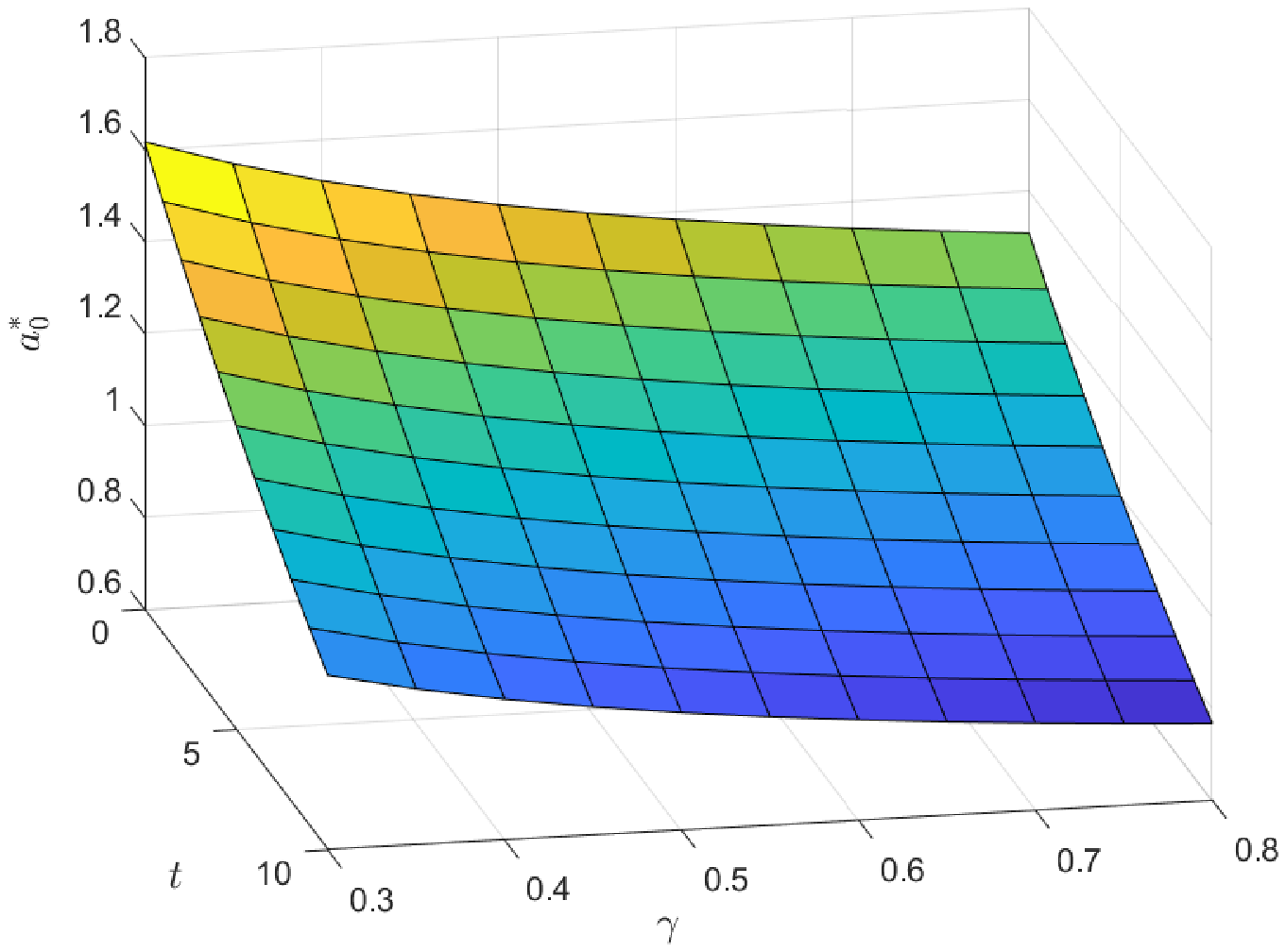}
  \includegraphics[totalheight=5cm]{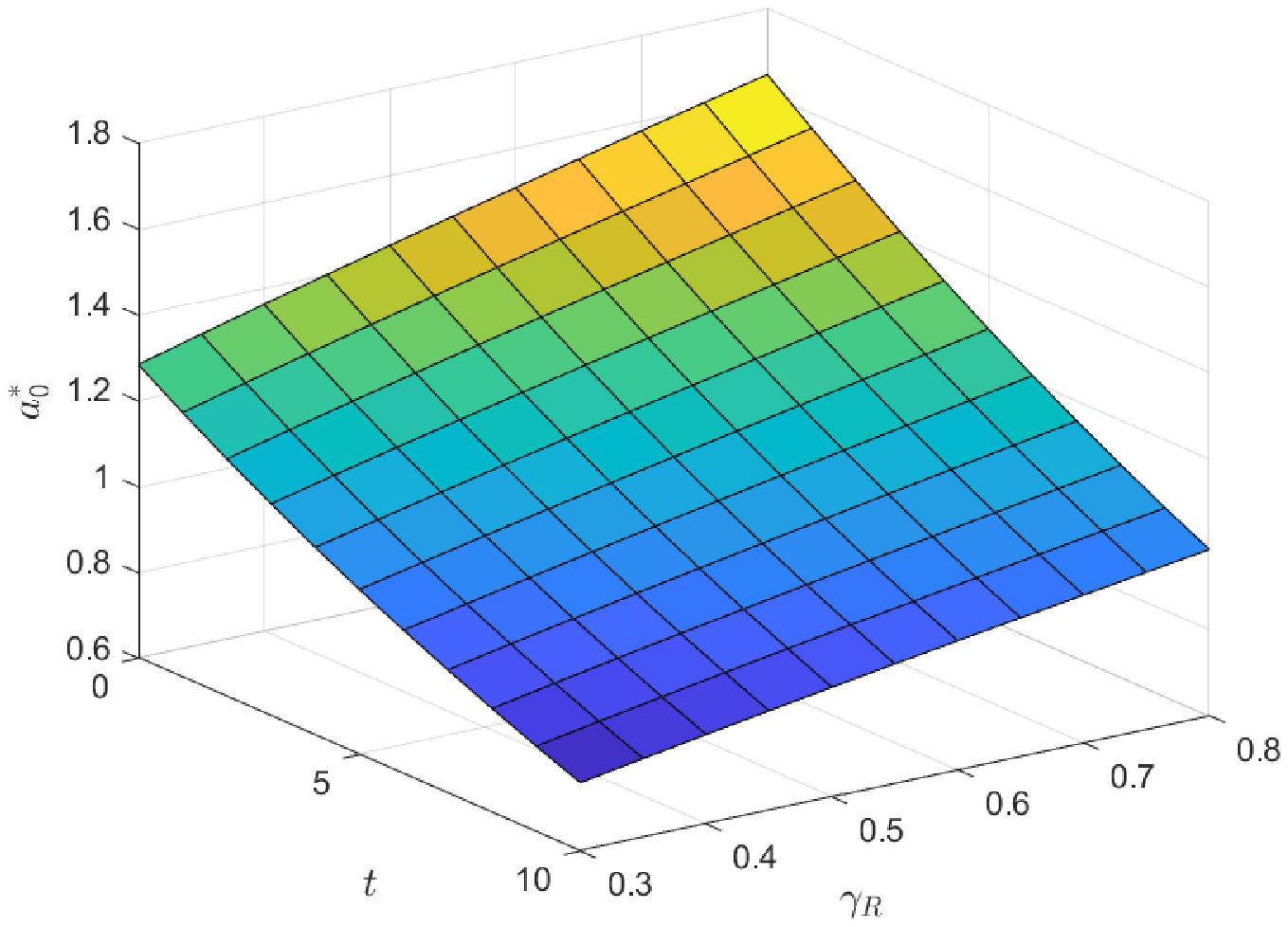}
\caption{{Effects of $\gamma$ and $\gamma_R$ on the equilibrium retention level $a^*_0$.}}
  \label{fig:aegamma}
\end{figure}
\begin{figure}[h]
  \centering
  \includegraphics[totalheight=5cm]{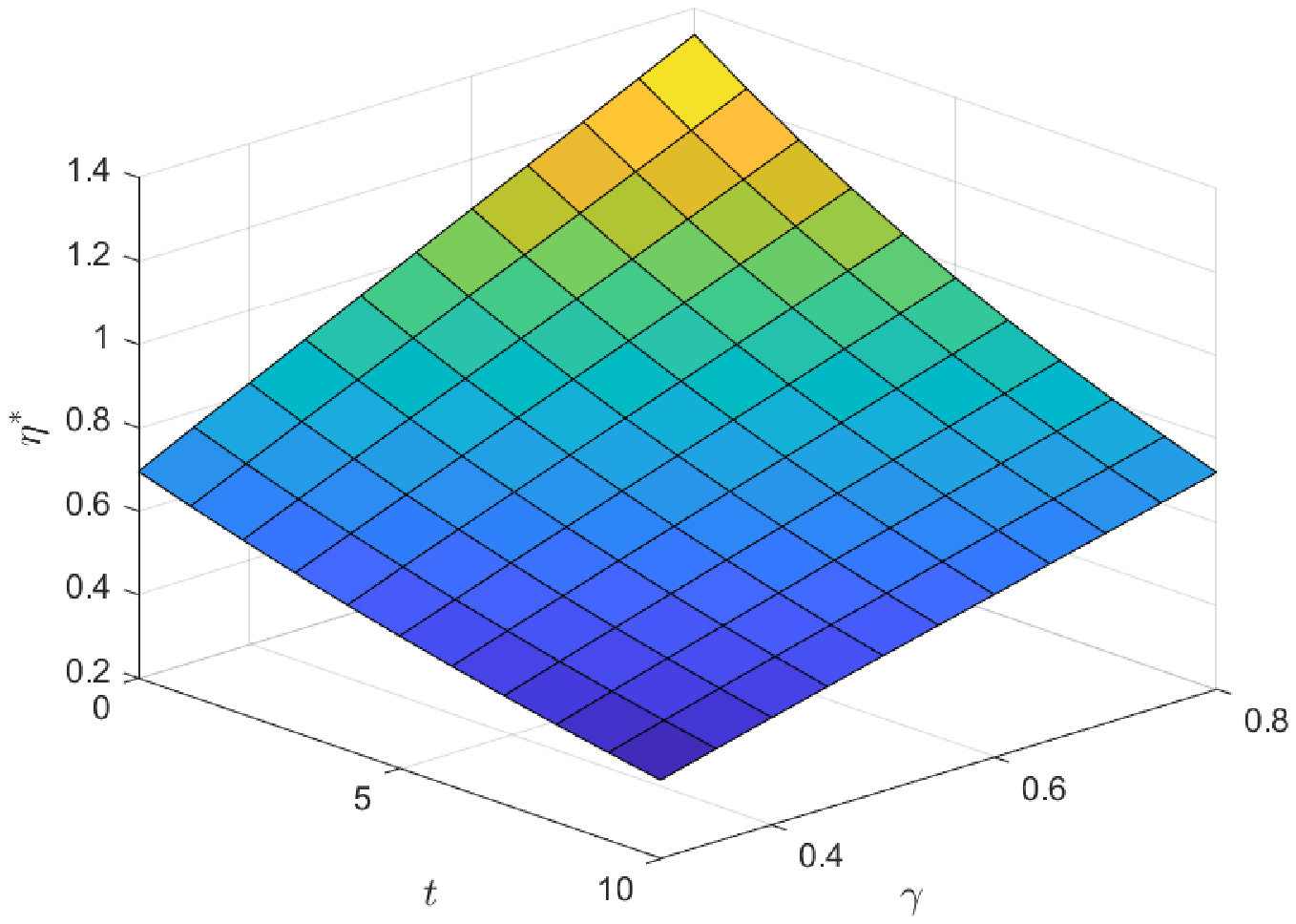}
  \includegraphics[totalheight=5cm]{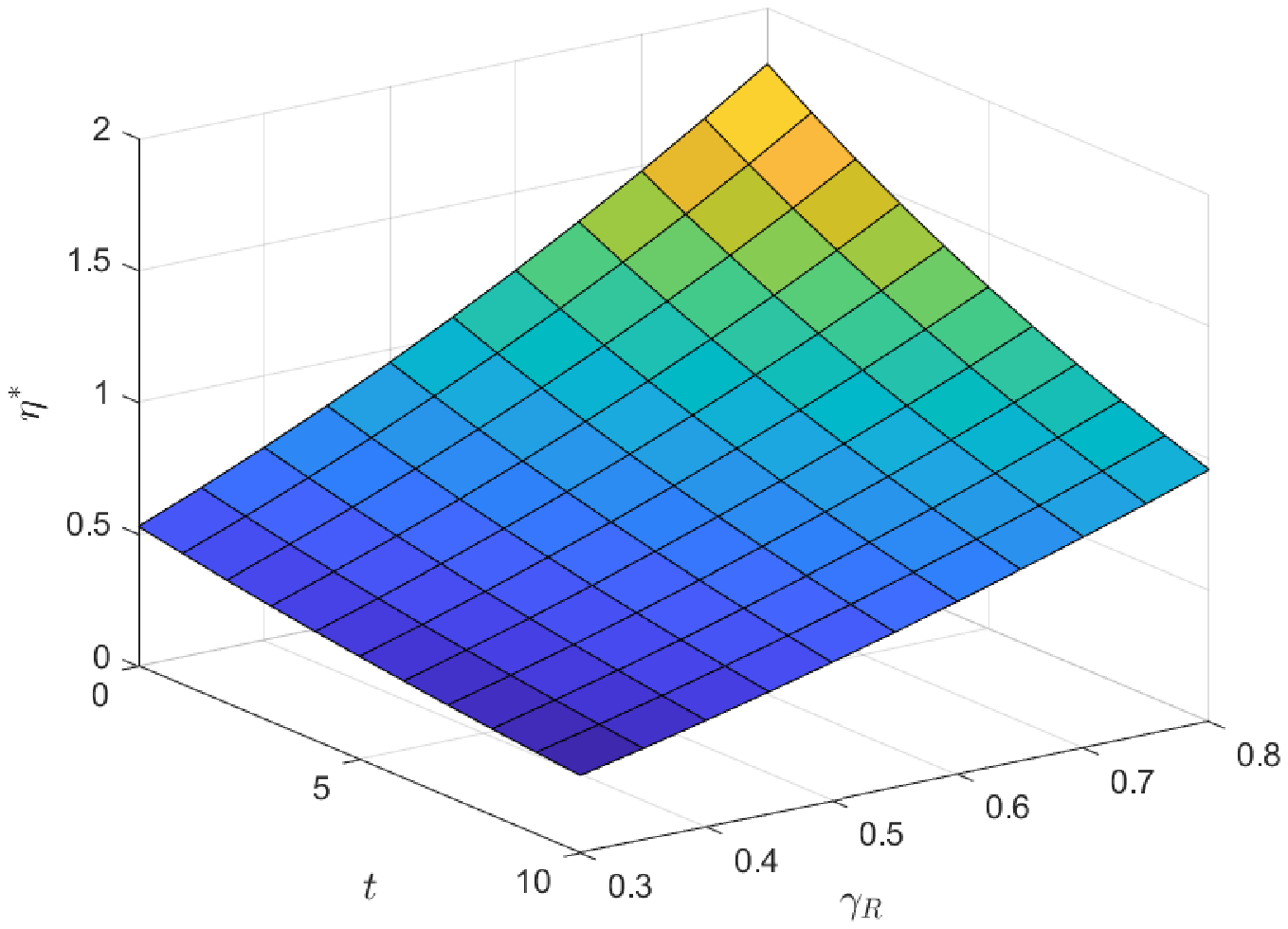}
\caption{{Effects of $\gamma$ and $\gamma_R$ on the equilibrium reinsurance premium $\eta^*$.}}
  \label{fig:etagamma}
\end{figure}
We are also interested in the effects of the risk aversion coefficients $\gamma$ and $\gamma_R$ on the AAI and AAR's behaviors. Fig.~\ref{fig:aegamma} shows the effect of $\gamma$ and $\gamma_R$ on the equilibrium retention level $a^*_0$. Observing Fig.~\ref{fig:aegamma}, we see that the equilibrium retention level decreases with $\gamma$, which is consistent with Fig.~\ref{fig:aetagamma}, {{while the decrease is less than that in Fig.~\ref{fig:aetagamma}}}.  Meanwhile, when the AAR is more risk averse, the equilibrium reinsurance premium becomes expensive, which leads to a higher equilibrium retention level. In Fig.~\ref{fig:etagamma}, the relationship between the equilibrium reinsurance premium and the risk aversion coefficients is depicted. If the AAI is more risk aversion, the AAI {{will expect the AAR to share more insurance risk and would like to purchase more reinsurance, which makes}} the AAR increase the reinsurance premium. 
On the other hand, if the AAR is more risk aversion, the reinsurance contract will become expensive to decrease the insurance risk divided from the AAI.
\begin{figure}[h]
	\centering
	\includegraphics[totalheight=5cm]{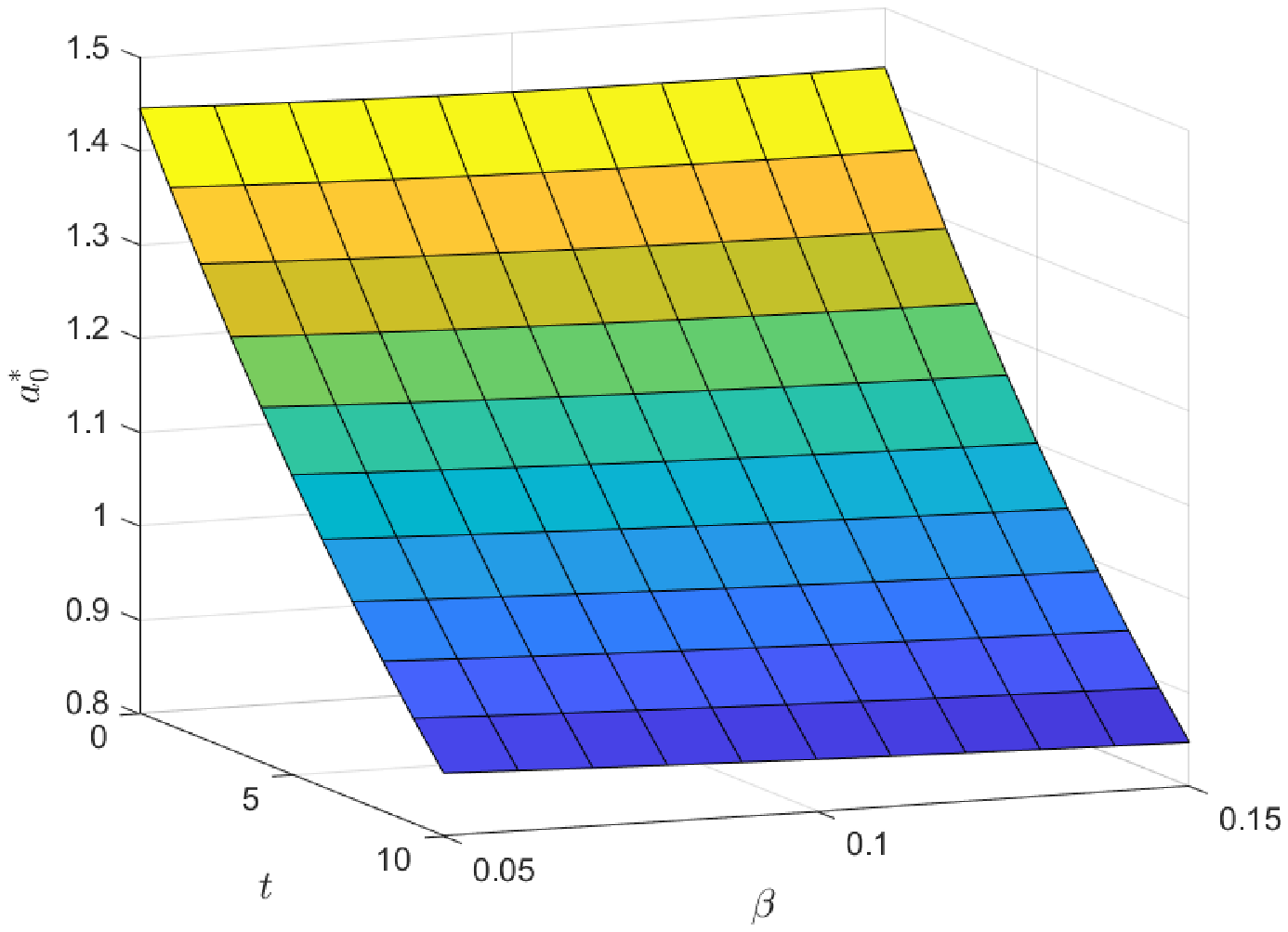}
	\includegraphics[totalheight=5cm]{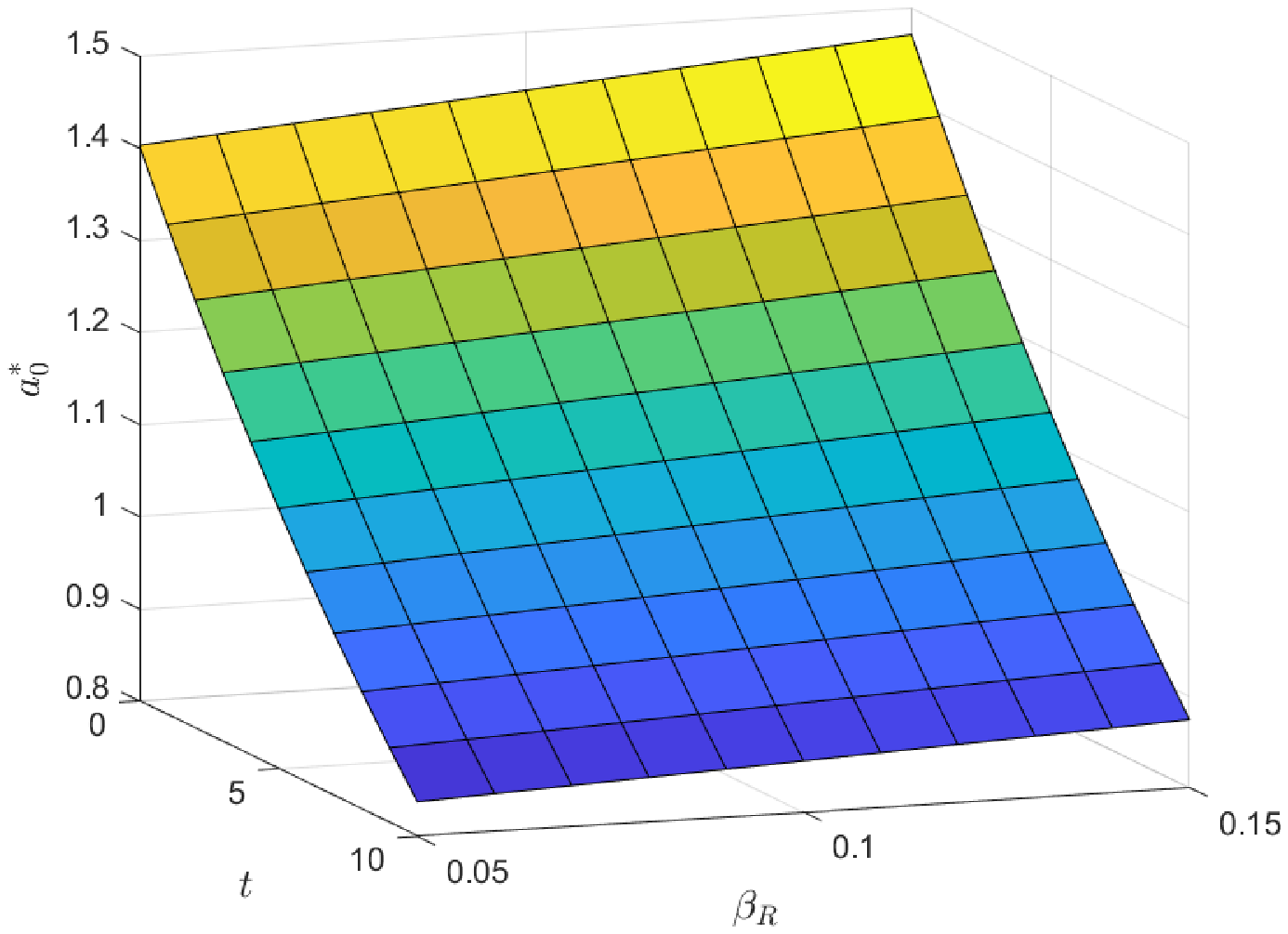}
	\caption{{Effects of $\beta$ and $\beta_R$ on the equilibrium reinsurance premium $a^*_0$.}}
	\label{fig:aebeta}
\end{figure}

\begin{figure}[h]
	\centering
  \includegraphics[totalheight=5cm]{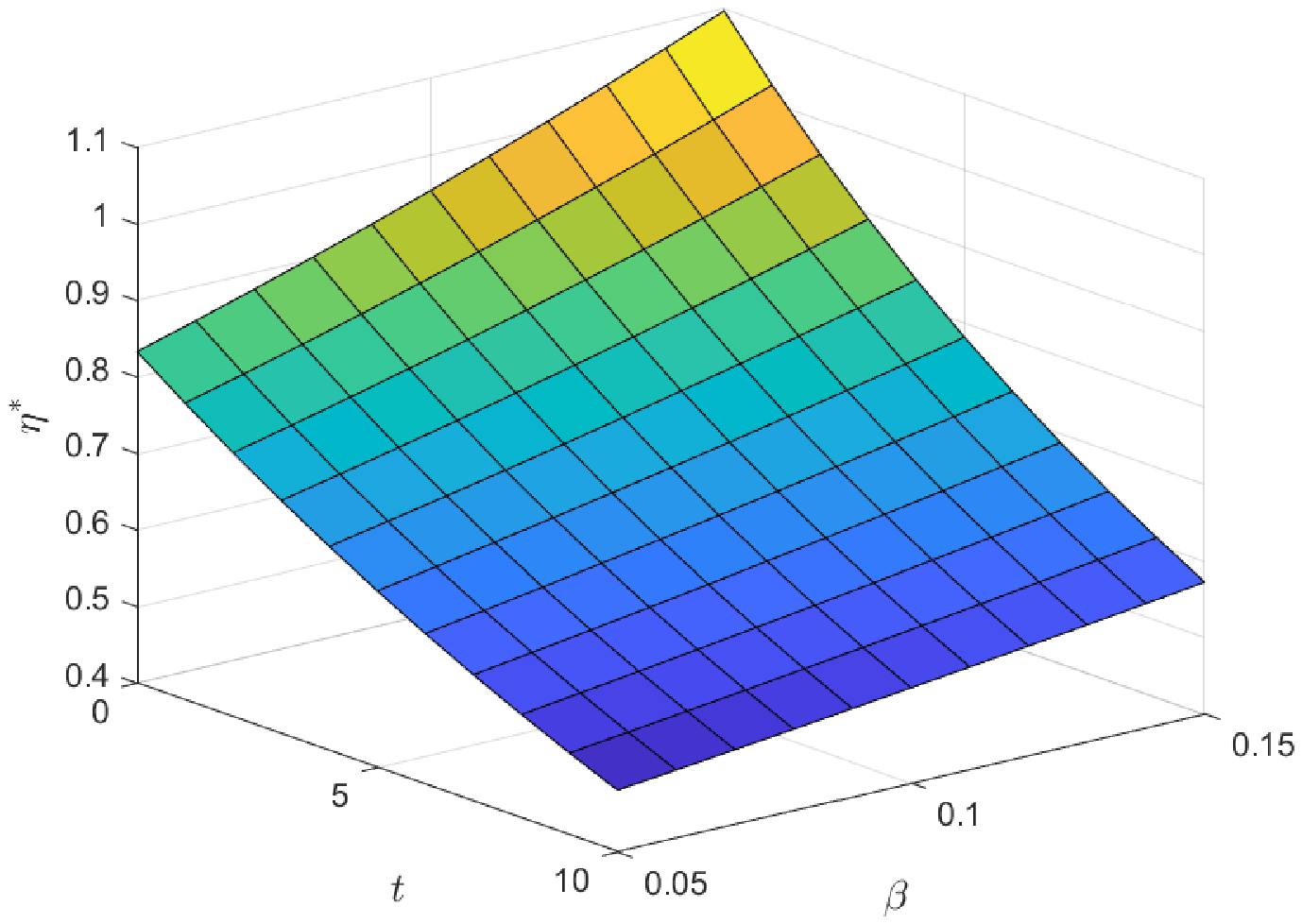}
\includegraphics[totalheight=5cm]{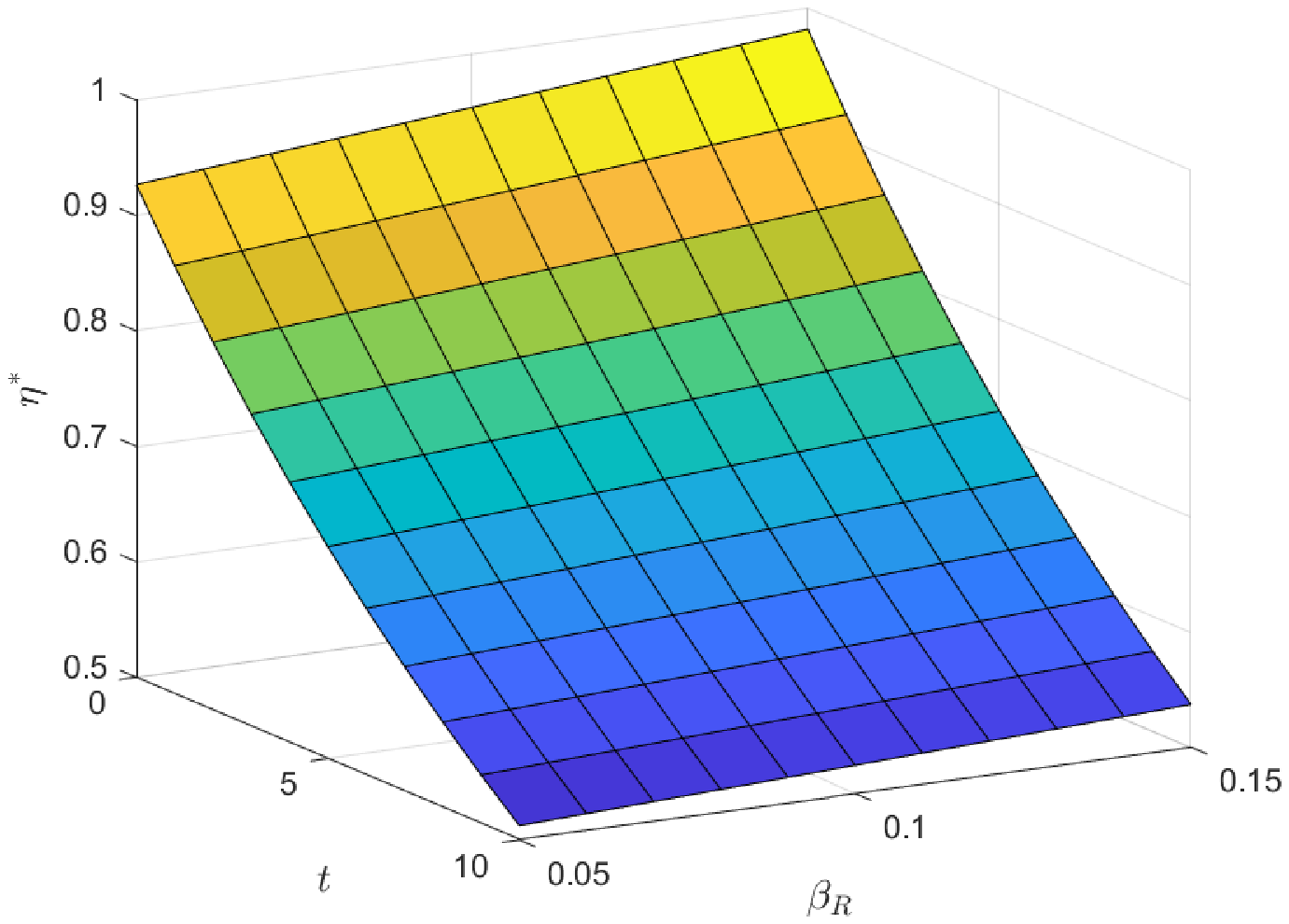}
\caption{{Effects of $\beta$ and $\beta_R$ on the equilibrium reinsurance premium $\eta^*$.}}
\label{fig:etabeta}
\end{figure}
{{Moreover, $\beta$ and $\beta_R$ represent the level of ambiguity towards the insurance market for the AAI and the AAR, respectively. The effects of the level of ambiguity towards the insurance market on the equilibrium retention level $a^*_0$ and the equilibrium reinsurance premium $\eta^*$ are shown in Figs.~\ref{fig:aebeta} and \ref{fig:etabeta}. Fig.~\ref{fig:aebeta} shows that 
the equilibrium retention level decreases with $\beta$ because the AAI with a higher level of ambiguity tends to undertake less insurance risk. 
Meanwhile, if the level of ambiguity for the AAR increases, the equilibrium reinsurance premium becomes expensive, which leads to a higher equilibrium retention level. In Fig.~\ref{fig:etabeta}, the relationship between the equilibrium reinsurance premium and the level of ambiguity towards the insurance market is depicted. If the AAI faces a higher level of ambiguity towards the insurance market, the AAI will expect the AAR to share more insurance risk and would like to purchase more reinsurance, which makes the AAR increase the reinsurance premium. On the other hand, if the level of ambiguity for the AAR increases, the reinsurance contract will become expensive to decrease the insurance risk divided from the AAI.}}

Comparing Figs.~\ref{fig:aealpha}-\ref{fig:etabeta}, we see that the risk aversion coefficient $\gamma$ ($\gamma_R$) {{and the level of ambiguity towards the insurance market $\beta$ ($\beta_R$)}} have the similar influences as the ambiguity aversion coefficient $\alpha$ ($\alpha_R$). On the one hand, when the AAI is more averse to risk or ambiguity, the equilibrium retention level decreases while the reinsurance premium increases. Then more insurance risk is undertaken by the AAR. On the other hand, when the AAR is more averse {{to}} risk or ambiguity, the equilibrium retention level and reinsurance premium increase, which leads to less insurance risk divided to the AAR.

%
\subsection{The equilibrium investment strategy}
We compare the equilibrium investment strategy $\pi^*$ and the equilibrium investment strategy $\tilde{\pi}$ without stochastic volatility. {Based on Eqs.~(\ref{pii}) and (\ref{pir}), the equilibrium investment strategies of the AAI and the AAR 
	have the same form, as such, we only present the results of $\pi^*_I$ as an example in this subsection.} Based on Eq.~(\ref{pii}), the equilibrium investment strategy is
\begin{equation}\nonumber
	\begin{split}
		\pi^*_I(t)=\frac{\xi -(2\alpha-1)\beta_0\rho_0\sigma A(t)-\gamma\sigma\rho_0\left(\alpha\underline{H}(t)+\hat{\alpha}\overline{H}(t)\right)}{[\gamma+(2\alpha-1)\beta_0]e^{r(T-t)}}.
	\end{split}
\end{equation}
$\pi^*_I(t)$ is composed of two parts: the first part to hedge equity risk and the second part to hedge volatility risk. In the case of a geometric Brownian motion with drift $\mu$ and volatility $\sigma_0$, the equilibrium investment strategy is given by
\begin{equation}\label{pii0}
	\begin{split}
		\tilde{\pi}_I(t)=\frac{(\mu-r)e^{-r(T-t)}}{\sigma_0^2[\gamma+(2\alpha-1)\beta_0]}.
	\end{split}
\end{equation}
We set $\xi=\frac{\mu-r}{\sigma_0^2}$ and have
\begin{equation}\label{pii01}
	\begin{split}
		\pi^*_I(t)=\tilde{\pi}_I(t)-\rho_0\sigma e^{-r(T-t)}\frac{\gamma\left(\alpha\underline{H}(t)+\hat{\alpha}\overline{H}(t)\right)+(2\alpha-1)\beta_0 A(t)}{\gamma+(2\alpha-1)\beta_0}.
	\end{split}
\end{equation}

\begin{figure}[htbp]
	\centering
	\begin{minipage}{0.5\textwidth}
		\centering
		\includegraphics[totalheight=5.5cm]{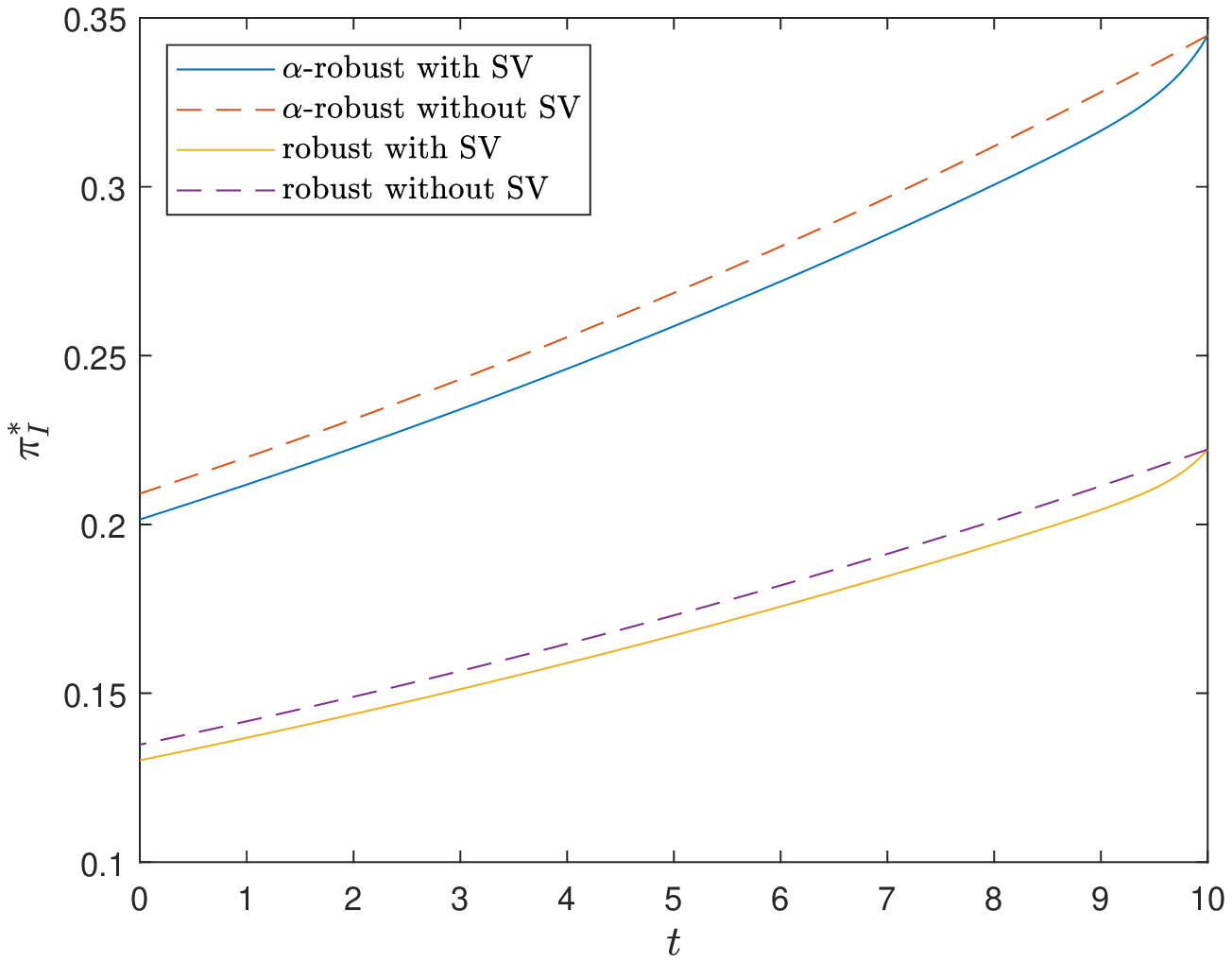}
		\caption{{Comparison of investment strategies for $\alpha$-robust and robust problems with SV and without SV.}}
		\label{fig:piicompare}
	\end{minipage}\hfill
	\begin{minipage}{0.5\textwidth}
		\centering
		\includegraphics[totalheight=5.5cm]{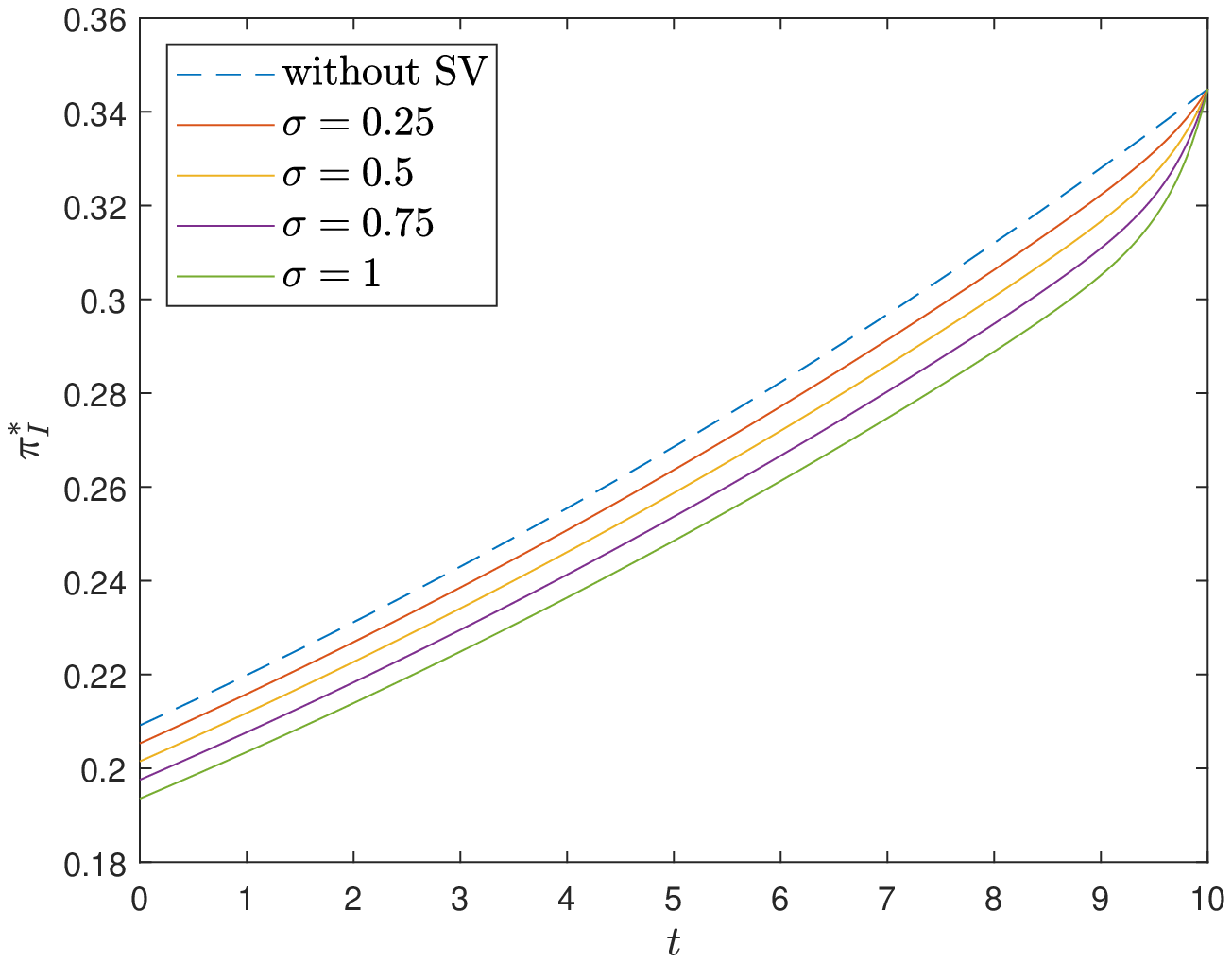}
		\caption{{Effect of $\sigma$ on the equilibrium investment strategy $\pi^*_I$.}}
		\label{fig:piisigma}
	\end{minipage}\hfill
\end{figure}

The comparisons of the equilibrium investment strategies {for $\alpha$-robust and robust problems} with stochastic volatility and without stochastic volatility are shown in Fig.~\ref{fig:piicompare}. Fig.~\ref{fig:piicompare} shows {{that}} the equilibrium investment strategy with stochastic volatility is smaller than that without stochastic volatility, which is because the AAI is risk aversion and ambiguity aversion. Meanwhile, the equilibrium investment strategy for the $\alpha$-robust problem is larger than that for the robust problem because the AAI with $\alpha$-maxmin preference is less ambiguity aversion than the AAI with maxmin robust preference.

Then we show the effect of the volatility $\sigma$ of the stochastic volatility $Y$ on the equilibrium investment strategy $\pi^*_I$ in Fig.~\ref{fig:piisigma}. The AAI lowers the investment in the stock when it is faced with a larger volatility risk. 
Fig.~\ref{fig:piisigma} shows that {the equilibrium investment strategy $\pi^*_I$ increases to the equilibrium investment strategy $\tilde{\pi}_I$ without SV when $\sigma$ declines to $0$, which is because the volatility $Y$ becomes deterministic if $\sigma$ declines to $0$.} 


\begin{figure}[htbp]
	\centering
	\begin{minipage}{0.5\textwidth}
		\centering
		\includegraphics[totalheight=5.5cm]{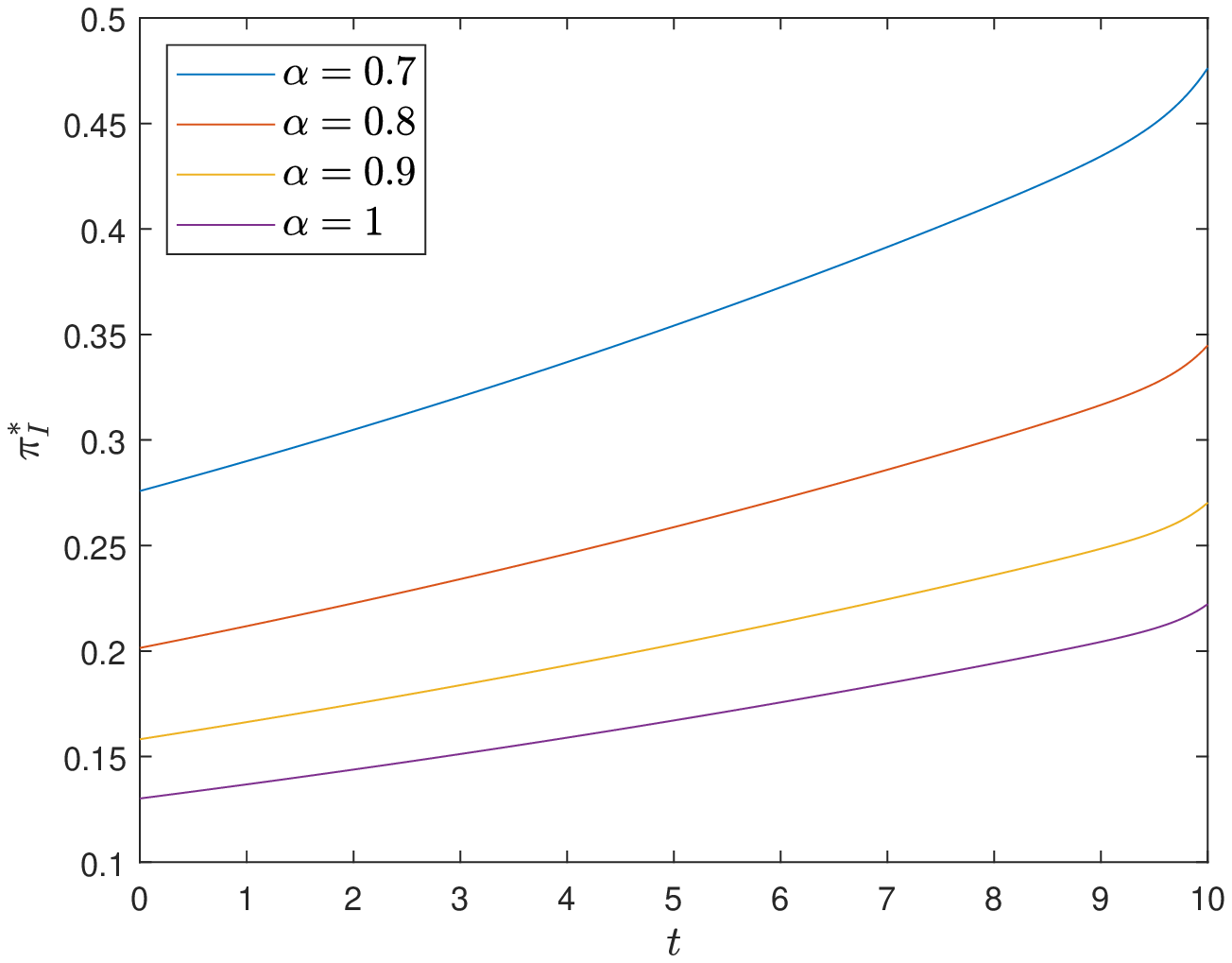}
		\caption{{Effect of $\alpha$ on the equilibrium investment strategy $\pi^*_I$.}}
		\label{fig:piialpha}
	\end{minipage}\hfill
	\begin{minipage}{0.5\textwidth}
		\centering
		\includegraphics[totalheight=5.5cm]{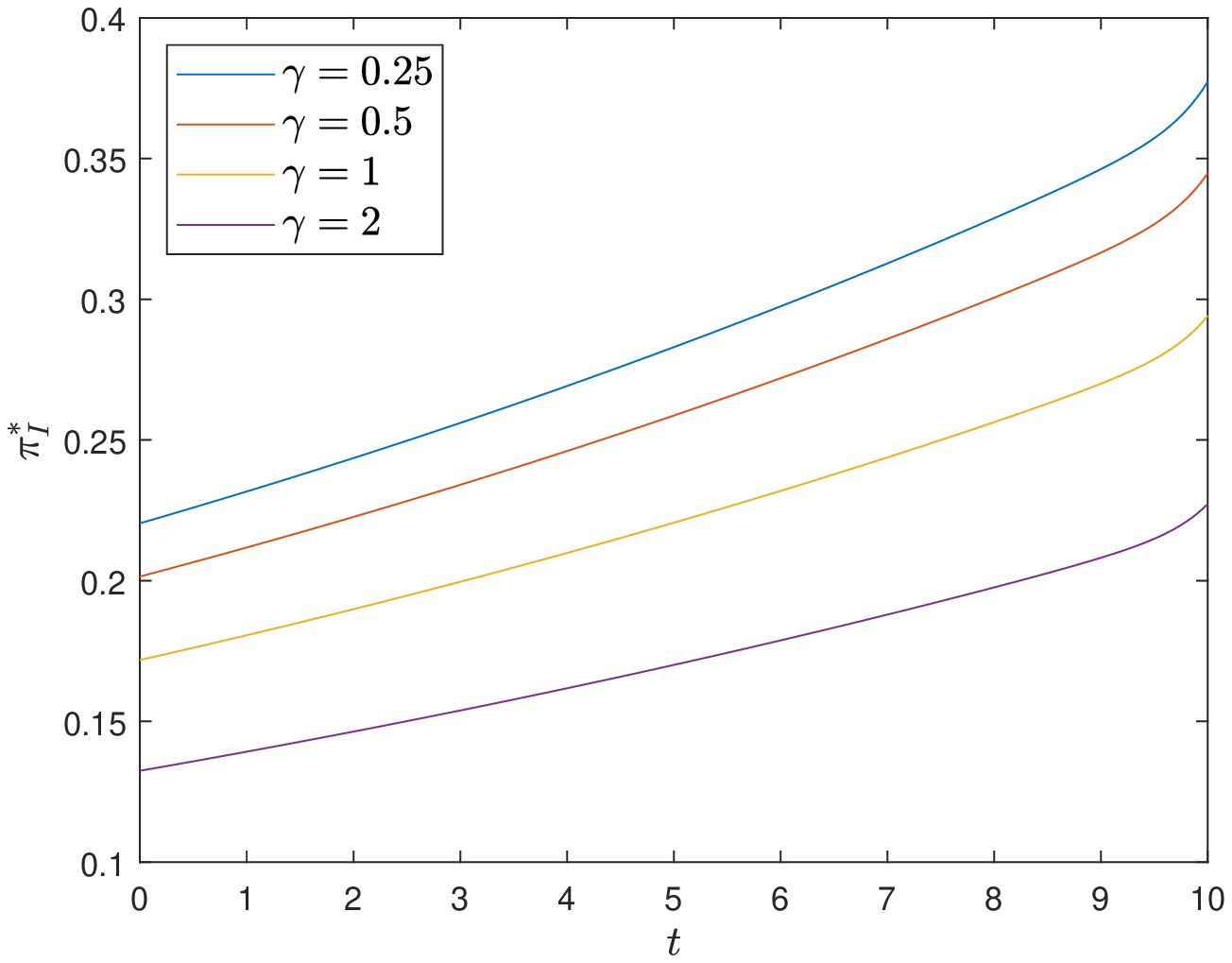}
		\caption{{Effect of $\gamma$ on the equilibrium investment strategy $\pi^*_I$.}}
		\label{fig:piigamma}
	\end{minipage}\hfill
	\begin{minipage}{0.5\textwidth}
		\centering
\includegraphics[totalheight=5.5cm]{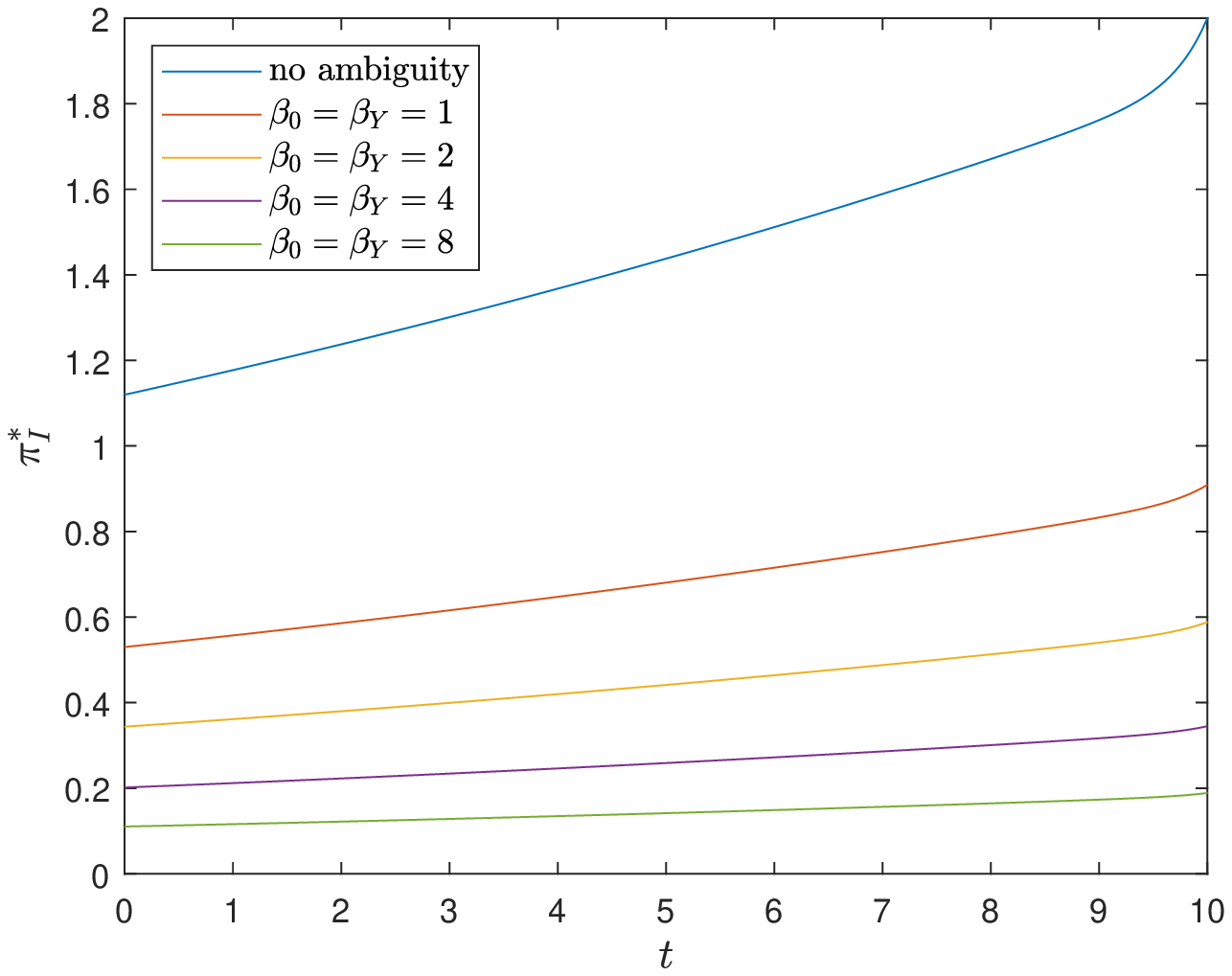}
\caption{{Effect of $\beta_0$ and $\beta_{Y}$ on the equilibrium investment strategy $\pi^*_I$.}}
\label{fig:piibeta0y}
\end{minipage}\hfill
		\begin{minipage}{0.5\textwidth}
	\centering
	\includegraphics[totalheight=5.5cm]{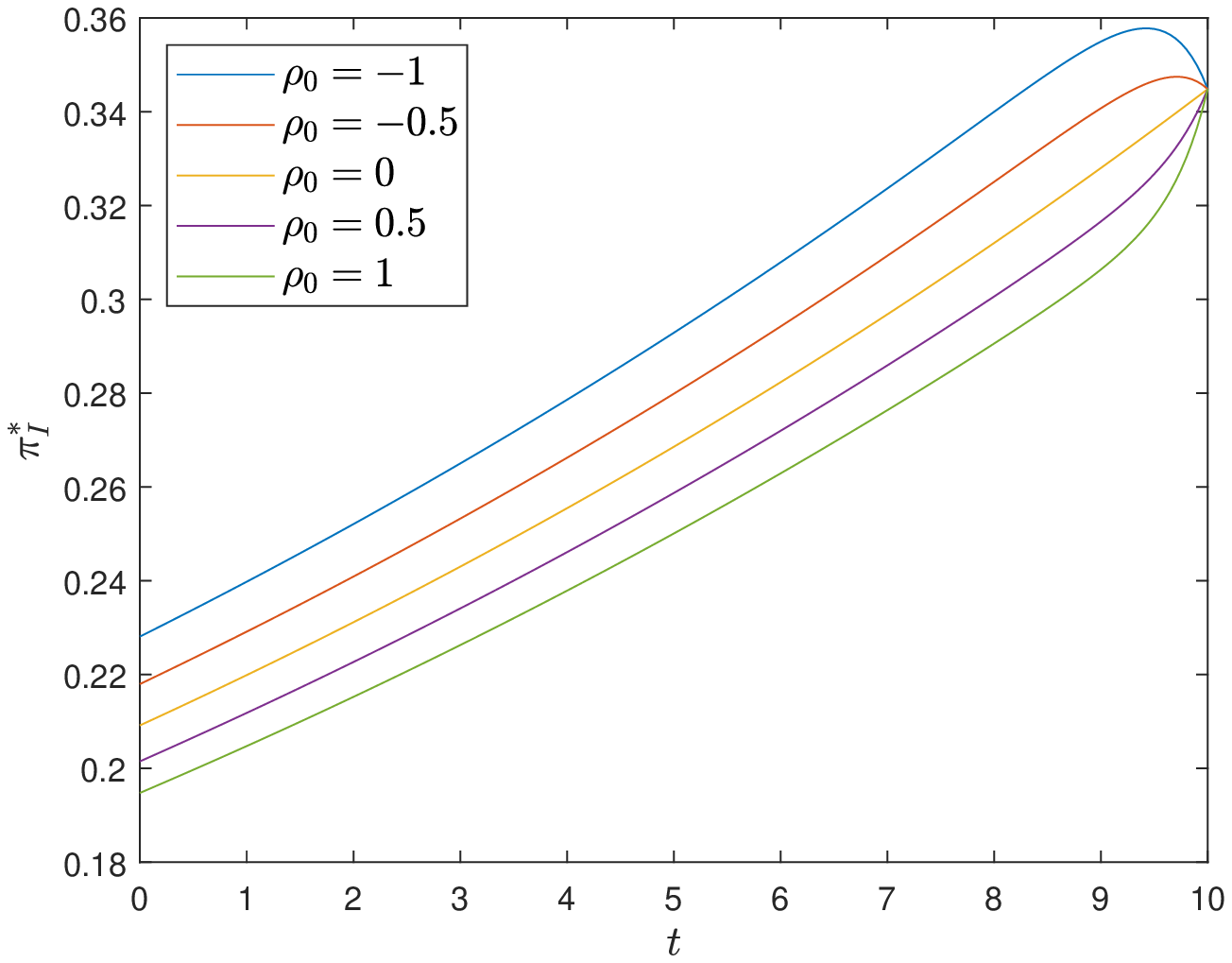}
	\caption{{Effect of $\rho_0$ on the equilibrium investment strategy $\pi_I^*$.}}
	\label{fig:piirho0}
\end{minipage}\hfill
\end{figure}
{The effects of the ambiguity aversion coefficient $\alpha$ and the risk aversion coefficient $\gamma$ on the equilibrium investment strategy $\pi^*_I$ are indicated in Figs.~\ref{fig:piialpha} and \ref{fig:piigamma}. When the AAI is more ambiguity aversion, the AAI tends to undertake less financial risk and the equilibrium investment strategy $\pi^*_I$ decreases with the ambiguity aversion coefficient $\alpha$. Meanwhile, when the AAI is more risk averse, the equilibrium investment strategy $\pi^*_I$ also decreases, which illustrates the negative relationship between $\gamma$ and $\pi^*_I$. Fig.~\ref{fig:piibeta0y} shows that the equilibrium investment strategy $\pi^*_I$ decreases with $(\beta, \beta_{Y})$, the level of ambiguity towards the financial market for the AAI, which means the AAI will be less interested in the financial market with a higher level of ambiguity. Figs.~\ref{fig:piialpha}-\ref{fig:piibeta0y} show that the level of ambiguity, ambiguity attitude, and risk attitude also have similar effects on the equilibrium investment strategy.}

{{Fig.~\ref{fig:piirho0} illustrates the effect of the correlation coefficient $\rho_0$ between the stock process and the stochastic volatility. Observing Fig.~\ref{fig:piirho0}, the equilibrium investment strategy $\pi_I^*$ decreases with $\rho_0$, which is because the financial market with larger $\rho_0$ contains more financial risk. Noting that, if $\rho_0$ becomes $0$, the equilibrium investment strategy $\pi_I^*$ degenerates into the equilibrium investment strategy $\tilde{\pi}_I$ without stochastic volatility. We observe from Fig.~\ref{fig:piirho0} that the AAI will invest more (less) than the case without stochastic volatility if the correlation coefficient $\rho_0$ is positive (negative).}}

\section{\bf Conclusion}
In this paper, we study the Stackelberg reinsurance game between the AAI and AAR under the $\alpha$-maxmin mean-variance criterion. The AAR acts as the leader, which can set the reinsurance premium. The AAI is the follower and purchases the per-loss reinsurance to diversify insurance risk. The AAI and AAR can also invest in a financial market with stochastic volatility. As the optimization problem is time-inconsistent, we study equilibrium strategies and present the extended HJB equations of the {$\alpha$-}robust mean-variance problem. The verification theorem is proved rigorously. The closed form of the AAI's feedback retention level $a_0$ to AAR's reinsurance premium $\eta$ is given. We find that the feedback form of the AAI with ambiguity aversion ($\alpha\in[0.5,1]$) does not depend on the L\'{e}vy measure $\nu(\mathrm{d}z)$, which means that the excess-of-loss reinsurance is always the optimal reinsurance strategy. The equilibrium investment strategies are obtained semi-explicitly, which are determined by a {system of} Riccati differential equations. The existence and uniqueness of the equilibrium retention level $a^*_0$ and reinsurance premium $\eta^*$ are also proved. Numerical results, show that {the level of ambiguity,} the ambiguity attitude and risk attitude of the AAI (AAR) have similar influences on the AAI and AAR's equilibrium strategies.

\vskip 15pt
{\bf Acknowledgements.}
The authors acknowledge the support from the National Natural Science Foundation of China (Grant  {No.12271290}, No.11901574 and  No.11871036). The authors also thank the members of the group of {Mathematical Finance and Actuarial Science}  at the Department of Mathematical Sciences, Tsinghua University for their feedbacks and useful conversations.

\renewcommand{\theequation}{\arabic{section}\arabic{equation}}
\numberwithin{equation}{section}
\setcounter{equation}{0}

\appendix
\section{Proof of Theorem \ref{THM4.1}}\label{A.1}
In order to prove Theorem \ref{THM4.1}, we first present the following lemma.
\begin{lemma}
For any deterministic strategy $u_I=\{(a(t),\pi_I(t)),  t\in[0,T]\}$, we have the following results:
\begin{enumerate}
\item The increments of the discounted wealth process $\{e^{-rt}X^u(t),t\in[0,T]\}$ are independent conditional on $\{Y(t),t\in[0,T]\}$ under any $\mathbb{Q}\in\mathcal{Q}$, i.e.,
\begin{equation}
\begin{split}
\mathbf{E}^{\phi}&[e^{r(T-w)}X^u(w)-e^{r(T-t)}X^u(t)|\mathcal{F}^X_t\vee\mathcal{F}^Y_t]\\
&=\mathbf{E}^{\phi}[e^{r(T-w)}X^u(w)-e^{r(T-t)}X^u(t)|Y(t)],~\forall w\in[t,T],~\phi\in\Theta;
\end{split}
\end{equation}

\item The function $\overline{J}^{u,\overline{\phi}^u}(t,x,y)-\underline{J}^{u,\underline{\phi}^u}(t,x,y)$ is independent of $x$.

\item For any $0\leq t\leq w\leq T$, the corresponding probability distortion functions attaining the infimum and the supremum in Eq.~(\ref{objecti}) are the same as the probability
distortion functions attaining the infimum in
\begin{equation}
\begin{split}
\inf_{\phi\in\Theta}&\left\{\mathbf{E}^{\phi}_{t,y}[e^{r(T-w)}X^u(w)-e^{r(T-t)}X^u(t)]+\mathbf{E}^\phi_{t,y}\left[\int_t^w h_{\mathbf{b}}(\Phi(s))\mathrm{d}s\right]\right.\\
&\left.-\frac{\gamma}{2}\rm{Var}^\phi_{t,y}[e^{r(T-w)}X^u(w)-e^{r(T-t)}X^u(t)]\right\}
\end{split}
\end{equation}
and the supremum in
\begin{equation}
\begin{split}
\sup_{\phi\in\Theta}&\left\{\mathbf{E}^{\phi}_{t,y}[e^{r(T-w)}X^u(w)-e^{r(T-t)}X^u(t)]-\mathbf{E}^\phi_{t,y}\left[\int_t^wh_{\mathbf{b}}(\Phi(s))\mathrm{d}s\right]\right.\\
&\left.-\frac{\gamma}{2}\rm{Var}^\phi_{t,y}[e^{r(T-w)}X^u(w)-e^{r(T-t)}X^u(t)]\right\}.
\end{split}
\end{equation}
\end{enumerate}
\end{lemma}

\begin{proof}
	We prove the three results of the lemma separately.

(1) Based on Eq.~(\ref{XIQ}), $\forall w\in[t,T]$, we have
\begin{equation}\label{itoa}
\begin{split}
e^{-rw}&X^u(T)-e^{-rt}X^u(t)\\
=&\int_t^w\int_0^\infty\left[(\theta-\eta(s))z+\eta(s)a(s,z)+\phi(s,z)a(s,z)\right]\nu(\dif z)\dif s\\
&-\int_t^w\int_0^\infty a(s,z)\tilde{N}^{\mathbb{Q}}(\dif s,\dif z)\\
&+\int_t^w\pi_I(t)\left[\left(\xi Y(s)
-\sqrt{Y(s)}\phi_0(s)\right)\dif s+\sqrt{Y(s)}\dif W(s)\right].
\end{split}
\end{equation}
The conditional independence follows from the assumption that $u$ is deterministic and $\phi$ only depends on $Y$.

(2) Subtracting Eq.~(\ref{JIm}) from Eq.~(\ref{JIp}) with $\underline{\phi}$, $\overline{\phi}\in \Theta$, we have
\begin{eqnarray*}
&&\overline{J}^{u,\overline{\phi}^u}(t,x,y)-\underline{J}^{u,\underline{\phi}^u}(t,x,y)\nonumber\\
&=&\mathbf{E}^{\overline{\phi}^u}_{t,x,y}[X^u(T)]-\mathbf{E}^{\underline{\phi}^u}_{t,x,y}[X^u(T)]
-\frac{\gamma}{2}{\rm{Var}}^{\overline{\phi}^u}_{t,x,y}[X^u(T)]+\frac{\gamma}{2}{\rm{Var}}^{\underline{\phi}^u}_{t,x,y}[X^u(T)]\nonumber\\
&+&\mathbf{E}^{\overline{\phi}^u}_{t,x,y}\left[\int_t^T h_{\mathbf{b}}\left(\overline{\Phi}^u(s)\right)\dif s\right]-\mathbf{E}^{\underline{\phi}^u}_{t,x,y}\left[
\int_t^T h_{\mathbf{b}}\left(\underline{\Phi}^u(s)\right)\dif s\right].
\end{eqnarray*}
 {{Using}} part (1), we know that $X^u(T)-e^{r(T-t)}X^u(t)$ is independent of $X^u(t)$ conditional on $Y(t)$ under both $\mathbb{Q}^{\overline{\phi}^u}$ and $\mathbb{Q}^{\underline{\phi}^u}$. Thus,
\begin{eqnarray*}
&&\overline{J}^{u,\overline{\phi}^u}(t,x,y)-\underline{J}^{u,\underline{\phi}^u}(t,x,y)\\
&=&\mathbf{E}^{\overline{\phi}^u}_{t,x,y}[X^u(T)-e^{r(T-t)}X^u(t)]-\mathbf{E}^{\underline{\phi}^u}_{t,x,y}[X^u(T)
-e^{r(T-t)}X^u(t)]\\
&&-\frac{\gamma}{2}{\rm{Var}}^{\overline{\phi}^u}_{t,x,y}[X^u(T)-e^{r(T-t)}X^u(t)]
+\frac{\gamma}{2}{\rm{Var}}^{\underline{\phi}^u}_{t,x,y}[X^u(T)-e^{r(T-t)}X^u(t)]\\
&&+\mathbf{E}^{\overline{\phi}^u}_{t,x,y}\left[\int_t^T h_{\mathbf{b}}\left(\overline{\Phi}^u(s)\right)\dif s\right]-\mathbf{E}^{\underline{\phi}^u}_{t,x,y}\left[\int_t^T h_{\mathbf{b}}\left(\underline{\Phi}^u(s)\right)\dif s\right]\\
&=&\mathbf{E}^{\overline{\phi}^u}_{t,y}[X^u(T)-e^{r(T-t)}X^u(t)]
-\mathbf{E}^{\underline{\phi}^u}_{t,y}[X^u(T)-e^{r(T-t)}X^u(t)]\\
&&-\frac{\gamma}{2}{\rm{Var}}^{\overline{\phi}^u}_{t,y}[X^u(T)-e^{r(T-t)}X^u(t)]
+\frac{\gamma}{2}{\rm{Var}}^{\underline{\phi}^u}_{t,y}[X^u(T)-e^{r(T-t)}X^u(t)]\\
&&+\mathbf{E}^{\overline{\phi}^u}_{t,y}\left[\int_t^T h_{\mathbf{b}}\left(\overline{\Phi}^u(s)\right)\dif s\right]
-\mathbf{E}^{\underline{\phi}^u}_{t,y}\left[\int_t^T h_{\mathbf{b}}\left(\underline{\Phi}^u(s)\right)\dif s\right],
\end{eqnarray*}
which is independent of $x$.

(3) It follows from Eq.~(\ref{itoa}) that
\begin{equation*}
\begin{split}
&\mathbf{E}^{\phi}_{t,x,y}\left[X^u(T)\right]-\frac{\gamma}{2}{\rm{Var}}^{\phi}_{t,x,y}\left[X^u(T)\right]\pm\mathbf{E}^\phi_{t,x,y}\left[\int_t^T h_{\mathbf{b}}(\Phi(s))\dif s\right]\\
=&e^{r(T-t)}x+\mathbf{E}^{\phi}_{t,x,y}\left[X^u(T)-e^{r(T-t)}X^u(t)\right]\\
&-\frac{\gamma}{2}{\rm{Var}}^{\phi}_{t,x,y}\left[X^u(T)-e^{r(T-t)}X^u(t)\right]\pm\mathbf{E}^\phi_{t,x,y}\left[\int_t^T h_{\mathbf{b}}(\Phi(s))\dif s\right]\\
=&e^{r(T-t)}x+\mathbf{E}^{\phi}_{t,y}\left[X^u(T)-e^{r(T-t)}X^u(t)\right]\\
&-\frac{\gamma}{2}{\rm{Var}}^{\phi}_{t,y}\left[X^u(T)-e^{r(T-t)}X^u(t)\right]\pm\mathbf{E}^\phi_{t,y}\left[\int_t^T h_{\mathbf{b}}(\Phi(s))\dif s\right]\\
=&e^{r(T-t)}x+\mathbf{E}^{\phi}_{t,y}\left\{\mathbf{E}^\phi\left[X^u(T)-e^{r(T-t)}X^u(t)|\mathcal{F}^Y_w\right]\right\}\\
&-\frac{\gamma}{2}\mathbf{E}^{\phi}_{t,y}\left\{{\rm{Var}}^{\phi}\left[X^u(T)-e^{r(T-t)}X^u(t)|\mathcal{F}^Y_w\right]\right\}\\
&-\frac{\gamma}{2}{\rm{Var}}^{\phi}_{t,y}\left\{\mathbf{E}^\phi\left[X^u(T)-e^{r(T-t)}X^u(t)|\mathcal{F}^Y_w\right]\right\}\pm\mathbf{E}^\phi_{t,y}\left[\int_t^T h_{\mathbf{b}}(\Phi(s))\dif s\right]\\
=&\mathbf{E}^{\phi}_{t,y}\left[e^{r(T-w)}X^u(w)-e^{r(T-t)}X^u(t)\right]+\mathbf{E}^{\phi}_{t,y}\left\{\mathbf{E}^\phi\left[X^u(T)-e^{r(T-w)}X^u(w)|\mathcal{F}^Y_w\right]\right\}\\
&+e^{r(T-t)}x-\frac{\gamma}{2}\mathbf{E}^{\phi}_{t,y}\left\{{\rm{Var}}^{\phi}\left[X^u(T)-e^{r(T-w)}X^u(w)|\mathcal{F}^Y_w\right]\right\}\\
&-\frac{\gamma}{2}{\rm{Var}}^{\phi}_{t,y}\left[e^{r(T-w)}X^u(w)-e^{r(T-t)}X^u(t)\right]\pm\mathbf{E}^\phi_{t,y}\left[\int_t^T h_{\mathbf{b}}(\Phi(s))\dif s\right]\\
\end{split}
\end{equation*}
\begin{equation*}
\begin{split}
=&e^{r(T-t)}x+\mathbf{E}^{\phi}_{t,y}\left[e^{r(T-w)}X^u(w)-e^{r(T-t)}X^u(t)\right]\\
&-\frac{\gamma}{2}{\rm{Var}}^{\phi}_{t,y}\left[e^{r(T-w)}X^u(w)-e^{r(T-t)}X^u(t)\right]\pm\mathbf{E}^\phi_{t,y}\left[\int_t^w h_{\mathbf{b}}(\Phi(s))\dif s\right]\\
&+\mathbf{E}^{\phi}_{t,y}\left\{\mathbf{E}^\phi\left[X^u(T)-e^{r(T-s)}X^u(w)|Y(w)\right]\pm\mathbf{E}^\phi\left[\int_w^T h_{\mathbf{b}}(\Phi(s))\dif s|Y(w)\right]\right.\\
&\left.-\frac{\gamma}{2}{\rm{Var}}^{\phi}\left[X^u(T)-e^{r(T-w)}X^u(w)|Y(w)\right]\right\}.
\end{split}
\end{equation*}
As such, the corresponding probability distortion functions $\underline{\phi}^u(t)$ and $\overline{\phi}^u(t)$ only depend on $t$, $Y(t)$ and $u(t)$. Then the lemma is proved.
\end{proof}

Next, we prove the verification theorem. The proof is divided into two steps.
\begin{proof}
First, we show that under the conditions of Theorem \ref{THM4.1},
\begin{equation}\label{ge}
\underline{g}(t,x,y)=\mathbf{E}^{\underline{\phi}}_{t,x,y}\left[X^{u^*}(T)\right],~\overline{g}(t,x,y)=\mathbf{E}^{\overline{\phi}}_{t,x,y}\left[X^{u^*}(T)\right],
\end{equation}
and
\begin{equation}\label{vj}
V(t,x,y)=J^{u^*}(t,x).
\end{equation}
Based on Eq.~(\ref{geni}) and  the second equality of Eq.~(\ref{bondi}), we have
\begin{equation}
\begin{split}
\mathbf{E}^{\underline{\phi}^*}_{t,x,y}[\underline{g}(T,X^{u^*}(T))]&=\underline{g}(t,x,y)+\mathbf{E}^{\underline{\phi}^*}_{t,x,y}\left[\int_t^T\mathcal{A}^{u^*,\underline{\phi}^*}\underline{g}(s,X^{u^*}(s),Y(s))\dif s\right]\\
&=\underline{g}(t,x,y).
\end{split}
\end{equation}
Besides, it follows from the third equality of Eq.~(\ref{bondi}) that
\begin{equation}
\underline{g}(t,x,y)=\mathbf{E}^{\underline{\phi}^*}_{t,x,y}[\underline{g}(T,X^{u^*}(T))]=\mathbf{E}^{\underline{\phi}}_{t,x,y}\left[X^{u^*}(T)\right].
\end{equation}
The same arguments imply $\overline{g}(t,x,y)=\mathbf{E}^{\overline{\phi}}_{t,x,y}\left[X^{u^*}(T)\right]$.

Next, we show $V(t,x,y)=J^{u^*}_\alpha(t,x,y)$. Because the optimal values in Eq.~(\ref{veri}) are achieved at $(u^*,\underline{\phi}^*,\overline{\phi}^*)$, by the second equalities of Eq.~(\ref{bondi}), we rewrite Eq.~(\ref{veri}) as
\begin{equation}\label{hjb1}
\begin{split}
0=&\alpha\left\{\mathcal{A}^{u^*,\underline{\phi}^*}V(t,x,y)-\frac{\gamma}{2}\mathcal{A}^{u^*,\underline{\phi}^*}\underline{g}^2(t,x.y)+h_{\mathbf{b}}\left(\underline{\phi}^*(t)\right)\right\}\\
+&\hat{\alpha}\left\{\mathcal{A}^{u^*,\overline{\phi}^*}V(t,x,y)-\frac{\gamma}{2}\mathcal{A}^{u^*,\overline{\phi}^*}\overline{g}^2(t,x,y)+h_{\mathbf{b}}\left(\overline{\phi}^*(t)\right)\right\}.
\end{split}
\end{equation}
By the first equality of Eq.~(\ref{bondi}), we have
\begin{equation}
\begin{split}
\mathbf{E}^{\underline{\phi}^*}_{t,x,y}[X^{u^*}(T)]&=\mathbf{E}^{\underline{\phi}^*}_{t,x,y}[V(T,X^{u^*}(T),Y(T))]\\
&=V(t,x,y)+\mathbf{E}^{\underline{\phi}^*}_{t,x,y}\left[\int_t^T\mathcal{A}^{u^*,\underline{\phi}^*}V(s,X^{u^*}(s),Y(s))\dif s\right],\\
\mathbf{E}^{\overline{\phi}^*}_{t,x,y}[X^{u^*}(T)]&=\mathbf{E}^{\overline{\phi}^*}_{t,x,y}[V(T,X^{u^*}(T),Y(T))]\\
&=V(t,x,y)+\mathbf{E}^{\overline{\phi}^*}_{t,x,y}\left[\int_t^T\mathcal{A}^{u^*,\overline{\phi}^*}V(s,X^{u^*}(s),Y(s))\dif s\right].\\
\end{split}
\end{equation}
A linear combination of the above two equations yields
\begin{equation}
\begin{split}
\alpha\mathbf{E}^{\underline{\phi}^*}_{t,x,y}[X^{u^*}(T)]&+\hat{\alpha}\mathbf{E}^{\overline{\phi}^*}_{t,x,y}[X^{u^*}(T)]\\
=&V(t,x,y)+\alpha\mathbf{E}^{\underline{\phi}^*}_{t,x,y}\left[\int_t^T\mathcal{A}^{u^*,\underline{\phi}^*}V(s,X^{u^*}(s),Y(s))\dif s\right]\\
&+\hat{\alpha}\mathbf{E}^{\overline{\phi}^*}_{t,x,y}\left[\int_t^T\mathcal{A}^{u^*,\overline{\phi}^*}V(s,X^{u^*}(s),Y(s))\dif s\right].
\end{split}
\end{equation}
Substituting Eq.~(\ref{hjb1}) into the last equation, we have
\begin{equation}\label{valuever}
\begin{split}
V(t,x,y)=\alpha&\left\{\mathbf{E}^{\underline{\phi}^*}_{t,x,y}[X^{u^*}(T)]+\mathbf{E}^{\underline{\phi}^*}_{t,x,y}\left[\int_t^Th_{\mathbf{b}}\left(\underline{\Phi}^*(t)\right)\dif s\right]\right.\\
&\left.-\frac{\gamma}{2}\mathbf{E}^{\underline{\phi}^*}_{t,x,y}\left[\int_t^T\mathcal{A}^{u^*,\underline{\phi}^*}\underline{g}(s,X^{u^*}(s),Y(s))\dif s\right]\right\}\\
+\hat{\alpha}&\left\{\mathbf{E}^{\overline{\phi}^*}_{t,x,y}[X^{u^*}(T)]-\mathbf{E}^{\overline{\phi}^*}_{t,x,y}\left[\int_t^Th_{\mathbf{b}}\left(\overline{\Phi}^*(t)\right)\dif s\right]\right.\\
&\left.-\frac{\gamma}{2}\mathbf{E}^{\overline{\phi}^*}_{t,x,y}\left[\int_t^T\mathcal{A}^{u^*,\overline{\phi}^*}\overline{g}(s,X^{u^*}(s),Y(s))\dif s\right]\right\}.
\end{split}
\end{equation}
Then we obtain the relationship among  $\mathbf{E}^{\underline{\phi}^*}_{t,x,y}\left[\int_t^T\mathcal{A}^{u^*,\underline{\phi}^*}\underline{g}(s,X^{u^*}(s),Y(s))\dif s\right]$,\\

$\mathbf{E}^{\overline{\phi}^*}_{t,x,y}\left[\int_t^T\mathcal{A}^{u^*,\overline{\phi}^*}\overline{g}(s,X^{u^*}(s),Y(s))\dif s\right]$ and the variance of $X^{u^*}(T)$,
\begin{equation}\label{varcal}
\begin{split}
\alpha\mathbf{E}^{\underline{\phi}^*}_{t,x,y}&[X^{u^*}(T)^2]+\hat{\alpha}\mathbf{E}^{\overline{\phi}^*}_{t,x,y}[X^{u^*}(T)^2]\\
=&\alpha\mathbf{E}^{\underline{\phi}^*}_{t,x,y}[\underline{g}^2(T,X^{u^*}(T),Y(T))]+\hat{\alpha}\mathbf{E}^{\overline{\phi}^*}_{t,x,y}[\overline{g}^2(T,X^{u^*}(T),Y(T))]\\
=&\alpha\underline{g}^2(t,x,y)+\alpha\mathbf{E}^{\underline{\phi}^*}_{t,x,y}\left[\int_t^T\mathcal{A}^{u^*,\underline{\phi}^*}\underline{g}(s,X^{u^*}(s),Y(s))\dif s\right]+\hat{\alpha}\overline{g}^2(t,x,y)\\
&+\hat{\alpha}\mathbf{E}^{\overline{\phi}^*}_{t,x,y}\left[\int_t^T\mathcal{A}^{u^*,\overline{\phi}^*}\overline{g}(s,X^{u^*}(s),Y(s))\dif s\right]\\
=&\alpha\left(\mathbf{E}^{\underline{\phi}^*}_{t,x,y}[X^{u^*}(T)]\right)^2+\alpha\mathbf{E}^{\underline{\phi}^*}_{t,x,y}\left[\int_t^T\mathcal{A}^{u^*,\underline{\phi}^*}\underline{g}(s,X^{u^*}(s),Y(s))\dif s\right]\\
&+\hat{\alpha}\left(\mathbf{E}^{\overline{\phi}^*}_{t,x,y}[X^{u^*}(T)]\right)^2+\hat{\alpha}\mathbf{E}^{\overline{\phi}^*}_{t,x,y}\left[\int_t^T\mathcal{A}^{u^*,\overline{\phi}^*}\overline{g}(s,X^{u^*}(s),Y(s))\dif s\right],
\end{split}
\end{equation}
which is equivalent to
\begin{equation}\label{var}
\begin{split}
&\alpha {\rm{Var}}^{\underline{\phi}^*}_{t,x,y}[X^{u^*}(T)]+\hat{\alpha}{\rm{Var}}^{\overline{\phi}^*}_{t,x,y}[X^{u^*}(T)]\\
=&\alpha\mathbf{E}^{\underline{\phi}^*}_{t,x,y}\!\left[\int_t^T\!\!\mathcal{A}^{u^*\!,\underline{\phi}^*}\underline{g}(s,X^{u^*}\!(s),Y(s))\dif s\right]\!\!+\!\hat{\alpha}\mathbf{E}^{\overline{\phi}^*}_{t,x,y}\!\left[\int_t^T\!\!\mathcal{A}^{u^*\!,\overline{\phi}^*}\overline{g}(s,X^{u^*}\!(s),Y(s))\dif s\right]\!\!.
\end{split}
\end{equation}
Finally, substituting Eq.~(\ref{var}) into Eq.~(\ref{valuever}) yields
\begin{equation}
\begin{split}
&V(t,x,y)\\
=&\alpha\left\{\mathbf{E}^{\underline{\phi}^*}_{t,x,y}[X^{u^*}(T)]-\frac{\gamma}{2}{\rm{Var}}^{\underline{\phi}^*}_{t,x,y}[X^{u^*}(T)]+\mathbf{E}^{\underline{\phi}^*}_{t,x,y}\left[\int_t^Th_{\mathbf{b}}\left(\underline{\Phi}^*(t)\right)\dif s\right]\right\}\\
&+\hat{\alpha}\left\{\mathbf{E}^{\overline{\phi}^*}_{t,x,y}[X^{u^*}(T)]-\frac{\gamma}{2}{\rm{Var}}^{\overline{\phi}^*}_{t,x,y}[X^{u^*}(T)]+\mathbf{E}^{\overline{\phi}^*}_{t,x,y}\left[\int_t^Th_{\mathbf{b}}\left(\overline{\Phi}^*(t)\right)\dif s\right]\right\}\\
=&\alpha\left\{\mathbf{E}^{\underline{\phi}^{u^*}}_{t,x,y}[X^{u^*}(T)]-\frac{\gamma}{2}{\rm{Var}}^{\underline{\phi}^{u^*}}_{t,x,y}[X^{u^*}(T)]+\mathbf{E}^{\underline{\phi}^*}_{t,x,y}\left[\int_t^Th_{\mathbf{b}}\left(\underline{\Phi}^{u^*}(t)\right)\dif s\right]\right\}\\
&+\hat{\alpha}\left\{\mathbf{E}^{\overline{\phi}^{u*}}_{t,x,y}[X^{u^*}(T)]-\frac{\gamma}{2}{\rm{Var}}^{\overline{\phi}^{u*}}_{t,x,y}[X^{u^*}(T)]+\mathbf{E}^{\overline{\phi}^*}_{t,x,y}\left[\int_t^Th_{\mathbf{b}}\left(\overline{\Phi}^{u*}(t)\right)\dif s\right]\right\}\\
=&J^{u^*}_\alpha(t,x,y).
\end{split}
\end{equation}


Second, we show that $u^*$ is an equilibrium strategy. Consider the perturbed strategy 
\begin{equation}
u_\epsilon(s)=
\begin{cases}
\tilde{u}, &s\in[t,t+\epsilon),\\
u^*(s), &s\in[t+\epsilon,T],
\end{cases}
\end{equation}
where $\tilde{u}=(\tilde{a},\tilde{\pi})\in\Pi$, $t\in[0,T]$, and $\epsilon>0$.

In order to show that $u^*$ is an equilibrium strategy, we know from Definition \ref{equilibrium} that it suffices to show
\begin{equation}
J^{u_\epsilon}_\alpha(t,x,y)-J^{u^*}_\alpha(t,x,y)\leq o(\epsilon).
\end{equation}
In the following, we first  derive an expression of $J^{u_\epsilon}_\alpha(t,x,y)-J^{u^*}_\alpha(t,x,y)$ and then show that it is bounded above by $o(\epsilon)$. Based on Eq.~(\ref{JIp}), we have
\begin{equation*}
\begin{split}
&\underline{J}^{u_\epsilon,\underline{\phi}^{u_\epsilon}}(t,x,y)\\
=&\mathbf{E}^{\underline{\phi}^{u_\epsilon}}_{t,x,y}\left[
X^{u_\epsilon}(T)-\frac{\gamma}{2}X^{u_\epsilon}(T)^2\right]+\frac{\gamma}{2}\left(\mathbf{E}^{\underline{\phi}^{u_\epsilon}}_{t,x,y}\left[
X^{u_\epsilon}(T)\right]\right)^2
+\mathbf{E}^{\underline{\phi}^{u_\epsilon}}_{t,x,y}\left[\int_t^T h_{\mathbf{b}}\left(\underline{\Phi}^{u_\epsilon}(s)\right)\dif s\right]\\
=&\mathbf{E}^{\underline{\phi}^{\tilde{u}}}_{t,x,y}\left[\mathbf{E}^{\underline{\phi}^{*}}_{t+\epsilon,X^{\tilde{u}}(t+\epsilon),Y(t+\epsilon)}\left[
X^{u^*}(T)-\frac{\gamma}{2}X^{u^*}(T)^2\right]\right]\\
&+\frac{\gamma}{2}\left(\mathbf{E}^{\underline{\phi}^{\tilde{u}}}_{t,x,y}\left[\mathbf{E}^{\underline{\phi}^{*}}_{t+\epsilon,X^{\tilde{u}}(t+\epsilon),Y(t+\epsilon)}\left[
X^{u^*}(T)\right]\right]\right)^2\\
&+\mathbf{E}^{\underline{\phi}^{\tilde{u}}}_{t,x,y}\left[\int_t^{t+\epsilon} h_{\mathbf{b}}\left(\underline{\Phi}^{\tilde{u}}(s)\right)\dif s+\int_{t+\epsilon}^T h_{\mathbf{b}}\left(\underline{\Phi}^{*}(s)\right)\dif s\right]\\
=&\mathbf{E}^{\underline{\phi}^{\tilde{u}}}_{t,x,y}\left[\underline{J}^{u^*,\underline{\phi}^{*}}(t+\epsilon,X^{\tilde{u}}(t+\epsilon),Y(t+\epsilon))\right]-\frac{\gamma}{2}\mathbf{E}^{\underline{\phi}^{\tilde{u}}}_{t,x,y}\left[\left(\mathbf{E}^{\underline{\phi}^{*}}_{t+\epsilon,X^{\tilde{u}}(t+\epsilon),Y(t+\epsilon)}\left[
X^{u^*}(T)\right]\right)^2\right]\\
&+\frac{\gamma}{2}\left(\mathbf{E}^{\underline{\phi}^{\tilde{u}}}_{t,x,y}\left[\mathbf{E}^{\underline{\phi}^{*}}_{t+\epsilon,X^{\tilde{u}}(t+\epsilon),Y(t+\epsilon)}\left[
X^{u^*}(T)\right]\right]\right)^2+\mathbf{E}^{\underline{\phi}^{\tilde{u}}}_{t,x,y}\left[\int_t^{t+\epsilon} h_{\mathbf{b}}\left(\underline{\Phi}^{\tilde{u}}(s)\right)\dif s\right]\\
=&\mathbf{E}^{\underline{\phi}^{\tilde{u}}}_{t,x,y}\left[\underline{J}^{u^*,\underline{\phi}^{*}}(t+\epsilon,X^{\tilde{u}}(t+\epsilon),Y(t+\epsilon))\right]-\frac{\gamma}{2}\mathbf{E}^{\underline{\phi}^{\tilde{u}}}_{t,x,y}\left[\underline{g}^2(t+\epsilon,X^{\tilde{u}}(t+\epsilon),Y(t+\epsilon))\right]\\
&+\frac{\gamma}{2}\left(\mathbf{E}^{\underline{\phi}^{\tilde{u}}}_{t,x,y}\left[\underline{g}(t+\epsilon,X^{\tilde{u}}(t+\epsilon),Y(t+\epsilon))\right]\right)^2+\mathbf{E}^{\underline{\phi}^{\tilde{u}}}_{t,x,y}\left[\int_t^{t+\epsilon} h_{\mathbf{b}}\left(\underline{\Phi}^{\tilde{u}}(s)\right)\dif s\right],
\end{split}
\end{equation*}
where the last equality holds as Eq.~(\ref{ge}). And the same arguments lead to
\begin{equation*}
\begin{split}
&\overline{J}^{u_\epsilon,\overline{\phi}^{u_\epsilon}}(t,x,y)\\
=&\mathbf{E}^{\overline{\phi}^{\tilde{u}}}_{t,x,y}\left[\overline{J}^{u^*,\overline{\phi}^{*}}(t+\epsilon,X^{\tilde{u}}(t+\epsilon),Y(t+\epsilon))\right]-\frac{\gamma}{2}\mathbf{E}^{\overline{\phi}^{\tilde{u}}}_{t,x,y}\left[\overline{g}^2(t+\epsilon,X^{\tilde{u}}(t+\epsilon),Y(t+\epsilon))\right]\\
&+\frac{\gamma}{2}\left(\mathbf{E}^{\overline{\phi}^{\tilde{u}}}_{t,x,y}\left[\overline{g}(t+\epsilon,X^{\tilde{u}}(t+\epsilon),Y(t+\epsilon))\right]\right)^2+\mathbf{E}^{\overline{\phi}^{\tilde{u}}}_{t,x,y}\left[\int_t^{t+\epsilon} h_{\mathbf{b}}\left(\overline{\Phi}^{\tilde{u}}(s)\right)\dif s\right].
\end{split}
\end{equation*}
As such
\begin{equation}\label{jeps}
\begin{split}
&J^{u_\epsilon}_\alpha(t,x,y)\\
=&\alpha \underline{J}^{u_\epsilon,\underline{\phi}^{u_\epsilon}}(t,x,y)+\hat{\alpha}\overline{J}^{u_\epsilon,\overline{\phi}^{u_\epsilon}}(t,x,y)\\
=&\alpha \mathbf{E}^{\underline{\phi}^{\tilde{u}}}_{t,x,y}\left[\underline{J}^{u^*,\underline{\phi}^{*}}(t+\epsilon,X^{\tilde{u}}(t+\epsilon),Y(t+\epsilon))\right]+\hat{\alpha}\mathbf{E}^{\overline{\phi}^{\tilde{u}}}_{t,x,y}\left[\overline{J}^{u^*,\overline{\phi}^{*}}(t+\epsilon,X^{\tilde{u}}(t+\epsilon),Y(t+\epsilon))\right]\\
&-\frac{\alpha\gamma}{2}\mathbf{E}^{\underline{\phi}^{\tilde{u}}}_{t,x,y}\left[\underline{g}^2(t+\epsilon,X^{\tilde{u}}(t+\epsilon),Y(t+\epsilon))\right]+\frac{\alpha\gamma}{2}\left(\mathbf{E}^{\underline{\phi}^{\tilde{u}}}_{t,x,y}\left[\underline{g}(t+\epsilon,X^{\tilde{u}}(t+\epsilon),Y(t+\epsilon))\right]\right)^2\\
&-\frac{\hat{\alpha}\gamma}{2}\mathbf{E}^{\overline{\phi}^{\tilde{u}}}_{t,x,y}\left[\overline{g}^2(t+\epsilon,X^{\tilde{u}}(t+\epsilon),Y(t+\epsilon))\right]+\frac{\hat{\alpha}\gamma}{2}\left(\mathbf{E}^{\overline{\phi}^{\tilde{u}}}_{t,x,y}\left[\overline{g}(t+\epsilon,X^{\tilde{u}}(t+\epsilon),Y(t+\epsilon))\right]\right)^2\\
&+\alpha\mathbf{E}^{\underline{\phi}^{\tilde{u}}}_{t,x,y}\left[\int_t^{t+\epsilon} h_{\mathbf{b}}\left(\underline{\Phi}^{\tilde{u}}(s)\right)\dif s\right]-\hat{\alpha}\mathbf{E}^{\overline{\phi}^{\tilde{u}}}_{t,x,y}\left[\int_t^{t+\epsilon} h_{\mathbf{b}}\left(\overline{\Phi}^{\tilde{u}}(s)\right)\dif s\right].
\end{split}
\end{equation}
 {{Subtracting}} $J^{u^*}_\alpha(t,x,y)$ to both sides of Eq.~(\ref{jeps}), we have
\begin{equation*}
\begin{split}
J^{u_\epsilon}_\alpha(t,x,y)-J^{u^*}_\alpha(t,x,y)&=J^{u_\epsilon}_\alpha(t,x,y)-\alpha\underline{J}^{u^*,\underline{\phi}^{*}}(t,x,y)-\hat{\alpha}\overline{J}^{u^*,\overline{\phi}^{*}}(t,x,y)\\
&\triangleq H_\epsilon,
\end{split}
\end{equation*}
where
\begin{equation*}
\begin{split}
H_\epsilon\triangleq&\alpha\left\{\mathbf{E}^{\underline{\phi}^{\tilde{u}}}_{t,x,y}\left[\underline{J}^{u^*,\underline{\phi}^{*}}(t+\epsilon,X^{\tilde{u}}(t+\epsilon),Y(t+\epsilon))\right]-\underline{J}^{u^*,\underline{\phi}^*}(t,x,y)\right\}\\
&+\hat{\alpha}\left\{\mathbf{E}^{\overline{\phi}^{\tilde{u}}}_{t,x,y}\left[\overline{J}^{u^*,\overline{\phi}^{*}}(t+\epsilon,X^{\tilde{u}}(t+\epsilon),Y(t+\epsilon))\right]-\overline{J}^{u^*,\overline{\phi}^*}(t,x,y)\right\}\\
&-\frac{\alpha\gamma}{2}\left\{\mathbf{E}^{\underline{\phi}^{\tilde{u}}}_{t,x,y}\left[\underline{g}^2(t+\epsilon,X^{\tilde{u}}(t+\epsilon),Y(t+\epsilon))\right]-\underline{g}^2(t,x,y)\right\}\\
&+\frac{\alpha\gamma}{2}\left\{\left(\mathbf{E}^{\underline{\phi}^{\tilde{u}}}_{t,x,y}\left[\underline{g}(t+\epsilon,X^{\tilde{u}}(t+\epsilon),Y(t+\epsilon))\right]\right)^2-\underline{g}^2(t,x,y)\right\}\\
&-\frac{\hat{\alpha}\gamma}{2}\left\{\mathbf{E}^{\overline{\phi}^{\tilde{u}}}_{t,x,y}\left[\overline{g}^2(t+\epsilon,X^{\tilde{u}}(t+\epsilon),Y(t+\epsilon))\right]-\underline{g}^2(t,x,y)\right\}\\
&+\frac{\alpha\gamma}{2}\left\{\left(\mathbf{E}^{\underline{\phi}^{\tilde{u}}}_{t,x,y}\left[\overline{g}(t+\epsilon,X^{\tilde{u}}(t+\epsilon),Y(t+\epsilon))\right]\right)^2-\overline{g}^2(t,x,y)\right\}\\
&+\alpha\mathbf{E}^{\underline{\phi}^{\tilde{u}}}_{t,x,y}\left[\int_t^{t+\epsilon} h_{\mathbf{b}}\left(\underline{\Phi}^{\tilde{u}}(s)\right)\dif s\right]-\hat{\alpha}\mathbf{E}^{\overline{\phi}^{\tilde{u}}}_{t,x,y}\left[\int_t^{t+\epsilon} h_{\mathbf{b}}\left(\overline{\Phi}^{\tilde{u}}(s)\right)\dif s\right].
\end{split}
\end{equation*}
Next we show that $H_\epsilon\leq o(\epsilon)$. To simplify the notation, we define the operator
\begin{equation}\nonumber
\mathcal{A}^{u,\phi}_\epsilon\psi(t,x,y)\triangleq \mathbf{E}^{\phi}_{t,x,y}\left[\psi(t+\epsilon,X^{u}(t+\epsilon),Y(t+\epsilon))\right]-\psi(t,x,y),
\end{equation}
where $u\in\Pi$, $\phi\in\Theta$, $\epsilon>0$  is a small constant, and $\psi\in C^{1,2,2}([0,T]\times\mathbb{R}\times\mathbb{R})$. Then
\begin{equation}\label{opr}
\lim\limits_{\epsilon\downarrow 0}\frac{1}{\epsilon}\mathcal{A}^{\tilde{u},\underline{\phi}^{u_\epsilon}}_\epsilon\!\!\psi(t,x,y)\!=\!\mathcal{A}^{\tilde{u},\underline{\phi}^{\tilde{u}}}\psi(t,x,y),~\text{and}~\lim\limits_{\epsilon\downarrow 0}\frac{1}{\epsilon}\mathcal{A}^{\tilde{u},\overline{\phi}^{u_\epsilon}}_\epsilon\!\!\psi(t,x,y)\!=\!\mathcal{A}^{\tilde{u},\overline{\phi}^{\tilde{u}}}\psi(t,x,y).
\end{equation}
And we rewrite $H_\epsilon$ as
\begin{equation*}
\begin{split}
H_\epsilon=&\alpha\mathcal{A}^{\tilde{u},\underline{\phi}^{\tilde{u}}}_\epsilon \underline{J}^{u^*,\underline{\phi}^*}(t,x,y)-\frac{\alpha\gamma}{2}\mathcal{A}^{\tilde{u},\underline{\phi}^{\tilde{u}}}_\epsilon \underline{g}^2(t,x,y)\\
&+\hat{\alpha}\mathcal{A}^{\tilde{u},\overline{\phi}^{\tilde{u}}}_\epsilon \overline{J}^{u^*,\overline{\phi}^*}(t,x,y)-\frac{\hat{\alpha}\gamma}{2}\mathcal{A}^{\tilde{u},\overline{\phi}^{\tilde{u}}}_\epsilon \overline{g}^2(t,x,y)\\
&+\frac{\alpha\gamma}{2}\left\{\left(\mathbf{E}^{\underline{\phi}^{\tilde{u}}}_{t,x,y}\left[\underline{g}(t+\epsilon,X^{\tilde{u}}(t+\epsilon),Y(t+\epsilon))\right]\right)^2-\underline{g}^2(t,x,y)\right\}\\
&+\frac{\hat{\alpha}\gamma}{2}\left\{\left(\mathbf{E}^{\underline{\phi}^{\tilde{u}}}_{t,x,y}\left[\overline{g}(t+\epsilon,X^{\tilde{u}}(t+\epsilon),Y(t+\epsilon))\right]\right)^2-\overline{g}^2(t,x,y)\right\}\\
&+\alpha\mathbf{E}^{\underline{\phi}^{\tilde{u}}}_{t,x,y}\left[\int_t^{t+\epsilon} h_{\mathbf{b}}\left(\underline{\Phi}^{\tilde{u}}(s)\right)\dif s\right]-\hat{\alpha}\mathbf{E}^{\overline{\phi}^{\tilde{u}}}_{t,x,y}\left[\int_t^{t+\epsilon} h_{\mathbf{b}}\left(\overline{\Phi}^{\tilde{u}}(s)\right)\dif s\right].
\end{split}
\end{equation*}
Using Dynkin's formula, we have
\begin{equation*}
\begin{split}
&\mathbf{E}^{\underline{\phi}^{\tilde{u}}}_{t,x,y}\left[\underline{g}(t+\epsilon,X^{\tilde{u}}(t+\epsilon),Y(t+\epsilon))\right]\\
=&\underline{g}(t,x,y)+\mathbf{E}^{\underline{\phi}^{\tilde{u}}}_{t,x,y}\left[\int_t^{t+\epsilon}\mathcal{A}^{\tilde{u},\underline{\phi}^{\tilde{u}}}\underline{g}(s,X^{\tilde{u}}(s),Y(s))\dif s\right],
\end{split}
\end{equation*}
which means that
\begin{equation*}
\begin{split}
&\left[\mathbf{E}^{\underline{\phi}^{\tilde{u}}}_{t,x,y}\left[\underline{g}(t+\epsilon,X^{\tilde{u}}(t+\epsilon),Y(t+\epsilon))\right]\right]^2-\underline{g}^2(t,x,y)\\
=&2\underline{g}(t,x,y)\mathbf{E}^{\underline{\phi}^{\tilde{u}}}_{t,x,y}\left[\int_t^{t+\epsilon}\mathcal{A}^{\tilde{u},\underline{\phi}^{\tilde{u}}}\underline{g}(s,X^{\tilde{u}}(s),Y(s))\dif s\right]+o(\epsilon),
\end{split}
\end{equation*}
In the similar way, {{we have}}
\begin{equation*}
\begin{split}
&\left[\mathbf{E}^{\overline{\phi}^{\tilde{u}}}_{t,x,y}\left[\overline{g}(t+\epsilon,X^{\tilde{u}}(t+\epsilon),Y(t+\epsilon))\right]\right]^2-\overline{g}^2(t,x,y)\\
=&2\overline{g}(t,x,y)\mathbf{E}^{\overline{\phi}^{\tilde{u}}}_{t,x,y}\left[\int_t^{t+\epsilon}\mathcal{A}^{\tilde{u},\overline{\phi}^{\tilde{u}}}\overline{g}(s,X^{\tilde{u}}(s),Y(s))\dif s\right]+o(\epsilon),
\end{split}
\end{equation*}
As such
\begin{equation*}
\begin{split}
H_\epsilon-o(\epsilon)=&\alpha\mathcal{A}^{\tilde{u},\underline{\phi}^{\tilde{u}}}_\epsilon \underline{J}^{u^*,\underline{\phi}^*}(t,x,y)-\frac{\alpha\gamma}{2}\mathcal{A}^{\tilde{u},\underline{\phi}^{\tilde{u}}}_\epsilon \underline{g}^2(t,x,y)\\
&+\hat{\alpha}\mathcal{A}^{\tilde{u},\overline{\phi}^{\tilde{u}}}_\epsilon \overline{J}^{u^*,\overline{\phi}^*}(t,x,y)-\frac{\hat{\alpha}\gamma}{2}\mathcal{A}^{\tilde{u},\overline{\phi}^{\tilde{u}}}_\epsilon \overline{g}^2(t,x,y)\\
&+\alpha\gamma\underline{g}(t,x,y)\mathbf{E}^{\underline{\phi}^{\tilde{u}}}_{t,x,y}\left[\int_t^{t+\epsilon}\mathcal{A}^{\tilde{u},\underline{\phi}^{\tilde{u}}}\underline{g}(s,X^{\tilde{u}}(s),Y(s))\dif s\right]\\
&+\hat{\alpha}\gamma\underline{g}(t,x,y)\mathbf{E}^{\underline{\phi}^{\tilde{u}}}_{t,x,y}\left[\int_t^{t+\epsilon}\mathcal{A}^{\tilde{u},\underline{\phi}^{\tilde{u}}}\underline{g}(s,X^{\tilde{u}}(s),Y(s))\dif s\right]\\
&+\alpha\mathbf{E}^{\underline{\phi}^{\tilde{u}}}_{t,x,y}\left[\int_t^{t+\epsilon} h_{\mathbf{b}}\left(\underline{\Phi}^{\tilde{u}}(s)\right)\dif s\right]-\hat{\alpha}\mathbf{E}^{\overline{\phi}^{\tilde{u}}}_{t,x,y}\left[\int_t^{t+\epsilon} h_{\mathbf{b}}\left(\overline{\Phi}^{\tilde{u}}(s)\right)\dif s\right].
\end{split}
\end{equation*}
Eq.~(\ref{vj}) implies that
\begin{equation*}
\begin{split}
&\alpha\mathcal{A}^{\tilde{u},\underline{\phi}^{\tilde{u}}}V(t,x,y)+\hat{\alpha}\mathcal{A}^{\tilde{u},\overline{\phi}^{\tilde{u}}}V(t,x,y)\\
=&\alpha\mathcal{A}^{\tilde{u},\underline{\phi}^{\tilde{u}}}\left[\alpha\underline{J}^{u^*,\underline{\phi}^*}(t,x,y)+\hat{\alpha}\overline{J}^{u^*,\overline{\phi}^*}(t,x,y)\right]+\hat{\alpha}\mathcal{A}^{\tilde{u},\overline{\phi}^{\tilde{u}}}\left[\alpha\underline{J}^{u^*,\underline{\phi}^*}(t,x,y)+\hat{\alpha}\overline{J}^{u^*,\overline{\phi}^*}(t,x,y)\right]\\
=&\alpha\hat{\alpha}(\mathcal{A}^{\tilde{u},\overline{\phi}^{\tilde{u}}}\!\!\!\!-\!\mathcal{A}^{\tilde{u},\underline{\phi}^{\tilde{u}}})\!\!\left[\overline{J}^{u^*\!,\overline{\phi}^*}\!\!\!(t,x,y)\!-\!\underline{J}^{u^*\!,\underline{\phi}^*}\!\!(t,x,y)\right]\!\!+\!\alpha\mathcal{A}^{\tilde{u},\underline{\phi}^{\tilde{u}}}_\epsilon\!\! \underline{J}^{u^*\!,\underline{\phi}^*}\!(t,x,y)\!+\!\hat{\alpha}\mathcal{A}^{\tilde{u},\overline{\phi}^{\tilde{u}}}_\epsilon \overline{J}^{u^*\!,\overline{\phi}^*}\!\!(t,x,y).
\end{split}
\end{equation*}
As $\overline{J}^{u^*,\overline{\phi}^*}(t,x,y)-\underline{J}^{u^*,\underline{\phi}^*}(t,x,y)$ is independent of $x$, we have
\begin{equation*}
\alpha\mathcal{A}^{\tilde{u},\underline{\phi}^{\tilde{u}}}V(t,x,y)+\hat{\alpha}\mathcal{A}^{\tilde{u},\overline{\phi}^{\tilde{u}}}V(t,x,y)=\alpha\mathcal{A}^{\tilde{u},\underline{\phi}^{\tilde{u}}} \underline{J}^{u^*,\underline{\phi}^*}(t,x,y)+\hat{\alpha}\mathcal{A}^{\tilde{u},\overline{\phi}^{\tilde{u}}} \overline{J}^{u^*,\overline{\phi}^*}(t,x,y).
\end{equation*}
Based on the extended HJB equation (\ref{veri}), we obtain
\begin{equation*}
\begin{split}
0\geq\alpha&\left[\mathcal{A}^{\tilde{u},\underline{\phi}^{\tilde{u}}}_\epsilon \underline{J}^{u^*,\underline{\phi}^*}(t,x,y)-\frac{\gamma}{2}\mathcal{A}^{\tilde{u},\underline{\phi}^{\tilde{u}}}\underline{g}^2(t,x,y)\right.\left.+\gamma \underline{g}(t,x,y)\mathcal{A}^{\tilde{u},\underline{\phi}^{\tilde{u}}}\underline{g}(t,x,y)+h_{\mathbf{b}}\left(\underline{\phi}^{\tilde{u}}(t,y)\right)\right]\\
+\hat{\alpha}&\left[\mathcal{A}^{\tilde{u},\overline{\phi}^{\tilde{u}}}_\epsilon \overline{J}^{u^*,\overline{\phi}^*}(t,x,y)-\frac{\gamma}{2}\mathcal{A}^{\tilde{u},\overline{\phi}^{\tilde{u}}}\overline{g}^2(t,x,y)\right.\left.+\gamma \overline{g}(t,x,y)\mathcal{A}^{\tilde{u},\overline{\phi}^{\tilde{u}}}\overline{g}(t,x,y)-h_{\mathbf{b}}\left(\overline{\phi}^{\tilde{u}}(t,y)\right)\right].
\end{split}
\end{equation*}
Using Eq.~(\ref{opr}), we obtain that $H_\epsilon\leq o(\epsilon)$, which means that $u^*$ is an equilibrium strategy.
\end{proof}

\section{Proof of Theorem \ref{insurer}}\label{A.2}
\begin{proof}
We first derive $V$, $\underline{g}$, $\overline{g}$, $\underline{\phi}^*$ and  $\overline{\phi}^*$, which satisfy Conditions (1) and (2) in Theorem \ref{THM4.1}. Then we check that Conditions
(3) and (4) in Theorem \ref{THM4.1} also hold.

We omit the tedious calculations and rewrite the extended HJB equation(\ref{veri}) as
\begin{equation}\label{hjbs}
\begin{split}
0=&\sup\limits_{u\in\Pi}\left\{V_t+\left[rx+\xi\pi_I y+\int_0^\infty\left[(\theta-\eta)z+(1+\eta)a)\right]\nu(\dif z)\right]V_x\right.\\
&+\frac{1}{2}\pi_I^2 y(V_{xx}-\alpha\gamma\underline{g}_x^2-\hat{\alpha}\gamma\overline{g}^2_x)+\frac{1}{2}\sigma^2y(V_{yy}-\alpha\gamma\underline{g}_y^2-\hat{\alpha}\gamma\overline{g}^2_y)\\
&+\sigma\pi_I y(V_{xy}-\alpha\gamma\underline{g}_x\underline{g}_y-\hat{\alpha}\gamma\underline{g}_x\underline{g}_y)+\kappa(\delta-y)V_y\\
&\left.+\alpha\inf\limits_{\phi\in\Theta}\left\{L^{u,\phi}(V,\underline{g})+h_{\mathbf{b}}(\phi)\right\}+\hat{\alpha}\sup\limits_{\phi\in\Theta}\left\{L^{u,\phi}(V,\overline{g})-h_{\mathbf{b}}(\phi)\right\}\right\},
\end{split}
\end{equation}
where we denote
\begin{equation*}
\begin{split}
L^{u,\phi}(V,g)\triangleq& \int_0^\infty\left(V(t,x-a(t,z),y)-V(t,x,y)\right.\\
&\left.-\frac{\gamma}{2}(g(t,x-a(t,z),y)-g(t,x,y))^2\right)(1-\phi(t,z))\nu(\dif z)\\
&-\pi_I\sqrt{y}\phi_0V_x-\sigma\rho_0\sqrt{y}\phi_0V_y-\sigma\rho\sqrt{y}\phi_YV_y.
\end{split}
\end{equation*}
By the first-order condition on Eq.~(\ref{hjbs}) with respect to $\phi$, the infimum and the supremum of $\phi$ in Eq.~(\ref{hjbs}) are achieved respectively at
\begin{equation}\label{phinp}
\begin{cases}
\underline{\phi}_0(t,x,y)=\sqrt{y}\beta_0\pi_I(t)V_x(t,x,y)+\sigma\sqrt{y}\beta_0\rho_0 V_y(t,x,y),\\
\underline{\phi}_Y(t,x,y)=\sigma\sqrt{y}\beta_Y\rho V_y(t,x,y),\\
\underline{\phi}_{{Z}}(t,x,z)=1-e^{\beta\left[\frac{\gamma}{2}\left(\underline{g}(t,x-a(t,z),y)-\underline{g}(t,x,y)\right)^2-\left(V(t,x-a(t,z),y)-V(t,x,y)\right)\right]},
\end{cases}
\end{equation}
and
\begin{equation}\label{phipp}
\begin{cases}
\overline{\phi}_0(t,x,y)=-\sqrt{y}\beta_0\pi_I(t)V_x(t,x,y)-\sigma\sqrt{y}\beta_0\rho_0 V_y(t,x,y),\\
\overline{\phi}_Y(t,x,y)=-\sigma\sqrt{y}\beta_Y\rho V_y(t,x,y),\\
\overline{\phi}_{{Z}}(t,x,z)=1-e^{-\beta\left[\frac{\gamma}{2}\left(\overline{g}(t,x-a(t,z),y)-\overline{g}(t,x,y)\right)^2-\left(V(t,x-a(t,z),y)-V(t,x,y)\right)\right]}.
\end{cases}
\end{equation}
We guess that
\begin{equation*}
\begin{cases}
V(t,x,y)=xe^{r(T-t)}+A(t)y+B(t),\\
\underline{g}(t,x)=xe^{r(T-t)}+\underline{H}(t)y+\underline{G}(t),\\
\overline{g}(t,x)=xe^{r(T-t)}+\overline{H}(t)y+\overline{G}(t),
\end{cases}
\end{equation*}
where $A(t)$, $B(t)$, $\underline{H}(t)$, $\underline{G}(t)$, $\overline{H}(t)$ and  $\overline{G}(t)$ are deterministic  functions of $t$. By the first and the third equalities of Eq.~(\ref{bondi}), the boundary conditions are given by
\begin{equation*}
A(T)=B(T)=\underline{H}(T)=\underline{G}(T)=\overline{H}(T)=\overline{G}(T)=0.
\end{equation*}
Then the HJB equation becomes
\begin{equation}\label{hjbc}
\begin{split}
0=B'(t)+&A'(t)y+(\theta-\eta(t))e^{r(T-t)}\int_0^\infty z\nu(\mathrm{d}z)\\
+\sup_{\pi_I}&\left\{\xi y\pi_I(t)e^{r(T-t)}-\frac{1}{2}[\gamma+(2\alpha-1)\beta_0]y\pi_I(t)^2e^{2r(T-t)}\right.\\
&-(2\alpha-1)\beta_0\sigma\rho_0 y\pi_I(t)A(t)e^{r(T-t)}-\frac{1}{2}\sigma^2y\left[\alpha\underline{H}(t)^2+\hat{\alpha}\overline{H}(t)^2\right]\\
&\left.+\kappa(\lambda-y)A(t)-\frac{1}{2}(2\alpha-1)\sigma^2[\beta_0\rho_0^2+\beta_Y(1-\rho_0^2)] yA(t)^2\right\}\\
+\sup_{a\leq z}&\left\{(1+\eta(t))\int_0^\infty a(t,z)e^{r(T-t)}\nu(\mathrm{d}z)\right.\\
&+\frac{\alpha}{\beta}\int_{0}^\infty\left(1-e^{\beta\left[a(t,z)e^{r(T-t)}+\frac{\gamma}{2}a(t,z)^2e^{2r(T-t)}\right]}\right)\nu(\mathrm{d}z)\\
&-\left.\frac{\hat{\alpha}}{\beta}\int_{0}^\infty\left(1-e^{-\beta\left[a(t,z)e^{r(T-t)}+\frac{\gamma}{2}a(t,z)^2e^{2r(T-t)}\right]}\right)\nu(\mathrm{d}z)\right\}.
\end{split}
\end{equation}
Using the first-order condition of Eq.~(\ref{hjbc}) with respect to $a$ and $\pi_I$, the supremum in Eq.~(\ref{hjbc}) is achieved at $u^*=(p^*,\pi_I^*)$ as follows
\begin{equation}\label{piip}
\pi_I^*(t)=\frac{\xi -(2\alpha-1)\beta_0\rho_0\sigma A(t)-\gamma\sigma\rho_0\left(\alpha\underline{H}(t)+\hat{\alpha}\overline{H}(t)\right)}{[\gamma+(2\alpha-1)\beta_0]e^{r(T-t)}},
\end{equation}
and
\begin{equation*}
a^*(t,z)=\arg\max\limits_{a(t,z)\leq z}\left\{\int_0^\infty f\left(a(t,z) e^{r(T-t)}\right)\nu(\mathrm{d}z)\right\},
\end{equation*}
where $f(x)=(1+\eta)x-\frac{1}{\beta}\left[\alpha e^{\beta\left(x+\frac{\gamma}{2}x^2\right)}-\hat{\alpha}e^{-\beta\left(x+\frac{\gamma}{2}x^2\right)}\right]$. Because $f_1(x)=\alpha e^{x}-\hat{\alpha}e^{-x}~(x\geq 0,~\alpha>\hat{\alpha})$,  $f_2(x)=x+\frac{\gamma}{2}x^2~(\gamma> 0)$ are both strictly convex, $f(x)=(1+\eta)x-\frac{1}{\beta}f_1(f_2(x))$ is strictly concave. Besides, $f'(0)=\eta>0$, and $f'(+\infty)=-\infty$. As such, $f$ attains the maximum value at a unique point $a_0$ in $[0,+\infty)$, and increases in $[0,a_0]$. Thus, the equilibrium reinsurance strategy is
$$a^*(t,z)=a_0(t)e^{-r(T-t)}\wedge z,$$
where $a_0$ is determined by the equation
\begin{equation}\label{ap}
0=(1+\eta)- \left(1+\gamma a_0\right)\left[\alpha e^{\beta\left(a_0+\frac{\gamma}{2}a_0^2\right)}+\hat{\alpha}e^{-\beta\left(a_0+\frac{\gamma}{2}a_0^2\right)}\right],
\end{equation}
Substituting Eqs. (\ref{piip}) and (\ref{ap}) into Eq.~(\ref{hjbc}), we have
\begin{equation*}
\begin{split}
A'(t)=&\kappa A(t)-\frac{1}{2}(2\alpha-1)\sigma^2[\beta_0\rho_0^2+\beta_Y(1-\rho_0^2)]A(t)^2+\frac{1}{2}\sigma^2\left[\alpha\underline{H}(t)^2+\hat{\alpha}\overline{H}(t)^2\right]\\
&-\frac{1}{2[\gamma+(2\alpha-1)\beta_0]}\left[\xi -(2\alpha-1)\beta_0\rho_0\sigma A(t)-\gamma\sigma\rho_0\left(\alpha\underline{H}(t)+\hat{\alpha}\overline{H}(t)\right)\right]^2,
\end{split}
\end{equation*}
and
\begin{equation*}
\begin{split}
B(t)=\int_t^T&\left\{\kappa\lambda A(s)+(\theta-\eta(s))e^{r(T-s)}\int_0^\infty z\nu(\mathrm{d}z)\right.\\
&+(1+\eta(s))\int_0^\infty a_0\wedge ze^{r(T-s)}\nu(\mathrm{d}z)\\
&+\frac{\alpha}{\beta}\int_{0}^\infty\left(1-e^{\beta\left[a_0\wedge ze^{r(T-s)}+\frac{\gamma}{2}(a_0\wedge ze^{r(T-s)})^2\right]}\right)\nu(\mathrm{d}z)\\
&\left.-\frac{\hat{\alpha}}{\beta}\int_{0}^\infty\left(1-e^{-\beta\left[a_0\wedge ze^{r(T-s)}+\frac{\gamma}{2}(a_0\wedge ze^{r(T-s)})^2\right]}\right)\nu(\mathrm{d}z)\right\}\dif s.
\end{split}
\end{equation*}
Similarly, substituting Eqs. (\ref{piip}) and (\ref{ap}) into the second equality of Eq.~(\ref{bondi}) yields
\begin{equation*}
\begin{split}
\underline{H}'(t)=&\kappa \underline{H}(t)-\xi\frac{\xi -(2\alpha-1)\beta_0\rho_0\sigma A(t)-\gamma\sigma\rho_0\left(\alpha\underline{H}(t)+\hat{\alpha}\overline{H}(t)\right)}{\gamma+(2\alpha-1)\beta_0}\\
&+\beta_0\left[\frac{\xi -(2\alpha-1)\beta_0\rho_0\sigma A(t)-\gamma\sigma\rho_0\left(\alpha\underline{H}(t)+\hat{\alpha}\overline{H}(t)\right)}{\gamma+(2\alpha-1)\beta_0}\right]^2\\
&+\sigma\beta_0\frac{\xi -(2\alpha-1)\beta_0\rho_0\sigma A(t)-\gamma\sigma\rho_0\left(\alpha\underline{H}(t)+\hat{\alpha}\overline{H}(t)\right)}{\gamma+(2\alpha-1)\beta_0}\left(A(t)+\underline{H}(t)\right)\\
&+\sigma^2(\beta_0\rho_0^2+\beta_Y\rho^2)A(t)\underline{H}(t),
\end{split}
\end{equation*}
and
\begin{equation*}
\begin{split}
\overline{H}'(t)=&\kappa \overline{H}(t)-\xi\frac{\xi -(2\alpha-1)\beta_0\rho_0\sigma A(t)-\gamma\sigma\rho_0\left(\alpha\underline{H}(t)+\hat{\alpha}\overline{H}(t)\right)}{\gamma+(2\alpha-1)\beta_0}\\
&-\beta_0\left[\frac{\xi -(2\alpha-1)\beta_0\rho_0\sigma A(t)-\gamma\sigma\rho_0\left(\alpha\underline{H}(t)+\hat{\alpha}\overline{H}(t)\right)}{\gamma+(2\alpha-1)\beta_0}\right]^2\\
&-\sigma\beta_0\frac{\xi -(2\alpha-1)\beta_0\rho_0\sigma A(t)-\gamma\sigma\rho_0\left(\alpha\underline{H}(t)+\hat{\alpha}\overline{H}(t)\right)}{\gamma+(2\alpha-1)\beta_0}\left(A(t)+\overline{H}(t)\right)\\
&-\sigma^2(\beta_0\rho_0^2+\beta_Y\rho^2)A(t)\overline{H}(t).
\end{split}
\end{equation*}
Next, we show the existence of solutions to   Eq.~(\ref{ABHH}). Similar to \cite{Yan2020}, denote $P(t)=diag\left(A(t),\underline{H}(t),\overline{H}(t)\right)$, \\
$N_{12}=\begin{bmatrix}
	0 & 1 & 0\\
	1 & 0 & 0\\
	0 & 0 & 1\\
\end{bmatrix}$,
$N_{13}=\begin{bmatrix}
	0 & 0 & 1\\
	0 & 1 & 0\\
	1 & 0 & 0\\
\end{bmatrix}$,
$N_{23}=\begin{bmatrix}
	1 & 0 & 0\\
	0 & 0 & 1\\
	0 & 1 & 0\\
\end{bmatrix}$,
$N=\begin{bmatrix}
	0 & 0 & 1\\
	1 & 0 & 0\\
	0 & 1 & 0\\
\end{bmatrix}$, \\
$K=diag\left(\frac{2\alpha-1}{2}\sigma^2\!\left[\!\frac{\rho_0^2\gamma\beta_0}{\gamma+(2\alpha-1)\beta_0}\!\!+\!\!\rho^2\beta_{Y}\right]\!\!, \frac{\alpha\sigma^2\rho_0\gamma\beta_0}{\gamma+(2\alpha-1)\beta_0}\!\!\left[\frac{\alpha\rho_0\gamma}{\gamma+(2\alpha-1)\beta_0}\!-\!\!1\right]\!\!, \frac{\hat{\alpha}\sigma^2\rho_0\gamma\beta_0}{\gamma+(2\alpha-1)\beta_0}\!\!\left[1\!\!-\!\frac{\hat{\alpha}\rho_0\gamma}{\gamma+(2\alpha-1)\beta_0}\right]\right)$,\\
$K_1=diag\left(\frac{\alpha}{2}\sigma^2\gamma\left[1-\frac{\alpha\rho_0\gamma}{\gamma+(2\alpha-1)\beta_0}\right], \frac{(2\alpha-1)\sigma^2\rho_0\beta_0^2}{\gamma+(2\alpha-1)\beta_0}\left[\frac{(2\alpha-1)\rho_0\beta_0}{\gamma+(2\alpha-1)\beta_0}-1\right], 0\right)$,\\
$K_2=diag\left(\frac{\hat{\alpha}}{2}\sigma^2\gamma\left[1-\frac{\hat{\alpha}\rho_0\gamma}{\gamma+(2\alpha-1)\beta_0}\right], 0, \frac{(2\alpha-1)\sigma^2\rho_0\beta_0^2}{\gamma+(2\alpha-1)\beta_0}\left[1-\frac{(2\alpha-1)\rho_0\beta_0}{\gamma+(2\alpha-1)\beta_0}\right]\right)$,\\
$K_3=diag\left(0, \frac{\hat{\alpha}^2\sigma^2\rho_0^2\gamma^2\beta_0}{\left[\gamma+(2\alpha-1)\beta_0\right]^2}, -\frac{\alpha^2\sigma^2\rho_0^2\gamma^2\beta_0}{\left[\gamma+(2\alpha-1)\beta_0\right]^2}\right)$,\\
$K_{12}=diag\left(-\frac{\alpha(2\alpha-1)\sigma^2\rho_0^2\gamma\beta_0}{\gamma+(2\alpha-1)\beta_0}, \frac{\sigma^2\rho_0\gamma\beta_0}{\gamma+(2\alpha-1)\beta_0}\left[\frac{2\alpha(2\alpha-1)\rho_0\beta_0}{\gamma+(2\alpha-1)\beta_0}+\hat{\alpha}\right]\!+\!\sigma^2\left[(\rho_0^2-\rho_0)\beta_0+\rho^2\beta_{Y}\right], 0\right)$,\\
$K_{13}=diag\left(-\frac{\hat{\alpha}(2\alpha-1)\sigma^2\rho_0^2\gamma\beta_0}{\gamma+(2\alpha-1)\beta_0}, 0,  \frac{\sigma^2\rho_0\gamma\beta_0}{\gamma+(2\alpha-1)\beta_0}\left[\frac{2\hat{\alpha}(2\alpha-1)\rho_0\beta_0}{\gamma+(2\alpha-1)\beta_0}+\alpha\right]\!+\!\sigma^2\left[(\rho_0^2-\rho_0)\beta_0+\rho^2\beta_{Y}\right]\right)$,\\
$K_{23}=diag\left(0, \frac{\hat{\alpha}\sigma^2\rho_0\gamma\beta_0}{\gamma+(2\alpha-1)\beta_0}\left[\frac{2\alpha\rho_0\gamma}{\gamma+(2\alpha-1)\beta_0}-1\right], \frac{\alpha\sigma^2\rho_0\gamma\beta_0}{\gamma+(2\alpha-1)\beta_0}\left[1-\frac{2\hat{\alpha}\rho_0\gamma}{\gamma+(2\alpha-1)\beta_0}\right]\right)$,\\
$K_0=diag\left(-\frac{\alpha\hat{\alpha}\sigma^2\rho_0^2\gamma^2}{\gamma+(2\alpha-1)\beta_0}, \frac{\hat{\alpha}\sigma^2\rho_0\gamma\beta_0}{\gamma+(2\alpha-1)\beta_0}\left[\frac{2(2\alpha-1)\rho_0\beta_0}{\gamma+(2\alpha-1)\beta_0}-1\right], \frac{\alpha\sigma^2\rho_0\gamma\beta_0}{\gamma+(2\alpha-1)\beta_0}\left[1-\frac{2\alpha\rho_0\gamma}{\gamma+(2\alpha-1)\beta_0}\right]\right)$,\\
$D=diag\left(\kappa\!+\!\frac{(2\alpha-1)\sigma\rho_0\beta_0}{\gamma+(2\alpha-1)\beta_0}\xi, \kappa\!+\!\frac{\sigma\left[\alpha\rho_0\left(1-\frac{2\beta_0}{\gamma+(2\alpha-1)\beta_0}\right)\gamma+\beta_0\!\right]}{\gamma+(2\alpha-1)\beta_0}\xi, \kappa\!+\!\frac{\sigma\left[\hat{\alpha}\rho_0\left(1+\frac{2\beta_0}{\gamma+(2\alpha-1)\beta_0}\right)\gamma-\beta_0\!\right]}{\gamma+(2\alpha-1)\beta_0}\xi\right)$,\\
$D_1=diag\left(\frac{\alpha\sigma\rho_0\gamma}{\gamma+(2\alpha-1)\beta_0}\xi, \frac{\left[(2\alpha-1)\rho_0\left(1-\frac{2\beta_0}{\gamma+(2\alpha-1)\beta_0}\right)+1\right]\sigma\beta_0}{\gamma+(2\alpha-1)\beta_0}\xi, 0\right)$,\\
$D_2=diag\left(\frac{\hat{\alpha}\sigma\rho_0\gamma}{\gamma+(2\alpha-1)\beta_0}\xi, 0, \frac{\left[(2\alpha-1)\rho_0\left(1+\frac{2\beta_0}{\gamma+(2\alpha-1)\beta_0}\right)-1\right]\sigma\beta_0}{\gamma+(2\alpha-1)\beta_0}\xi\right)$,\\
$D_3=diag\left(0, \frac{\hat{\alpha}\sigma\rho_0\gamma}{\gamma+(2\alpha-1)\beta_0}\left[1-\frac{2\beta_0}{\gamma+(2\alpha-1)\beta_0}\right]\xi, \frac{\alpha\sigma\rho_0\gamma}{\gamma+(2\alpha-1)\beta_0}\left[1+\frac{2\beta_0}{\gamma+(2\alpha-1)\beta_0}\right]\xi\right)$,\\
$Q=diag\left(-\frac{1}{2[\gamma+(2\alpha-1)\beta_0]}\xi^2, -\frac{\gamma-2\hat{\alpha}\beta_0}{[\gamma+(2\alpha-1)\beta_0]^2}\xi^2, -\frac{\gamma+2\alpha\beta_0}{[\gamma+(2\alpha-1)\beta_0]^2}\xi^2\right)$.

Then the system (\ref{ABHH}) becomes
\begin{equation}\label{MR}
\begin{split}
	\frac{\mathrm{d}P(t)}{\mathrm{d}t}&=N_{12}PN_{12}K_1N_{12}PN_{12}+N_{13}PN_{13}K_2N_{13}PN_{13}+N_{23}PN_{23}K_3N_{23}PN_{23}\\
	&+PKP+NPKNPN+N_{12}PK_{12}N_{12}P+N_{13}PK_{13}N_{13}P+N_{23}PK_{23}N_{23}P\\
	&+PD+N_{12}PN_{12}D_1+N_{13}PN_{13}D_2+N_{23}PN_{23}D_3+Q,
\end{split}
\end{equation}
which is a matrix Riccati differential equation. Define
\begin{equation*}
\begin{split}
d=&||D||_{\sup}\!+||D_1||_{\sup}\!+||D_2||_{\sup}\!+||D_3||_{\sup},\\
k=&||K||_{\sup}\!+||K_0||_{\sup}\!+||K_1||_{\sup}\!+||K_2||_{\sup}\!+||K_3||_{\sup}\!+||K_{12}||_{\sup}\!+||K_{13}||_{\sup}\!+||K_{23}||_{\sup},\\
q=&||Q||_{\sup},
\end{split}
\end{equation*}
and $\Delta=d^2-4kq$, $\zeta_1=\frac{-d+\sqrt{\Delta}}{2k}$, $\zeta_2=\frac{-d-\sqrt{\Delta}}{2k}$. If $T$ satisfies the following condition:
\begin{itemize}
	\item If $\Delta=0$, then $T<\frac{2}{d}$;
	\item If $\Delta>0$, then $T<\frac{1}{\sqrt{\Delta}}\ln\left(\frac{\zeta_2}{\zeta_1}\right)$;
	\item If $\Delta<0$, then $T<\frac{1}{\sqrt{|\Delta|}}\left[\pi+2\arctan\left(\frac{Re(\zeta_1)}{Im(\zeta_1)}\right)\right]$.
\end{itemize}
Based on Lemma 3 in \cite{Yan2020},  the existence of the solution to the matrix Riccati differential equation \ref{MR} is ensured.

And the expressions of $\underline{\phi}^*$ and $\overline{\phi}^*$ given in Theorem \ref{insurer} are directly proved by substituting Eqs. (\ref{piip}) and (\ref{ap}) into Eqs. (\ref{phinp}) and (\ref{phipp}).

Furthermore, based on the conditions in Eq.~(\ref{bondi}), Conditions (3) and (4) of  Theorem \ref{THM4.1} can be easily verified.
\end{proof}
\section{Proof of Proposition \ref{exists}}\label{existproof}
\begin{proof}
Let
\begin{equation*}
\begin{split}
M(a_0)=&\int_{a_0e^{-r(T-t)}}^\infty\left\{\left[ze^{r(T-t)}-a_0-(1+\eta) \frac{\partial a_0}{\partial \eta}\right]\right.\\
&+\alpha_R\frac{\partial a_0}{\partial \eta}\left[1+\gamma_R\left(ze^{r(T-t)}-a_0\right)\right]e^{\beta_r\left[ze^{r(T-t)}-a_0+\frac{\gamma_R}{2}(ze^{r(T-t)}-a_0)^2\right]}\\
&\left.+\hat{\alpha}_R\frac{\partial a_0}{\partial \eta}\left[1+\gamma_R\left(ze^{r(T-t)}-a_0\right)\right]e^{-\beta_r\left[ze^{r(T-t)}-a_0+\frac{\gamma_R}{2}(ze^{r(T-t)}-a_0)^2\right]}\right\}\nu(\mathrm{d}z).
\end{split}
\end{equation*}
Based on Eq.~(\ref{a}), we have

\begin{equation*}
\begin{split}
1=\frac{\partial a_0}{\partial \eta}&\left\{\gamma \left[\alpha e^{\beta\left(a_0+\frac{\gamma}{2}a_0^2\right)}+\hat{\alpha}e^{-\beta\left(a_0+\frac{\gamma}{2}a_0^2\right)}\right]\right.\\
&+\left.\beta\left(1+\gamma a_0\right)^2\left[\alpha e^{\beta\left(a_0+\frac{\gamma}{2}a_0^2\right)}-\hat{\alpha}e^{-\beta\left(a_0+\frac{\gamma}{2}a_0^2\right)}\right]\right\}.
\end{split}
\end{equation*}
Then we know
\begin{equation*}
\begin{split}
M(0)=&\int_{0}^\infty\left\{\left[ze^{r(T-t)}-(1+\gamma_R) \frac{\partial a_0}{\partial \eta}\right]+\alpha_R\frac{\partial a_0}{\partial \eta}\left[1+\gamma_R ze^{r(T-t)}\right]e^{\beta_r\left[ze^{r(T-t)}+\frac{\gamma_R}{2}z^2e^{2r(T-t)}\right]}\right.\\
&\left.+\hat{\alpha}_R\frac{\partial a_0}{\partial \eta}\left[1+\gamma_R ze^{r(T-t)}\right]e^{-\beta_r\left[ze^{r(T-t)}+\frac{\gamma_R}{2}z^2e^{r(T-t)}\right]}\right\}\nu(\mathrm{d}z)\\
=&\int_{0}^\infty\left\{\left[ze^{r(T-t)}-\frac{1+\gamma_R}{\beta(2\alpha-1)+\gamma} \right]+\alpha
\frac{1+\gamma_R ze^{r(T-t)}}{\beta(2\alpha-1)+\gamma}e^{\beta_r\left[ze^{r(T-t)}+\frac{\gamma_R}{2}z^2e^{2r(T-t)}\right]}\right.\\
&\left.+\hat{\alpha}_R\frac{1+\gamma_R ze^{r(T-t)}}{\beta(2\alpha-1)+\gamma}e^{-\beta_r\left[ze^{r(T-t)}+\frac{\gamma_R}{2}z^2e^{r(T-t)}\right]}\right\}\nu(\mathrm{d}z)\\
\geq &\int_{0}^\infty ze^{r(T-t)}\nu(\mathrm{d}z)+\frac{\gamma_R}{\beta(2\alpha-1)+\gamma}\left[e^{r(T-t)}\int_{0}^\infty z\nu(\mathrm{d}z)-\nu(0,\infty)\right]\\
\geq &0,
\end{split}
\end{equation*}
{where the last inequatity  holds true by adjusting the unit of the claim $z$.}

{Moreover}, we have
\begin{equation*}
\begin{split}
&\lim\limits_{a_0\rightarrow +\infty}\frac{\left[1+\gamma_R\left(ze^{r(T-t)}-a_0\right)\right]e^{\beta_r\left[ze^{r(T-t)}-a_0+\frac{\gamma_R}{2}(ze^{r(T-t)}-a_0)^2\right]}}{1+\eta}=0,\\
&\lim\limits_{a_0\rightarrow +\infty}\frac{\left[1+\gamma_R\left(ze^{r(T-t)}-a_0\right)\right]e^{-\beta_r\left[ze^{r(T-t)}-a_0+\frac{\gamma_R}{2}(ze^{r(T-t)}-a_0)^2\right]}}{1+\eta}=0.
\end{split}
\end{equation*}
Based on Condition (2) in Assumption \ref{ass1}, we have $M(a_0)<0$ when $a_0$ is large enough.

As such, $M(0)>0$, $M(+\infty)<0$, and $M(\cdot)$ is continuous, which {proves} that Eq.~(\ref{eta}) has a solution.
\end{proof}
\vskip 10pt
\setcounter{equation}{0}
\vskip 20pt
\bibliographystyle{elsarticle-harv}
\bibliography{a}
\end{document}